\documentclass[12pt]{article}
\usepackage{apacite}
\usepackage{natbib}
\usepackage[colorlinks, citecolor=blue]{hyperref}
\usepackage{amsmath}
\usepackage{amsthm}
\usepackage{multirow}
\usepackage{booktabs}
\usepackage{amsfonts}
\usepackage[justification=centering]{caption}
\usepackage{amssymb}
\usepackage{threeparttable}
\usepackage{graphicx}
\usepackage{float}
\usepackage{setspace}
\usepackage{comment}
\usepackage{subcaption}
\usepackage{cleveref}
\usepackage{soul}
\usepackage{scalerel}
\usepackage{pifont}
\usepackage{xr}
\usepackage{tikz}
\usepackage{algorithm}
\usepackage{algpseudocode}

\newcommand{\blind}{1}
\newtheorem{assumption}{Assumption}
\newtheorem{theorem}{Theorem}

\newtheorem{definition}{Definition}
\newtheorem{lemma}{Lemma}
\newtheorem{corollary}{Corollary}
\newtheorem*{remark}{Remark}

\newcommand{\R}[1][q]{r_{#1}}
\newcommand{\aop}{a^*}
\newcommand{\bop}{b^*}

\newcommand{\tauid}[1][q]{\tau_i^d(#1)}
\newcommand{\taud}[1][q]{\tau^d(#1)}
\newcommand{\tauidhat}[1][q]{\hat \tau_i^d(#1)}
\newcommand{\taudhat}[1][q]{\hat \tau^d(#1)}
\newcommand{\tauis}[1][z]{\tau_i^s(#1)}
\newcommand{\taus}[1][z]{\tau^s(#1)}
\newcommand{\tauishat}[1][z]{\hat \tau_i^s(#1)}
\newcommand{\taushat}[1][z]{\hat \tau^s(#1)}

\newcommand{\yit}[3]{Y_{i,t}({#1} \mathbf 1,{#2}\mathbf {#3})}
\newcommand{\yitp}[3]{Y_{i,t'}({#1} \mathbf 1,{#2}\mathbf {#3})}
\newcommand{\yjt}[3]{Y_{j,t}({#1} \mathbf 1,{#2}\mathbf {#3})}
\newcommand{\yjtp}[3]{Y_{j,t'}({#1} \mathbf 1,{#2}\mathbf {#3})}
\newcommand{\Iit}[3]{I_{i,t}({#1} \mathbf 1,{#2}\mathbf {#3})}
\newcommand{\Pit}[3]{\mathrm{pr}_{i,t}({#1} \mathbf 1,{#2} \mathbf {#3})}
\newcommand{\yitmis}[4]{Y_{i,t}(#1 \mathbf 1_{#4+1}, {#2}\mathbf {#3}_{#4+1})}
\newcommand{\Iitmis}[4]{I_{i,t}(#1 \mathbf 1_{#4+1}, {#2}\mathbf {#3}_{#4+1})}
\newcommand{\Pitmis}[4]{\mathrm{pr}_{i,t}(#1 \mathbf 1_{#4+1},{#2}\mathbf {#3}_{#4+1})}

\newcommand{\yik}[3]{\tilde Y_{i,k}({#1} \mathbf 1, {#2} \mathbf {#3})}

\newcommand{\yk}[3]{\tilde {\mathbf Y}_{k}({#1} \mathbf 1,{#2} \mathbf {#3})}

\newcommand{\yikc}[3]{\check Y_{i,k}({#1} \mathbf 1, {#2} \mathbf {#3})}
\newcommand{\yikcplus}[3]{\check Y_{i,k+1}({#1} \mathbf 1, {#2} \mathbf {#3})}
\newcommand{\yjkc}[3]{\check Y_{j,k}({#1} \mathbf 1, {#2} \mathbf {#3})}

\newcommand{\ykc}[3]{\check {\mathbf Y}_{k}({#1} \mathbf 1,{#2} \mathbf {#3})}
\newcommand{\ykcplus}[3]{\check {\mathbf Y}_{k+1}({#1} \mathbf 1,{#2} \mathbf {#3})}
\newcommand{\ykd}[3]{\Delta {\mathbf Y}_{k}({#1} \mathbf 1,{#2} \mathbf {#3})}

\newcommand{\Iik}[3]{\tilde{I}_{i,k}({#1} \mathbf 1, {#2} \mathbf {#3})}

\newcommand{\Iikc}[3]{\check{I}_{i,k}({#1} \mathbf 1, {#2} \mathbf {#3})}

\newcommand{\Ijkc}[3]{\check{I}_{j,k}({#1} \mathbf 1, {#2} \mathbf {#3})}

\newcommand{\Ik}[3]{\tilde {\mathbf I}_{k}({#1} \mathbf 1,{#2} \mathbf {#3})}

\newcommand{\Ikc}[3]{\check {\mathbf I}_{k}({#1} \mathbf 1,{#2} \mathbf {#3})}
\newcommand{\Ikcplus}[3]{\check {\mathbf I}_{k+1}({#1} \mathbf 1,{#2} \mathbf {#3})}

\newcommand{\ykco}{\check {\mathbf Y}_{k}}
\newcommand{\ykcoplus}{\check {\mathbf Y}_{k+1}}
\newcommand{\ykdo}{\Delta {\mathbf Y}_{k}}

\newcommand{\taudorder}[1][q]{\tau^d_{[#1]}(q)}
\newcommand{\taudhatorder}[1][q]{\hat \tau^d_{[#1]}(q)}
\newcommand{\tausorder}[1][q]{\tau^s_{[#1]}(z)}
\newcommand{\taushatorder}[1][q]{\hat \tau^s_{[#1]}(z)}

\DeclareMathOperator*{\argmin}{arg\,min}

\begin{document}

\def\spacingset#1{\renewcommand{\baselinestretch}%
{#1}\small\normalsize} \spacingset{1}


\if1\blind
{
  \title{\bf Minimax Optimal Design with Spillover and Carryover Effects}
  \author{Haoyang Yu\\
			Department of Statistics and Data Science, Tsinghua University, \\
            Beijing, 100084, China\vspace{.5cm}\\
   Wei Ma\\
   Institute of Statistics and Big Data,  Renmin University of China, \\
   Beijing, 100872, China\vspace{.5cm}\\
	Hanzhong Liu\thanks{
    Corresponding author: lhz2016@tsinghua.edu.cn. Dr. Liu was supported by \textit{the National Natural Science Foundation of China (12071242).}}\\
			Department of Statistics and Data Science, Tsinghua University, \\ Beijing, 100084, China\vspace{.5cm}\\
			}
    \date{}
  \maketitle
}\fi

\if0\blind
{
	\title{\bf Minimax Design with Spillover and Carryover Effects}
	\date{}
	\maketitle
}\fi


\begin{abstract}
In various applications, the potential outcome of a unit may be influenced by the treatments received by other units, a phenomenon known as interference, as well as by prior treatments, referred to as carryover effects. These phenomena violate the stable unit treatment value assumption and pose significant challenges in causal inference. To address these complexities, we propose a minimax optimal experimental design that simultaneously accounts for both spillover and carryover effects, enhancing the precision of estimates for direct and spillover effects. This method is particularly applicable to multi-unit experiments, reducing sample size requirements and experimental costs. We also investigate the asymptotic properties of the Horvitz--Thompson estimators of direct and spillover effects, demonstrating their consistency and asymptotic normality under the minimax optimal design. To facilitate valid inferences, we propose conservative variance estimators. Furthermore, we tackle the challenges associated with potential misspecifications in the order of carryover effects. Our approach is validated by comprehensive numerical studies that demonstrate superior performance compared to existing experimental designs.
\end{abstract}

\noindent%
{\it Keywords:} carryover effects; causal inference; design-based inference; minimax optimal design; spillover effects.
\vfill

\newpage
\spacingset{1.9} 

\section{Introduction}

Causal inference is a fundamental aspect of statistics and data science, providing essential tools for understanding intervention effects across various fields. The potential outcomes framework, introduced by \citet{Neyman1923} and popularized by \citet{rubin1974}, is central to this methodology. A key assumption within this framework is the Stable Unit Treatment Value Assumption (SUTVA) \citep{rubin1980}, which posits that a unit's potential outcomes are unaffected by the treatment assignments of others, thus ruling out interference. While SUTVA simplifies causal effect estimation, it often overlooks the complexities of real-world interactions.
In practice, interactions among units are both common and significant. For instance, in public health, vaccination protects not just the unit but also enhances herd immunity, affecting the health outcomes of the members of the same family or community. Similarly, in economics, the policy in one region may influence the behaviors of other areas. 


To address scenarios where SUTVA does not hold, researchers have developed adaptations to causal inference methods that recognize and account for interactions. Specific assumptions regarding interference, such as partial interference and stratified interference, are often employed to manage interdependencies and structure interactions within groups or clusters \citep{hudgens2008, tchetgen2012, liu2014}. Partial interference posits that a unit's potential outcomes are influenced only by treatment assignments within its own cluster, eliminating interference between different clusters. In addition, stratified interference indicates that a unit's potential outcomes depend solely on its own treatment and the proportion of treated units within the same cluster. These assumptions are essential for identifying both direct and spillover effects.
\citet{hudgens2008} defined direct effects as the impact of an individual's treatment, while spillover effects refer to the influence of others' treatments, proposing a two-stage design for estimating both. \citet{tchetgen2012} expanded these methods to observational studies and introduced a valid variance estimator. Additionally, \citet{liu2014} investigated the asymptotic properties of causal effect estimators and the construction of confidence intervals. 
Building on these concepts, \citet{aronow2017} utilized ``exposure mapping'' to analyze various forms of interference, which has been applied to complex systems, including network interference \citep{leung2020} and scenarios with unknown interference \citep{savje2021}.

On the other hand, causal inference research increasingly emphasizes the importance of temporal dependencies, referred to as carryover effects, where potential outcomes are influenced by both current and past treatments \citep[see, e.g.,][]{boruvka2018,bojinov2019,bojinov2021,viviano2023}. For instance, a patient's response to treatment can significantly depend on their treatment history, complicating the estimation of causal effects.
Many studies simplify their analysis by assuming ``no carryover effects'', which can neglect real-world complexities and result in biased estimates. 


The presence of carryover effects complicates the process of randomization and experimental design, shifting the focus from merely determining how to randomize treatments to also considering when to randomize them in order to ensure accurate estimates of causal effects. This added complexity has spurred research into optimal design strategies, particularly minimax optimal designs. These minimax optimal designs seek to minimize the variance of effect estimators in worst-case scenarios \citep{bojinov2023, ni2023}, thereby improving the reliability and robustness of experimental results. Moreover, from an optimization perspective, minimax optimal designs can be interpreted as solutions to a class of robust optimization problems, as highlighted by \citet{zhao2024}.
\citet{bojinov2023} advanced the field by developing a minimax optimal design that accommodates carryover effects. Their approach assumes uniform treatment assignment across all units at each time point, mirroring a single-unit experiment. 

However, the minimax optimal design that facilitates the identification of spillover and carryover effects has yet to be thoroughly explored, which is crucial for policy evaluation across various domains. For example, in public health, particularly within vaccination programs, it is essential to understand the dynamics of herd immunity, where protective effects transcend the vaccinated individuals to benefit families and broader communities. Similarly, in the business sector, platforms like ride-sharing and e-commerce face significant challenges in evaluating incentive strategies, as subsidies or promotions create complicated networks of influence where one user's behavior affects others both spatially and temporally. Educational environments also present similar complexities, where the assessment of intervention programs must account for both peer effects among students and the temporal persistence of learning outcomes, especially in evaluating educational innovations. These widespread challenges across interconnected systems necessitate a rigorous framework that can precisely measure and optimize intervention strategies while accounting for both spillover and carryover effects.

To address this gap, we propose a minimax optimal design that effectively integrates spillover and carryover effects within a multi-unit two-stage framework. The first stage involves assigning treated probabilities over a specified time period, while the second stage entails assigning treatment statuses to each unit based on the treated probability. This minimax optimal design aims to identify the optimal time points for randomization to minimize the weighted sum of the mean squared errors of causal effect estimators under worst-case scenarios. Notably, while the design proposed by \citet{bojinov2023} is pioneering for single-unit experiments, it may not preserve its minimax optimality in multi-unit experiments under certain practical scenarios, such as when the order of carryover effects is large or when there is a substantial emphasis on estimating direct effects. Our proposed design addresses this issue and offers a broader generalization that maintains optimality across these challenging scenarios.




Our second major contribution is the analysis of the asymptotic behavior of the Horvitz--Thompson estimators for both direct and spillover effects, along with establishing estimable upper bounds for their variances, under the proposed minimax optimal design. We demonstrate that these estimators are consistent and asymptotically normal as the duration of observation or population size approaches infinity at an appropriate rate. Additionally, we address the potential misspecification of the order of carryover effects. We show that our estimators remain asymptotically normal despite these misspecifications and introduce a method for accurately identifying the order of carryover effects.

To validate our theoretical findings, we conduct comprehensive simulation studies and application-based data analysis, demonstrating the practical efficacy of our proposed design. These results show that the proposed minimax optimal design significantly outperforms existing designs, with the Horvitz--Thompson estimators remaining unbiased and asymptotically normal across a range of scenarios.


The remainder of this paper is organized as follows: In \Cref{sec:setup}, we introduce the framework and notation. In \Cref{sec:minimaxdesign}, we derive unbiased causal effect estimators and present the minimax optimal design. In \Cref{sec:inference}, we examine the asymptotic behavior of the causal effect estimators. In \Cref{sec:simulation} and \Cref{sec:empirical}, we evaluate the performance of the proposed methods through simulation studies and application-based analysis, respectively. We conclude the paper in \Cref{sec:conclusion}. Proofs are relegated to the Supplementary Material.


\section{Framework and notation}
\label{sec:setup}

\subsection{Notation and assumptions}

Consider a public health program aimed at evaluating the effectiveness of a new medication in preventing the spread of an infectious disease. In this study, let $ N $ denote the number of participants enrolled, and $ T $ represent the number of days over which the participants are monitored for their exposure to either the medication or a placebo. At each time point $ t $ (where $ t = 1, \dots, T $), participants are assigned to either the treatment group (receiving the medication) or the control group (receiving the placebo). The assignment status for each participant $ i $ at time $ t $ is denoted by a binary variable $ Z_{i,t} $, where $ Z_{i,t} = 1 $ indicates that participant $ i $ receives the treatment on day $ t $, and $ Z_{i,t} = 0 $ indicates that the placebo is received.
The treatment path for participant $ i $ over a specified time interval from $ t_1 $ to $ t_2 $ is represented by the vector $ \mathbf{Z}_{i,t_1:t_2} = (Z_{i,t_1}, Z_{i,t_1+1}, \dots, Z_{i,t_2})^{\top} $. Furthermore, the treatment matrix $ \mathbf{Z}_{1:N,t_1:t_2} = (\mathbf{Z}_{1,t_1:t_2}, \mathbf{Z}_{2,t_1:t_2}, \dots, \mathbf{Z}_{N,t_1:t_2})^{\top} $ is constructed to present the treatment assignments for all participants over the specified time window. Each row of this matrix corresponds to an individual participant, while each column signifies a distinct time point, thereby providing a comprehensive overview of the treatment assignment dynamics for the entire population throughout the time period.

We define causal effects using the Neyman--Rubin potential outcomes framework \citep{Neyman1923,rubin1974}. Let $ Y_{i,t}(\mathbf{Z}_{1:N,1:T}) $ denote the potential outcome of unit $ i $ at time $ t $ under the entire treatment matrix of all units over the whole time period. Since the number of potential outcomes grows exponentially with the number of units and time periods, it becomes challenging to identify causal effects. To manage this complexity, the traditional causal inference literature often invokes the Stable Unit Treatment Value Assumption (SUTVA) \citep{rubin1980}, which assumes that (i) a unit’s outcome is influenced only by its own treatment (no interference), and (ii) the outcome does not vary depending on how the treatment is administered (no hidden variations).

However, in public health and various other applications, particularly in the context of community-wide interventions, the ``no interference'' assumption required by SUTVA may not be applicable. For instance, when a participant receives a medication to prevent disease, the benefits may extend beyond the individual, potentially reducing transmission to others and resulting in spillover effects. Furthermore, if the treatment induces long-term effects, such as enhancing immunity over time, this can lead to carryover effects, where the impact of a prior treatment endures and continues to influence subsequent outcomes.

In this paper, we relax SUTVA by allowing for both spillover and carryover effects. Specifically, we assume that a unit's outcome at any time $ t $ can be influenced not only by its own treatment but also by the treatments received by others, as well as by their past treatments. Below, we outline three key assumptions that form the basis of our analysis.

\begin{assumption}[Non-anticipativity]
    \label{assumption.nonanticipativity}
    The potential outcome for any unit $i$ at time $ t $ is not influenced by future treatment assignments of any unit. Specifically, for any $ t \in [T-1] = \{1,\dots,T-1\} $ and for all units $ i = 1, \dots, N $, we have
    $Y_{i,t}(\mathbf{Z}_{1:N,1:t}, \mathbf{Z}'_{1:N,(t+1):T}) = Y_{i,t}(\mathbf{Z}_{1:N,1:t}, \mathbf{Z}''_{1:N,(t+1):T})$
    for any $ \mathbf{Z}_{1:N,1:t}$, $ \mathbf{Z}'_{1:N,(t+1):T}$ and $ \mathbf{Z}''_{1:N,(t+1):T}$.
\end{assumption}

\Cref{assumption.nonanticipativity} stipulates that the potential outcome for a unit at time $ t $ depends only on the treatment assignments up to and including time $ t $, and is independent of future treatment assignments. This assumption is common in dynamic treatment regimes and has been used in previous studies \citep{bojinov2019, bojinov2023, han2024}. Under \Cref{assumption.nonanticipativity}, the potential outcome can be simplified to depend only on past and present treatment assignments, denoted as $ Y_{i,t}(\mathbf{Z}_{1:N,1:t}) $. 

\begin{assumption}[$m$-carryover effects]
    \label{assumption.carryover}
    The potential outcome for any unit $i$ at time $t$ depends solely on the treatment history of the past $m+1$ time periods. Specifically, for $t = m+1, \ldots, T$ and $i = 1, \ldots, N$, we have
    $Y_{i,t}(\mathbf{Z}'_{1:N,1:(t-m-1)}, \mathbf{Z}_{1:N,(t-m):t}) = Y_{i,t}(\mathbf{Z}''_{1:N,1:(t-m-1)}, \mathbf{Z}_{1:N,(t-m):t})$
    for any $\mathbf{Z}_{1:N,(t-m):t}$, $\mathbf{Z}'_{1:N,1:(t-m-1)}$, and $\mathbf{Z}''_{1:N,1:(t-m-1)}$.
\end{assumption}

\Cref{assumption.carryover} implies that a unit's potential outcome at a given time is influenced exclusively by the treatment matrix within a pre-defined temporal window. This assumption is essential in capturing carryover effects and is prevalent in related studies \citep{basse2018, imai2021a, bojinov2021, jiang2023, bojinov2023, han2024}. Under Assumptions~\ref{assumption.nonanticipativity} and \ref{assumption.carryover}, the potential outcome simplifies to $Y_{i,t}(\mathbf{Z}_{1:N,(t-m):t})$ when $t \geq m+1$.

\begin{assumption}[Stratified interference]
    \label{assumption.stratifiedinterference}
    The potential outcome for unit $i$ is influenced by other units through treated probabilities of the entire population. That is, for $t\in[m+1,T]$,
    $Y_{i,t}(\mathbf{Z}_{1:N,(t-m):t}) = Y_{i,t}(\mathbf{Z}_{1:N,(t-m):t}')$ if $\mathbf Z_{i,(t-m):t} =\mathbf Z_{i,(t-m):t}'$ and $E(N^{-1}\sum_{j=1}^NZ_{j,t'}) = E(N^{-1}\sum_{j=1}^NZ_{j,t'}')$ for $t' \in [t-m,t]$.
\end{assumption}

\Cref{assumption.stratifiedinterference} extends the stratified interference assumption \citep{hudgens2008, tchetgen2012, liu2014, basse2018, imai2021a, jiang2023} to account for carryover effects. It posits that units are influenced by others only through treated probabilities during the time window $[t-m,t]$. This assumption is meaningful in various real-world contexts. For instance, in vaccination programs, a higher community vaccination rate reduces disease transmission risks, benefiting both vaccinated and unvaccinated units, a phenomenon known as herd immunity. In such cases, a unit's potential outcome is not only influenced by their own treatment status, but also by the community's overall vaccination rate, i.e., the treated probabilities of the entire population. 

\begin{remark}
    In \Cref{assumption.stratifiedinterference}, we stipulate that the potential outcome of each unit is influenced by other units solely through the expected proportion of treated units. However, in certain applications, it might be more fitting to assume that the potential outcome depends on the actual proportion of treated units, i.e., $N^{-1}\sum_{j=1}^N Z_{j,t'}$ instead of its expectation. This distinction tends to be negligible when $N$ is large.
\end{remark}

Assume that at each time point $t$, all units share the same treated probability, denoted by $Q_t = E(Z_{i,t})$ for $i=1,\ldots,N$. Let $\mathbf{Q}_{(t-m):t}$ represent the treated probability path over the time window $[t-m,t]$.
Under Assumptions~\ref{assumption.nonanticipativity}--\ref{assumption.stratifiedinterference}, the potential outcome is further simplified to $Y_{i,t}(\mathbf{Q}_{(t-m):t},  \mathbf{Z}_{i,(t-m):t})$ when $t\geq p+1$. Collectively, these three assumptions extend our analytical framework beyond the confines of SUTVA, enabling a more dynamic exploration of how treatment assignments affect potential outcomes while accounting for interactions across the population.

In practice, we may not always know the true value of $ m $; we can only obtain a specified order, denoted as $ p $, through empirical evidence. In the content preceding \Cref{theorem.misspecified}, we assume $ p = m $. In \Cref{theorem.misspecified}, we specifically address the case where $ p \neq m $ and then present a methodology for identifying $ m $.

\subsection{Causal effects}

We aim to analyze two primary types of effects: direct effects and spillover effects. Under Assumptions~\ref{assumption.nonanticipativity}--\ref{assumption.stratifiedinterference}, direct effects are defined as the causal impact resulting from a unit's treatment status, while maintaining a constant probability of receiving that treatment across the population. For example, in a public health program, a direct effect could be observed as a reduction in disease risk for units who receive the treatment compared to those who do not.
Spillover effects, on the other hand, capture how variations in treated probabilities or proportions influence a unit's outcome when the unit's own treatment status remains unchanged. In the context of a public health program, spillover effects manifest when untreated units benefit from a reduced risk of disease transmission due to others receiving the treatment, illustrating the concept of herd immunity.
Let $\mathbf{1}_{k}$ and $\mathbf{0}_{k}$ denote the $k$-dimensional vectors of all ones and zeros, respectively. For simplicity, the subscript is omitted when $k = p + 1$. The precise formulations for the direct and spillover effects are provided as follows:

\textbf{Direct Effect}: This effect represents the impact of treatment on the unit while keeping treated probabilities constant. The unit-level lag-$p$ direct effect is defined as $\tauid = (T-p)^{-1}\sum_{t=p+1}^{T} \left\{ \yit{q}{}{1} - \yit{q}{}{0} \right\}$,
where $q$ can be $q_1$ or $q_2$ with $0< q_1,q_2<1$, representing different treated probability scenarios. The population-level lag-$p$ direct effect is then obtained by averaging the unit-level effects across all units: $\taud = N^{-1} \sum_{i=1}^{N} \tauid$.

\textbf{Spillover Effect}: This effect examines how treated probabilities affect outcomes for units with the same treatment status. In this paper, we focus on spillover effects between two predefined fixed treated probabilities $q_1$ and $q_2$. Our methods can be extended to the case of multiple treated probabilities. The unit-level lag-$p$ spillover effect is given by $\tauis = (T-p)^{-1} \sum_{t=p+1}^{T} \left\{ \yit{q_1}{z}{1} - \yit{q_2}{z}{1} \right\}$,
where $z$ can be 0 or 1, indicating two different treatment statuses. Similarly, the population-level lag-$p$ spillover effect is the average of the unit-level lag-$p$ spillover effects: $\taus = N^{-1} \sum_{i=1}^{N} \tauis$.

\subsection{Assignment mechanism}

To estimate and infer these causal effects, we employ a treatment assignment mechanism consisting of three key aspects:

\textbf{Decision Points}: There are $L+1$ decision points $1 = t_0 < t_1 < \cdots < t_L \leq T$, which divide the time frame $[T] = \{1, 2, \ldots, T\}$ into $L+1$ intervals: $[t_0, t_1 - 1], [t_1, t_2 - 1], \ldots, [t_L, t_{L+1} - 1]$, where $t_{L+1} = T + 1$. Treatment assignment is made independently at each decision point $t_l$, and the treated probability and treatment statuses remain fixed within each time interval $[t_l, t_{l+1}-1]$ for $l = 0, \ldots, L$. 

\textbf{Treated Probability}: At each decision point $t_l$, we randomly choose the treated probability $Q_{t_l}$ from two options, $q_1$ and $q_2$, with pre-specified probabilities $\R[q_1] \in (0, 1)$ and $\R[q_2] = 1 - \R[q_1]$. Once selected, treated probability remains constant within each time interval $[t_l, t_{l+1}-1]$, such that $Q_t = Q_{t_l}$ for $t \in [t_l, t_{l+1}-1]$.

\textbf{Treatment Status}: At each decision point $t_l$, units are randomly assigned to treatment or control independently based on the treated probability, following $Z_{i,t_l} \stackrel{\mathrm{i.i.d.}}{\sim} \mathrm{Bernoulli}(Q_{t_l})$, where $i = 1, \ldots, N$ and i.i.d. stands for independent and identically distributed. This treatment status persists across the interval, i.e., $Z_{i,t} = Z_{i,t_l}$ for all $t \in [t_l, t_{l+1}-1]$. 

\begin{remark}
    Our proposed methodology could potentially be extended to completely randomized experiments, though the variance expressions under such settings are considerably more complex and may require additional detailed analysis.
\end{remark}

Under this mechanism, all units are assigned to the treatment group with the same treated probability at any given time point, though treatment may not be reassigned at every time step. If the time is not a decision point, the treated probability and treatment statuses from the previous step are retained. This framework encompasses various design options.
For instance, the independent design $\mathbb{T}^1 = \{1, 2, 3, \ldots, T\}$ is the most common design used in previous literature (e.g., \citealt{han2024}). Another design, $\mathbb{T}^2 = \{1, p + 2, 2p + 3, \ldots\}$, divides $[T]$ into periods of length $p + 1$, thereby reducing the influence of previous time points. In this paper, we will obtain a novel design to minimize the mean squared error or risk of causal effect estimators under worst-case scenarios \citep{bojinov2023,zhao2024}.

\begin{remark}
While sharing similarities with two-stage experiments \citep{hudgens2008} and switchback designs \citep{bojinov2023}, our framework differs in significant ways. In two-stage experiments, clusters are first randomized to treatment or control groups, and then within treated clusters, units are randomized to receive the treatment. Our method is distinct in that it is conducted over a time series, requiring careful consideration of when to assign treatments and controls while also accounting for time effects. 
Moreover, our approach does not simply switch between states but rather incorporates a dynamic allocation strategy. This enables us to address specific challenges, particularly minimizing the maximum risks of spillover effects. Thus, our approach is not just an extension of these existing designs but a targeted innovation tailored to the unique conditions of our study.
\end{remark}

Based on this setup, the entire experiment can be characterized by the following components: a set of decision points ($\mathbb{T} = \{t_0, t_1, \ldots, t_{L}\}$); chances of selecting treated probabilities from $q_1,q_2$ ($\R[q_1]$ and $\R[q_2]$); a set of treated probabilities ($\mathbb{Q} = \{Q_{t_0}, Q_{t_1}, \ldots, Q_{t_{L}}\}$); an assignment matrix ($\mathbf{Z}_{1:N, 1:T}$); and a set of all potential outcomes ($\mathbb{Y} = \{Y_{i,t}(\mathbf q_{(t-p):t},\mathbf z_{i,(t-p):t}): i\in [N], \  t\in [T], \ \mathbf q_{(t-p):t}\in \{q_1,q_2\}^{p+1},\ \mathbf z_{i,(t-p):t}\in \{0,1\}^{p+1} \}$). We refer to $\mathbb{T}$ as a design when no confusion arises.

\section{Minimax optimal design}
\label{sec:minimaxdesign}

\subsection{Causal effect estimators}

To estimate the causal effects of interest, we utilize the Horvitz--Thompson estimators, which are designed to account for complex sampling and treatment assignment mechanisms. Let $Y_{i,t}$ represent the observed outcome for the $i$-th unit at time $t$. 

\textbf{Direct Effect Estimator}: The Horvitz--Thompson estimator of the unit-level direct effect is defined as
\begin{align*}
    \tauidhat=\frac 1{T-p}\sum_{t=p+1}^T\left\{Y_{i,t}\frac{\Iit{q}{}{1}}{\Pit{q}{}{1}}-Y_{i,t}\frac{\Iit{q}{}{0}}{\Pit{q}{}{0}}\right\},\quad q=q_1,q_2,
\end{align*}
where $I_{i,t}(\mathbf Q_{(t-p):t},\mathbf Z_{i,(t-p):t})$ is an indicator function that equals 1 if the treated probabilities $\mathbf Q_{(t-p):t}$ and the treatment statuses $\mathbf Z_{i,(t-p):t}$ over the previous $p+1$ time periods match the observed sequences, and 0 otherwise. The term $\mathrm{pr}_{i,t}(\cdot,\cdot)$ denotes the corresponding probability. The estimator for the population-level direct effect is obtained by averaging the unit-level estimators across all units: $\taudhat=N^{-1}\sum_{i=1}^N\tauidhat$.

\textbf{Spillover Effect Estimator}: The Horvitz--Thompson estimator for the unit-level spillover effect is defined as
\begin{align*}
    \tauishat=\frac 1{T-p}\sum_{t=p+1}^T\left\{Y_{i,t}\frac{\Iit{q_1}{z}{1}}{\Pit{q_1}{z}{1}}-Y_{i,t}\frac{\Iit{q_2}{z}{1}}{\Pit{q_2}{z}{1}}\right\}, \quad z=0,1,
\end{align*}
where $q_1$ and $q_2$ represent two different treated probabilities.
The population-level spillover effect estimator is then obtained by averaging the unit-level spillover effect estimators: $\taushat=N^{-1}\sum_{i=1}^N\tauishat$.

\subsection{Minimax optimal design}

The utilization of Horvitz--Thompson estimators guarantees unbiasedness, as the expected value of the indicator function $I_{i,t}(\cdot, \cdot)$ is equal to the corresponding probability $\mathrm{pr}_{i,t}(\cdot, \cdot)$. Consequently, we can evaluate the performance of these estimators by defining the associated risk. One of the most common choices for the risk is the mean squared error, expressed as: $\mathrm{risk}^d(q) = E\left\{ \taudhat - \taud \right\}^2$ and $\mathrm{risk}^s(z) = E\left\{ \taushat - \taus \right\}^2$.
Based on the defined risk functions, experimental design problems can be framed as optimization problems. In this context, our objective function is formulated as $\mathcal L(\psi_d,\psi_s)=\psi_d\{\mathrm{risk}^d(q_1)+\mathrm{risk}^d(q_2)\}+\psi_s\{\mathrm{risk}^s(1)+\mathrm{risk}^s(0)\}$, where $\psi_d\geq 0$, $\psi_s\geq 0$ and $\psi_d+\psi_s=1$. This formulation represents a weighted sum of the risks associated with direct and spillover effects. When $\psi_d = 1$ and $\psi_s = 0$, the focus is solely on the risks associated with direct effects. Conversely, when $\psi_d = 0$ and $\psi_s = 1$, the emphasis shifts entirely to the risks of spillover effects. When $\psi_d = \psi_s = 0.5$, the risks are equally weighted between direct and spillover effects. In this case, our minimax optimization problem features an objective function that resembles the form used in A-optimal design criterion for $\hat{\boldsymbol \tau}=(\taudhat[q_1],\taudhat[q_2],\taushat[1],\taushat[0])$ \citep{atkinson2007,xiong2024}.

In this article, we primarily concentrate on the finite population framework, in which the potential outcomes are treated as fixed, with treatment assignment representing the sole source of randomness. We investigate a minimax optimal design, which can be regarded as a form of robust optimization as discussed by \citet{zhao2024}. Specifically, this approach seeks to minimize the maximum value of the objective function (combined risk) across a given range of potential outcomes, as shown in \Cref{assumption.bounded}.

\begin{assumption}[Bounded potential outcomes]
\label{assumption.bounded}
The potential outcomes are uniformly bounded, meaning there exists a constant $ B > 0 $, such that $| Y_{i,t}(\mathbf q_{(t-p):t},\mathbf z_{i,(t-p):t}) | \leq B$ for $\ i\in [N], \ t\in [T], \ \mathbf q_{(t-p):t}\in \{q_1,q_2\}^{p+1},\ \mathbf z_{i,(t-p):t}\in \{0,1\}^{p+1}$.

\end{assumption}

The bounded potential outcomes assumption is a common condition used in prior analyses of randomized experiments with interference \citep{aronow2017, imai2021a, leung2022,bojinov2023, ni2023, han2024}. We further define the range of potential outcomes as 
$\mathcal{Y} = \{ \mathbb{Y} : | Y_{i,t}(\mathbf q_{(t-p):t},\mathbf z_{i,(t-p):t}) | \leq B \textnormal{ for all } i\in [N], \ t\in [T], \ \mathbf q_{(t-p):t}\in \{q_1,q_2\}^{p+1},\ \mathbf z_{i,(t-p):t}\in \{0,1\}^{p+1}\}$.




To derive the minimax optimal design, i.e., $\argmin_{\mathbb T}\min_{\R[q_1],\R[q_2]}\max_{\mathbb Y\in \mathcal Y}\mathcal L(\psi_d,\psi_s)$, we require some additional notation. Given a design $\mathbb{T}$, we define $\mathcal F_{\mathbb T}(t) = \max\{j \mid j \in \mathbb{T}, j \leq t\}$, which represents the decision point corresponding to time $t$; define $\mathcal F_{\mathbb T}^p(t) = \{j \mid \exists i \in \{t - p, \ldots, t\} \text{ such that } j = \mathcal F_{\mathbb T}(i)\}$, which represents the set of decision points corresponding to time period $[t - p, t]$. Let $J_t = |\mathcal F_{\mathbb T}^p(t)|$ and $J_{t,t'}^{\circ} = |\mathcal F_{\mathbb T}^p(t) \cap \mathcal F_{\mathbb T}^p(t')|$. Finally, we define $\mathcal{J}_j=|\{t\mid J_t=j\}|+|\{(t,t')\mid J_{t,t'}^{\circ}=j,t\neq t'\}|$ as the sum of the number of time points $t$ such that $J_t = j$ and the number of pairs $(t, t')$ with $t \neq t'$ such that $J_{t,t'}^{\circ} = j$. Let $\bar{q}_1 = 1 - q_1$, $\bar{q}_2 = 1 - q_2$, and $\zeta_j=\R[q_1]^{-j}q_1^{-j} + \R[q_1]^{-j}\bar{q}_1^{-j} + \R[q_2]^{-j}q_2^{-j} + \R[q_2]^{-j}\bar{q}_2^{-j}$.

\begin{theorem}[Optimal probability of selecting $q_1$ and $q_2$]
    \label{theorem.r}
    Under Assumptions~\ref{assumption.nonanticipativity}--\ref{assumption.bounded}, for any given design $\mathbb T$,  
 (1) when $N \geq (1 - \max\{\R[q_1], \R[q_2]\})^{-1}$, we have 
    $\argmin_{\R[q_1],\R[q_2]}\max_{\mathbb Y\in \mathcal Y}\mathcal L(\psi_d,\psi_s) = \argmin_{\R[q_1], \R[q_2]}\sum_{j=1}^{p+1}\mathcal J_j \{(4\psi_d+2\psi_s)(\R[q_1]^{-j} + \R[q_2]^{-j})(1-N^{-1})-8\psi_d +\zeta_j N^{-1} \}$; and (2) when $N \leq (1 - \min\{\R[q_1]^{p+1}, \R[q_2]^{p+1}\})^{-1}$, we have 
    $\argmin_{\R[q_1],\R[q_2]}\max_{\mathbb Y\in \mathcal Y}\mathcal L(\psi_d,\psi_s) = \argmin_{\R[q_1], \R[q_2]}\sum_{j=1}^{p+1} $ $\mathcal J_j \{2\psi_s(\R[q_1]^{-j} + \R[q_2]^{-j})(1-N^{-1})+ \zeta_j N^{-1} \}$. 
    In particular, when either $N\to \infty$ or $q_1+q_2=1$, the solution to the minimization problems in both (1) and (2) is $\R[q_1]=\R[q_2]=0.5$.
\end{theorem}


Theorem~\ref{theorem.r} presents how to derive the optimal values of $\R[q_1]$ and $\R[q_2]$ that minimize the maximum combined risk over bounded potential outcomes. When $N\to \infty$ or $q_1+q_2=1$, the optimal values are $\R[q_1]=\R[q_2]=0.5$. However, for a fixed $N$ and $q_1 + q_2 \neq 1$, the exact formula for the optimal $\R[q_1]$ and $\R[q_2]$ is complex and may vary across different scenarios. This conclusion differs from previous findings by \citet{bojinov2023}, highlighting the added complexity introduced by spillover effects. Notably, the optimal values of $\R[q_1]$ and $\R[q_2]$ approach 0.5 as $N \rightarrow \infty$, as shown in \Cref{fig.r} in the Supplementary Material. Thus, for larger values of $N$, $\R[q_1]=\R[q_2]=0.5$ becomes nearly optimal, and we will focus on this scenario for the remainder of the paper.

\begin{theorem}[Minimax optimal design]  
    \label{theorem.minimaxdesign}  
    Under Assumptions~\ref{assumption.nonanticipativity}--\ref{assumption.bounded} and $\R[q_1] = \R[q_2] = 0.5$, the minimax optimal design defined by $\argmin_{\mathbb T} \max_{\mathbb Y\in \mathcal Y}\mathcal L(\psi_d,\psi_s)$ can be obtained by minimizing the following quantity:
    \begin{align*}
        \left\{\sum_{l=0}^{L}(t_{l+1}-t_{l})^2 + (L-1)p^2 + 2p(t_L - t_1)\right\} \gamma_1^* + Lp^2 \gamma_2^* + \left(\sum_{l=2}^{L} \left[\{(p - t_{l} + t_{l-1})^+\}^2\right]\right) \gamma_3^*,
    \end{align*}
    where $\gamma_J^*=\psi_d\gamma_J^d+\psi_s\gamma_J^s$ with $\gamma_J^d$ and $\gamma_J^s$ ($J=1,2,3$) being defined in  \eqref{gamma.d}--\eqref{gamma.s} in the Supplementary Material.
\end{theorem}
Theorem~\ref{theorem.minimaxdesign} presents the optimal design based on the minimax criterion, which typically requires solving an integer optimization problem. Obtaining a general explicit solution for this problem is challenging. Instead, we offer a polynomial-time algorithm with a time complexity of $O(T^2)$ to achieve the minimax optimal design, as detailed in \Cref{poly.algorithm} in the Supplementary Material. In certain special cases, we can derive the explicit form of the minimax optimal design. Specifically, define
\begin{align*}
    \theta^*=\frac{\gamma_2^*}{\gamma_1^*}=\frac{4N\psi_d I(N\geq 2)+(2q_1^{-2}+2\bar q_1^{-2}-2q_1^{-1}-2\bar q_1^{-1}+2q_2^{-2}+2\bar q_2^{-2}-2q_2^{-1}-2\bar q_2^{-1})}{(4N-4-4\psi_d)I(N\geq 2)+(q_1^{-1}+\bar q_1^{-1}+q_2^{-1}+\bar q_2^{-1})},
\end{align*}
where $I(\cdot)$ is the indicator function. Based on $\theta^*$, the quantity in \Cref{theorem.minimaxdesign} can be rewritten as $\{\sum_{l=0}^{L}(t_{l+1}-t_{l})^2 + (\theta^*L+L-1)p^2 + 2p(t_L - t_1)\} \gamma_1^*+ \sum_{l=2}^{L} [\{(p - t_{l} + t_{l-1})^+\}^2] \gamma_3^*$.
\Cref{corollary.minimaxdesign} below provides additional insights and explicit solutions for specific scenarios. 

\begin{corollary}
    \label{corollary.minimaxdesign}
    Under Assumptions~\ref{assumption.nonanticipativity}--\ref{assumption.bounded} and $\R[q_1] = \R[q_2] = 0.5$, we have
    
    (1) when $p=0$, the minimax optimal design is given by $\mathbb T^1=\{1,2,3,\ldots,T\}$; 
    
    (2) when $p>0$, let $b^*,a^*\in \mathbb N$ be defined such that $b^*\in[\{-1+\sqrt{1+4(\theta^*+1)p^2}\}/2,\{1+\sqrt{1+4(\theta^*+1)p^2}\}/2]$ and $a^*\in[\{2p+b^*+(\theta^*+1)p^2/b^*-1\}/2,\{2p+b^*+(\theta^*+1)p^2/b^*+1\}/2]$. When $T-2a^*$ is a multiple of $b^*$ with $(T-2a^*)/b^*=K-4\geq 0$, where $K$ is an integer, the minimax optimal design is $\mathbb T^*=\{1,a^*+1,a^*+b^*+1,\ldots,a^*+(K-4)b^*+1\}$. Two typical minimax optimal designs are given by:

    (2.1) when $\theta^*\leq 1/p$, $\aop=2p$ and $\bop=p$. If $T-4p$ is a multiple of $p$ with $(T-4p)/p=K-4\geq 0$, the minimax optimal design is $\mathbb T^*_1=\{1,2p+1,3p+1,\ldots,(K-2)p+1\}$;

    (2.2) when $\theta^*\in(1/p, (3p+2)/p^2]$, $\aop=2p+1$ and $\bop=p+1$. If $T-4p-2$ is a multiple of $p+1$ with $(T-4p-2)/(p+1)=K-4\geq 0$, the minimax optimal design is $\mathbb T^*_2=\{1,2p+2,3p+3,\ldots,(K-2)(p+1)\}$.
\end{corollary}

\Cref{corollary.minimaxdesign} provides a comprehensive overview of the minimax optimal design under various scenarios. When $ p = 0 $ (indicating the absence of carryover effects), the minimax optimal design simplifies to the independent design $ \mathbb{T}^1 $. When $ p > 0 $ (indicating the presence of carryover effects), the minimax optimal design varies depending on the value of $\theta^*$.
Specifically, when $\theta^*\leq 1/p$, the proposed minimax optimal design is given by $\mathbb T^*_1=\{1,2p+1,3p+1,\ldots,(K-2)p+1\}$, which is the same as the design in single-unit experiments proposed by \citet{bojinov2023}. However, when $\theta^*>1/p$, the structure of the minimax optimal design exhibits increased complexity. It manifests manifests under two conditions: either in the presence of high-order carryover effects, or when the estimation of direct effects receives greater priority. The latter scenario occurs because $\gamma_2^d/\gamma_1^d$ surpasses $\gamma_2^s/\gamma_1^s$ for sufficiently large values of $N$. In this case, a typical minimax optimal design is $\mathbb T^*_2=\{1,2p+2,3p+3,\ldots,(K-2)(p+1)\}$ for $\theta^*\in(1/p, (3p+2)/p^2]$. Table~\ref{table.example} below presents an example for $T = 16$ and $p = 2$, comparing $\mathbb{T}^*_1$, $\mathbb{T}^*_2$, $\mathbb{T}^1$ and $\mathbb{T}^2$. In this table, decision points are marked with ``\ding{51}'' and non-decision points with ``\ding{55}''.

\begin{table}[h!]
    \centering
    \caption{An example of $\mathbb T^*_1$, $\mathbb T^*_2$, $\mathbb T^1$ and $\mathbb T^2$ when $T=16$ and $p=2$}\label{table.example}
    \begin{tabular}{|c|c|c|c|c|c|c|c|c|c|c|c|c|c|c|c|c|}
        \hline
        \textbf{Design} & 1 & 2 & 3 & 4 & 5 & 6 & 7 & 8 & 9 & 10 & 11 & 12 & 13 & 14 & 15 & 16
        \\ \hline
        $\mathbb{T}^*_1$ & \ding{51} & \ding{55} & \ding{55} & \ding{55} & \ding{51} & \ding{55} & \ding{51} & \ding{55} & \ding{51} & \ding{55} & \ding{51} & \ding{55} & \ding{51} & \ding{55} & \ding{55} & \ding{55} \\ \hline
        $\mathbb{T}^*_2$ & \ding{51} & \ding{55} & \ding{55} & \ding{55} & \ding{55} & \ding{51} & \ding{55} & \ding{55} & \ding{51} & \ding{55} & \ding{55} & \ding{51} & \ding{55} & \ding{55} & \ding{55} & \ding{55} \\ \hline
        $\mathbb{T}^1$ & \ding{51} & \ding{51} & \ding{51} & \ding{51} & \ding{51} & \ding{51} & \ding{51} & \ding{51} & \ding{51} & \ding{51} & \ding{51} & \ding{51} & \ding{51} & \ding{51} & \ding{51} & \ding{51} \\ \hline
        $\mathbb{T}^2$ & \ding{51} & \ding{55} & \ding{55} & \ding{51} & \ding{55} & \ding{55} & \ding{51} & \ding{55} & \ding{55} & \ding{51} & \ding{55} & \ding{55} & \ding{51} & \ding{55} & \ding{55} & \ding{55} \\ \hline
    \end{tabular}
    \label{tab:design_comparison}
\end{table}

For different values of $ N $, $ \psi_d $ and $ \psi_s $, we can calculate the exact value of $ \theta^* $. Based on the computed value of $ \theta^* $, we can subsequently determine $ b^* $ and $ a^* $, which allows us to obtain the minimax optimal design $\mathbb T^*$ based on \Cref{corollary.minimaxdesign}. In the next section, we will derive the design-based asymptotic properties of causal effect estimators under the minimax optimal design $\mathbb T^*$.

\section{Asymptotic properties and inference}
\label{sec:inference}

\subsection{Variances and their estimators}

Under $\mathbb T^*$, if $p>0$ and $T-2a^*$ is a multiple of $b^*$, we can derive the exact expression for the variances of the Horvitz--Thompson estimators. Due to its lengthy form, the detailed expression is relegated to \Cref{theorem.variance.sm} in the Supplementary Material.


The variances of the Horvitz--Thompson estimators involve terms that account for the variability of outcomes across time periods and assignment paths. However, due to the inherent nature of randomization, these variances also depend on covariances of potential outcomes that cannot be unbiasedly estimated since only one outcome can be observed per unit at each time. This limitation necessitates the use of upper bounds derived from the Cauchy--Schwarz inequality to conservatively estimate the variances. Due to space limitations, the explicit formulas of the conservative variance estimators are provided in~\Cref{corollary.estimator} in the Supplementary Material. Notably, the variances and their conservative estimators can be applied to all minimax optimal designs discussed in \Cref{corollary.minimaxdesign}, including two typical designs $\mathbb T^*_1$ and $\mathbb T^*_2$.

\subsection{Asymptotic normality}
\label{sec:normality}

In some experimental scenarios, geographical constraints may necessitate conducting studies across multiple independent centers. In these cases, we can independently conduct the minimax optimal design within each center and estimate causal effects separately, and then combine these estimates. This approach allows us to account for potential heterogeneity across different centers while still providing a unified estimate. In this section, we will discuss the asymptotic properties of the Horvitz--Thompson estimators within this multi-center framework.
It is important to note that the single-center case can be considered a special instance of the multi-center setting, where the number of centers is simply one. 
The direct and spillover effects can be estimated separately in each center, indexed by $g = 1, \ldots, G$. The proportion of units in center $ g $ is denoted by $\pi_{[g]} = N_{[g]}/N$, with $ N_{[g]} $ indicating the number of units in center $ g $ and $ N $ representing the total number of units across all centers.
The overall causal effect is then estimated by the weighted estimator $\hat{\tau}^{*}(\dagger) = \sum_{g=1}^{G} \pi_{[g]} \hat{\tau}^{*}_{[g]}(\dagger)$ for $* = d,s$ corresponding to $\dagger = q,z$, where $\hat{\tau}^{*}_{[g]}(\dagger)$ is the Horvitz--Thompson estimator obtained in center $ g $.

With an additional assumption that the centers are independent, for example, they are far enough from each other, the variance of the weighted estimator is $\mathrm{var}\{\hat{\tau}^{*}(\dagger)\}=\sum_{g=1}^{G} \pi_{[g]}^2 \mathrm{var}\{\hat{\tau}^{*}_{[g]}(\dagger)\}$. It can be estimated by $\widehat{\mathrm{var}}^U\{\hat{\tau}^{*}(\dagger)\} = \sum_{g=1}^{G} \pi_{[g]}^2 \widehat{\mathrm{var}}^U\{\hat{\tau}^{*}_{[g]}(\dagger)\}$, where $\widehat{\mathrm{var}}^U\{\hat{\tau}^{*}_{[g]}(\dagger)\}$ is the variance estimator for the Horvitz--Thompson estimator in center $g$ defined in the Supplementary Material. This formula captures the variability within each center while accounting for the contribution of each center to the overall estimate, weighted by the squared proportion of units. In the following discussion, we will explore the asymptotic behavior of the weighted estimator $\hat{\tau}^{*}(\dagger)$ under the minimax optimal design, laying out the necessary assumptions for ensuring consistency and asymptotic normality. 

\begin{remark}
    In the multi-center framework, we essentially extend the single-center framework by assuming partial interference \citep{hudgens2008}. Specifically, the units are grouped into $G$ clusters, and both interference and carryover effects are restricted to occur within the same cluster.
\end{remark}


\begin{assumption}
    \label{assumption.variance}
    As $NT\rightarrow \infty$, $\mathrm{var}\{\sqrt{N(T-p)}\hat \tau^d(q)\}\rightarrow \sigma_{d}^2(q)>0$ and $\mathrm{var}\{\sqrt{N(T-p)}\hat \tau^s(z)\}\rightarrow \sigma_{s}^2(z)>0$ for $q=q_1,q_2$ and $z=0,1$.
\end{assumption}

\Cref{assumption.variance} ensures that the variability associated with the estimators does not vanish or diverge as the number of units ($N$) and time periods ($T$) increase. It rules out situations where the variance could be affected by only a few time points, thereby ensuring robust asymptotic inference. 

\begin{assumption}
    \label{assumption.NT}
    Let $N_{\max} = \max_{g=1,\ldots,G} N_{[g]}$.
    As $NT\rightarrow \infty$, there exists a pair of constants $(\alpha,\beta)$, such that $N^{-\alpha}T^{-\beta}\rightarrow 0$ and $N^{\alpha-1}T^{\beta-1}N_{\max}\rightarrow 0$.
\end{assumption}

\Cref{assumption.NT} imposes a growth rate condition on the total number of units ($N$) and time periods ($T$), requiring that the maximum number of units in any single center ($N_{\max}$) does not grow too rapidly relative to the overall population size. It is satisfied if (i) $N$ is fixed and $T$ tends to infinity; (ii) $T$ is fixed and $N$ tends to infinity with $N_{\max} = o(N^{1-\alpha})$ for $\alpha > 0$; or (iii) both $N$ and $T$ tend to infinity. \citet{bojinov2023} considered a single-unit experiment, which is a special case of (i). Notably, when $N$ tends to infinity with $T$ fixed, the experiment needs to be conducted in multiple centers to ensure the validity of the asymptotic results.
Therefore, in practice, when the total number of units is large, we recommend extending the time horizon or conducting the experiment in multiple independent centers if possible to satisfy \Cref{assumption.NT}.
Based on Assumptions~\ref{assumption.nonanticipativity}--\ref{assumption.NT}, we can derive the asymptotic normality of the Horvitz--Thompson estimators.

\begin{theorem}[Asymptotic normality]
    \label{theorem.CLT}
    Suppose that $p$ is fixed and Assumptions~\ref{assumption.nonanticipativity}--\ref{assumption.NT} hold. Under the minimax optimal design $\mathbb T^*$ with $\R[q_1] = \R[q_2] = 0.5$ and as $NT\rightarrow\infty$, the Horvitz--Thompson estimators for the direct and spillover effects satisfy
    \begin{align*}
        \frac{\taudhat-\taud}{\sqrt{\mathrm{var}\{\taudhat\}}}\stackrel{d}{\rightarrow}\mathcal N(0,1),\quad q = q_1, q_2, \quad \frac{\taushat-\taus}{\sqrt{\mathrm{var}\{\taushat\}}}\stackrel{d}{\rightarrow}\mathcal N(0,1), \quad z=0,1.
    \end{align*}
\end{theorem}

\Cref{theorem.CLT} implies that the Horvitz--Thompson estimators are consistent and asymptotically normal under the minimax optimal design when the order of carryover effects is known in advance.
Notably, the asymptotic normality holds not only for $\mathbb T^*_1$ or $\mathbb T^*_2$, but also for all minimax optimal designs discussed in \Cref{corollary.minimaxdesign}.
Using consistent upper bound estimates for the variances (as provided in \Cref{corollary.estimator} in the Supplementary Material), we can construct asymptotically conservative $1-\alpha$ confidence intervals for the direct and spillover effects.


In practice, however, the true order of carryover effects may not be known. Recall that $m$ is the true order and $p$ is the specified order. \Cref{theorem.misspecified} below presents the asymptotic results when the order of carryover effects is misspecified.

\begin{theorem}[Asymptotic normality when the order of carryover effects is misspecified]
    \label{theorem.misspecified}
    Suppose that $p$ and $m$ are fixed and Assumptions~\ref{assumption.nonanticipativity}--\ref{assumption.NT} hold. For the cases where either $p>m$ or $p<m$, under the minimax optimal design $\mathbb T^*$ with $\R[q_1] = \R[q_2] = 0.5$, and as $NT\rightarrow\infty$, the Horvitz--Thompson estimators are still asymptotically normal. Specifically, 
    \begin{align*}
        \frac{\taudhat-E\{\taudhat\}}{\sqrt{\mathrm{var}\{\taudhat\}}}\stackrel d{\rightarrow}\mathcal N(0,1), \quad q = q_1, q_2, \quad \frac{\taushat-E\{\taushat\}}{\sqrt{\mathrm{var}\{\taushat\}}}\stackrel d{\rightarrow}\mathcal N(0,1), \quad z=0,1.
    \end{align*}
\end{theorem}

\Cref{theorem.misspecified} shows that even if the specified order of carryover effects is incorrect, the Horvitz--Thompson estimators retain asymptotic normality. However, the expectations and the variances of the estimators differ depending on whether $p > m$ or $p < m$. Specifically,

(i) when $p > m$, the expectations satisfy $E\{\taudhat\} = \taud$ and $E\{\taushat\} = \taus$, with the variance formulas matching those when the order is correctly specified. Although the design remains optimal with respect to $p$, the true variances may be larger than those when the design is minimax optimal for $m$; see the simulation results in \Cref{sec:simulation};

(ii) when $p < m$, a non-negligible asymptotic bias is introduced, which makes it challenging to derive conservative variance estimators from the observed data. The precise expressions of expectations can be found in the Supplementary Material. 

A larger $p$ typically reduces the effective sample size, leading to a larger risk, but a smaller $p$ may introduce a non-negligible asymptotic bias when the true order of carryover effects is larger than $p$.
Therefore, it is advisable to choose $p$ larger than $m$ but as close as possible. Leveraging domain knowledge and prior experimental or observational data can help in accurately determining the true order $m$. In \Cref{section.order}, we will discuss a data-driven approach to identify $m$.

\subsection{Identifying the order of carryover effects}
\label{section.order}

In this section, we present a method similar to \cite{bojinov2023} for identifying the order of carryover effects $m$ based on asymptotic normality and hypothesis testing. 
The approach relies on comparing two different experimental designs, which can be applied to either distinct but comparable populations or the same population at different time points sufficiently spaced apart to avoid interference and carryover effects. 
The designs are based on minimax optimal designs with specified orders $ p_1 $ and $ p_2 $, where $ p_1 < p_2 $. We aim to test the null hypothesis $ H_0: m \leq p_1 $, indicating that the true order of carryover effects does not exceed $ p_1 $.
To facilitate the discussion, we use subscripts $ [p_1] $, $ [p_2] $, and $ [m] $ to denote quantities under the designs with orders $ p_1 $, $ p_2 $, and the true order $ m $, respectively. Under $ H_0: m \leq p_1 $, we have
(i) $\taudhatorder[p_1]$ and $\taudhatorder[p_2]$ are unbiased for $\taudorder[m]$, and $\taushatorder[p_1]$ and $\taushatorder[p_2]$ are unbiased for $\tausorder[m]$, and (ii) the following test statistics are asymptotically standard normal:
    \begin{align*}
        \frac{\taudhatorder[p_1]-\taudhatorder[p_2]}{\sqrt{\mathrm{var}\{\taudhatorder[p_1]\}+\mathrm{var}\{\taudhatorder[p_2]\}}},\quad \frac{\taushatorder[p_1]-\taushatorder[p_2]}{\sqrt{\mathrm{var}\{\taushatorder[p_1]\}+\mathrm{var}\{\taushatorder[p_2]\}}}.
    \end{align*}
These results allow us to construct test statistics for testing $ H_0: m \leq p_1 $, given by:
$$
T_d(q) = \frac{\taudhatorder[p_1]-\taudhatorder[p_2]}{\sqrt{\widehat{\mathrm{var}}^U\{\taudhatorder[p_1]\}+\widehat{\mathrm{var}}^U\{\taudhatorder[p_2]\}}}, \quad T_s(z) = \frac{\taushatorder[p_1]-\taushatorder[p_2]}{\sqrt{\widehat{\mathrm{var}}^U\{\taushatorder[p_1]\}+\widehat{\mathrm{var}}^U\{\taushatorder[p_2]\}}}.
$$
We reject the null hypothesis $ H_0: m \leq p_1 $ if either $ |T_d| $ or $ |T_s| $ exceeds $ \Phi^{-1}(1 - \alpha/2) $. This rejection indicates that the true order of carryover effects exceeds $ p_1 $.

The rationale behind this approach is that, if $ m \leq p_1 $, then the estimators from both designs ($ p_1 $ and $ p_2 $) should be consistent for the same target parameter, leading to similar estimates with variances reflecting the differences in design. However, if $ m > p_1 $, the estimators from the design with $ p_1 $ may exhibit systematic bias, resulting in significant deviations between the estimators from the two designs.
This method provides a systematic way to identify the order of carryover effects by leveraging different designs and asymptotic properties of the estimators.

\section{Simulation}
\label{sec:simulation}

In this section, we evaluate the finite sample performance of the Horvitz--Thompson estimators under different experimental designs through simulations. The study comprises three parts, each focusing on a specific aspect: minimax optimal design, inference, and carryover effects order misspecification.

\subsection{Minimax optimal design}
\label{sec:od}

We first assess the finite sample performance of two typical minimax optimal design $\mathbb T^*_1=\{1,2p+1,3p+1,\ldots,T-2p+1\}$ and $\mathbb T^*_2=\{1,2p+2,3p+3,\ldots,T-2p\}$, by comparing them with two alternative designs: $\mathbb T^1=\{1,2,3,\ldots,T\}$, a common design that assigns units in each time period, and $\mathbb T^2=\{1,p+2,2p+3,\ldots\}$, a common design that divides $[T]$ into periods of length $p+1$. It is important to note that $\mathbb T^*_1$ serves as the minimax optimal design when $\theta^*\leq 1/p$ and $\mathbb T^*_2$ is optimal when $\theta^*\in (1/p,(3p+2)/p^2]$.
The comparison focuses on how the value of objective function (combined risk) changes when using the minimax optimal design suggested by \Cref{corollary.minimaxdesign} versus alternative designs, under the condition that the order of carryover effects is correctly specified. We set parameters as $q_1=0.6$ and $q_2=0.4$. The potential outcomes are generated by two models:

Model 1 (worst case). We set $\yit{q}{}{1}=B$ and $\yit{q}{}{0}=-B$, aligning with the worst-case scenario. In this model, we set $B=1$ without loss of generality.

Model 2 (linear model with time-specific fixed effect and normal noise). 
The potential outcomes are generated by the following linear model: for $q\in\{q_1,q_2\}$ and $z\in\{0,1\}$, $\yit{q}{z}{1}=\alpha_t+\epsilon_{i,t}+\sum_{\Delta t=0}^m\{\delta_q^{(\Delta t)} I_{i,t-\Delta t}(q=q_1)+\delta_z^{(\Delta t)} I_{i,t-\Delta t}(z=1)+\delta_{q,z}^{(\Delta t)}I_{i,t-\Delta t}(q=q_1,z=1)\},$
where $\alpha_t$ is a time-specific fixed effect, $\epsilon_{i,t}\stackrel{\text{i.i.d.}}{\sim}\mathcal N(0,1)$ is random error, the coefficients $\delta_q^{(\Delta t)}, \delta_z^{(\Delta t)}, \delta_{q,z}^{(\Delta t)}$ represent the effects of treated probabilities, assignment statuses, and their interactions, and $I_{i,t}(\cdot)$ is the indicator function. We set all $\delta$ terms to 1 when $\Delta t\leq m$ and 0 otherwise, additionally with $\alpha_t = \log (t)$.

The potential outcomes remain fixed throughout the simulations, while the treatment assignment process is repeated 1,000 times to calculate three objective functions: $\mathcal L(1,0)$, $\mathcal L(0,1)$, and $\mathcal L(0.5,0.5)$.
The results across various scenarios are summarized in Table~\ref{table.optimaldesign}. Our analysis reveals that the minimax optimal design $\mathbb{T}^*$ ($\mathbb{T}^*_1$ or $\mathbb{T}^*_2$), consistently yields lower risk values compared to other design options. The optimal design reduces the risk by an average of 19.7\% and 3.7\%, respectively, compared to $\mathbb T^1$ and $\mathbb T^2$. Furthermore, the choice between the optimal designs $\mathbb{T}^*_1$ and $\mathbb{T}^*_2$ is contingent upon the value of $\theta^*$: $\mathbb{T}^*_1$ is optimal when $\theta^* \leq 1/p$, while $\mathbb{T}^*_2$ is preferred when $1/p < \theta^* \leq (3p+2)/p^2$. This finding underscores the superior performance of the minimax designs in terms of risk reduction, aligning well with the theoretical predictions outlined in \Cref{theorem.minimaxdesign} and \Cref{corollary.minimaxdesign}.


\begin{table}[ht]
    \centering
    \caption{Value of the objective function (combined risk) under different designs}\label{table.optimaldesign}
    \begin{threeparttable}
        \begin{tabular}{ccccccccccc} 
            \hline
            Model & $p$ & $N$ & $T$ & Obj. Func. & $\theta^*$ & $\mathbb T^*$ & $\mathbb T^*_1$ & $\mathbb T^*_2$ & $\mathbb T^1$ & $\mathbb T^2$\\
            \hline
            \multirow{6}{*}{1} & \multirow{3}{*}{1} & \multirow{3}{*}{20} & \multirow{3}{*}{100} & $\mathcal L(1,0)$ & 1.238 & $\mathbb T^*_2$ & 0.412 & \textbf{0.403} & 0.416 & 0.410\\
            &&&& $\mathcal L(0,1)$ & 0.296 & $\mathbb T^*_1$ & \textbf{0.351} & 0.382 & 0.354 & 0.388\\
            &&&& $\mathcal L(0.5,0.5)$ & 0.722 & $\mathbb T^*_1$ & \textbf{0.382} & 0.392 & 0.385 & 0.399\\
            \cline{2-11}
            & \multirow{3}{*}{2} & \multirow{3}{*}{20} & \multirow{3}{*}{160} & $\mathcal L(1,0)$ & 1.238 & $\mathbb T^*_2$ & 0.526 & \textbf{0.497} & 0.854 & 0.504\\
            &&&& $\mathcal L(0,1)$ & 0.296 & $\mathbb T^*_1$ & \textbf{0.447} & 0.455 & 0.634 & 0.459\\
            &&&& $\mathcal L(0.5,0.5)$ & 0.722 & $\mathbb T^*_2$ & 0.487 & \textbf{0.476} & 0.744 & 0.481\\
            \cline{1-11}
            \multirow{6}{*}{2} & \multirow{3}{*}{1} & \multirow{3}{*}{20} & \multirow{3}{*}{100} & $\mathcal L(1,0)$ & 1.238 & $\mathbb T^*_2$ & 2.949 & \textbf{2.895} & 3.003 & 2.954\\
            &&&& $\mathcal L(0,1)$ & 0.296 & $\mathbb T^*_1$ & \textbf{15.335} & 16.502 & 15.530 & 16.940\\
            &&&& $\mathcal L(0.5,0.5)$ & 0.722 & $\mathbb T^*_1$ & \textbf{9.142} & 9.698 & 9.266 & 9.947\\
            \cline{2-11}
            & \multirow{3}{*}{2} & \multirow{3}{*}{20} & \multirow{3}{*}{160} & $\mathcal L(1,0)$ & 1.238 & $\mathbb T^*_2$ & 6.876 & \textbf{6.471} & 13.975 & 6.566\\
            &&&& $\mathcal L(0,1)$ & 0.296 & $\mathbb T^*_1$ & \textbf{30.685} & 31.007 & 43.367 & 31.412\\
            &&&& $\mathcal L(0.5,0.5)$ & 0.722 & $\mathbb T^*_2$ & 18.781 & \textbf{18.739} & 28.671 & 18.989\\
            \cline{1-11}
        \end{tabular}
        \begin{tablenotes}
        \footnotesize
        \item[] Note: ``Obj. Func.'' is short for ``Objective Function''.
        \end{tablenotes}
    \end{threeparttable}
\end{table}

\subsection{Asymptotic normality and confidence interval}
\label{section.simulation.normal}

In this section, we conduct simulations to evaluate the asymptotic normality and confidence intervals under the minimax optimal design. We begin by utilizing the same parameters as outlined in \Cref{sec:od}. Furthermore, we extend our analysis to multi-center experiments and investigate a broader range of sample sizes and time periods to assess the impact of Assumption~\ref{assumption.NT}.
Our simulations encompass three scenarios regarding the specification of the carryover effects order:

Case 1 (Correct Specification): $p=m=2$. The estimands are well-defined, and the Horvitz--Thompson estimators exhibit asymptotic normality; 

Case 2 (Over-Specification): $p=3>m = 2$. The estimands remain well-defined, and the Horvitz--Thompson estimators still display asymptotic normality; 

Case 3 (Under-Specification): $p=1<m = 2$. Although the estimands are not well-defined, the Horvitz--Thompson estimators still exhibit asymptotic normality. 

The results presented in this section focus on Model 2 under the minimax optimal design $\mathbb{T}^*_1$, while analogous findings under another typical minimax optimal design $\mathbb{T}^*_2$ can be found in the Supplementary Material.
Table~\ref{table.single} presents the simulation results for Cases 1 and 2 under $\mathbb T^*_1$, including metrics such as bias, variance, variance estimator, and empirical coverage probability (CP) for 95\% confidence intervals. Key observations include: (1) The bias of the Horvitz--Thompson estimators is negligible, supporting the unbiasedness of the estimators under both correctly specified and over-specified carryover effects orders. (2) The variance estimators closely approximate the asymptotic variances, although they are slightly conservative. This observation is consistent with \Cref{theorem.variance.sm}, \Cref{corollary.upperbound}, and \Cref{corollary.estimator}. (3) The coverage probability of the confidence intervals approaches the target level of 95\%, validating the asymptotic normality of the estimators as stated in Theorem~\ref{theorem.CLT}. (4) The variance is larger in the case of over-specification (Case 2) compared to correct specification (Case 1), which is consistent with discussions in Section~\ref{sec:normality}.

\begin{table}[!h]
    \centering
    \caption{Simulation results in single-center randomized experiments under $\mathbb T^*_1$}\label{table.single}
    \begin{threeparttable}
        \begin{tabular}{ccccccccccc}
            \hline
            $N$ & $T$ & $p$ & Estimand & Value & Bias & $\mathrm{var}(\hat \tau)$ & $\widehat{\mathrm{var}}^U(\hat \tau)$ & CP\\
            \hline
            \multirow{4}{*}{10} & \multirow{4}{*}{480} & \multirow{4}{*}{2} & $\taud[q_1]$ & 6 & -0.05 & 3.20 & 3.16 & 0.953\\
            & & & $\taud[q_2]$ & 3 & 0.04 & 1.33 & 1.40 & 0.956\\
            & & & $\taus[1]$ & 6 & -0.08 & 9.21 & 10.66 & 0.969\\
            & & & $\taus[0]$ & 3 & 0.01 & 3.98 & 4.13 & 0.955\\
            \hline
            \multirow{4}{*}{10} & \multirow{4}{*}{480} & \multirow{4}{*}{3} & $\taud[q_1]$ & 6 & -0.04 & 4.44 & 4.76 & 0.955\\
            & & & $\taud[q_2]$ & 3 & 0.04 & 1.99 & 2.09 & 0.946\\
            & & & $\taus[1]$ & 6 & -0.05 & 14.99 & 16.05 & 0.953\\
            & & & $\taus[0]$ & 3 & 0.04 & 6.21 & 6.21 & 0.960\\
            \hline
            \multirow{4}{*}{20} & \multirow{4}{*}{480} & \multirow{4}{*}{2} & $\taud[q_1]$ & 6 & -0.02 & 1.84 & 1.95 & 0.950\\
            & & & $\taud[q_2]$ & 3 & -0.01 & 0.75 & 0.78 & 0.952\\
            & & & $\taus[1]$ & 6 & -0.03 & 8.68 & 9.78 & 0.964\\
            & & & $\taus[0]$ & 3 & -0.02 & 3.18 & 3.62 & 0.963\\
            \hline
            \multirow{4}{*}{20} & \multirow{4}{*}{480} & \multirow{4}{*}{3} & $\taud[q_1]$ & 6 & 0.03 & 2.94 & 2.96 & 0.944\\
            & & & $\taud[q_2]$ & 3 & 0.03 & 1.20 & 1.18 & 0.940\\
            & & & $\taus[1]$ & 6 & -0.03 & 13.77 & 14.76 & 0.959\\
            & & & $\taus[0]$ & 3 & -0.03 & 5.28 & 5.42 & 0.953\\
            \hline
            \multirow{4}{*}{20} & \multirow{4}{*}{720} & \multirow{4}{*}{2} & $\taud[q_1]$ & 6 & 0.04 & 1.43 & 1.38 & 0.939\\
            & & & $\taud[q_2]$ & 3 & -0.02 & 0.56 & 0.56 & 0.934\\
            & & & $\taus[1]$ & 6 & 0.11 & 6.93 & 7.01 & 0.949\\
            & & & $\taus[0]$ & 3 & 0.05 & 2.66 & 2.70 & 0.949\\
            \hline
            \multirow{4}{*}{20} & \multirow{4}{*}{720} & \multirow{4}{*}{3} & $\taud[q_1]$ & 6 & -0.02 & 1.94 & 2.04 & 0.950\\
            & & & $\taud[q_2]$ & 3 & -0.02 & 0.81 & 0.85 & 0.953\\
            & & & $\taus[1]$ & 6 & -0.03 & 9.32 & 10.46 & 0.962\\
            & & & $\taus[0]$ & 3 & -0.03 & 3.56 & 4.04 & 0.965\\
            \hline
        \end{tabular}
        \begin{tablenotes}
        \footnotesize
        \item[] Note: Value, true value; CP, coverage probability.
        \end{tablenotes}
    \end{threeparttable}
\end{table}

For Case 3 (under-specification of $m$), although the estimands are not well-defined, the Horvitz--Thompson estimators still demonstrate asymptotic normality, consistent with the predictions of Theorem~\ref{theorem.misspecified}. This behavior is visually confirmed in Figure~\ref{fig.mis}, which presents Q-Q plots for the Horvitz--Thompson estimators under $\mathbb T^*_1$ with $ N = 10 $ and $ T = 480 $. The points closely follow the diagonal line, indicating a good fit to the normal distribution.

\begin{figure}[!h]
    \centering
    
    \begin{subfigure}[b]{0.4\linewidth}
        \caption{Q-Q plot of $\taudhat[q_1]$}
    \includegraphics[width=\linewidth]{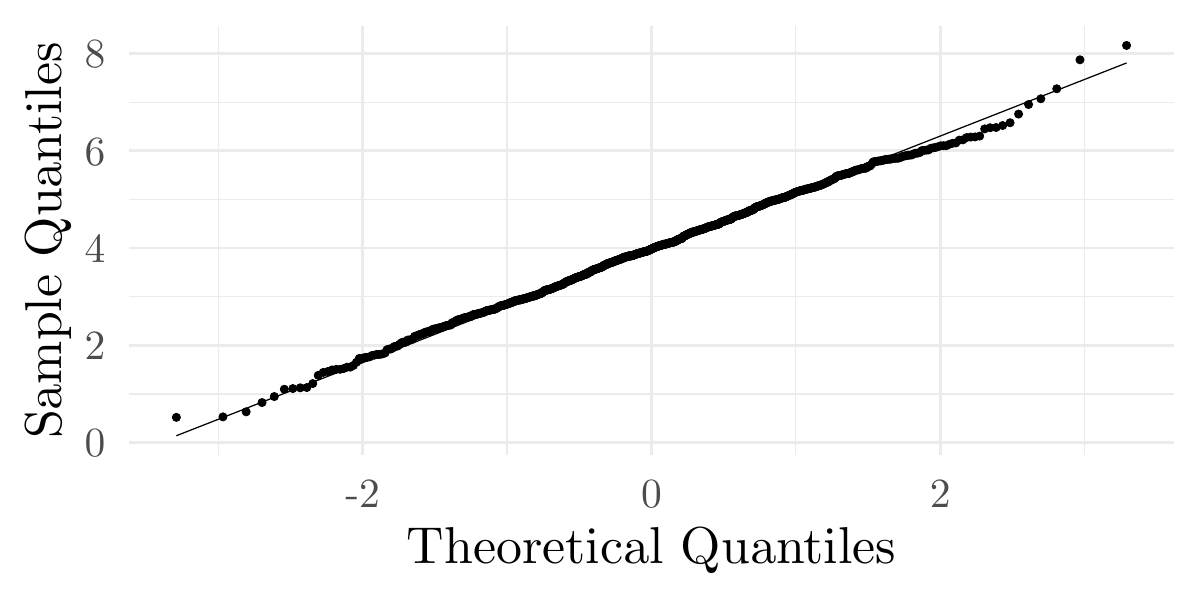}
    \end{subfigure}
    \quad
    \begin{subfigure}[b]{0.4\linewidth}
        \caption{Q-Q plot of $\taudhat[q_2]$}
    \includegraphics[width=\linewidth]{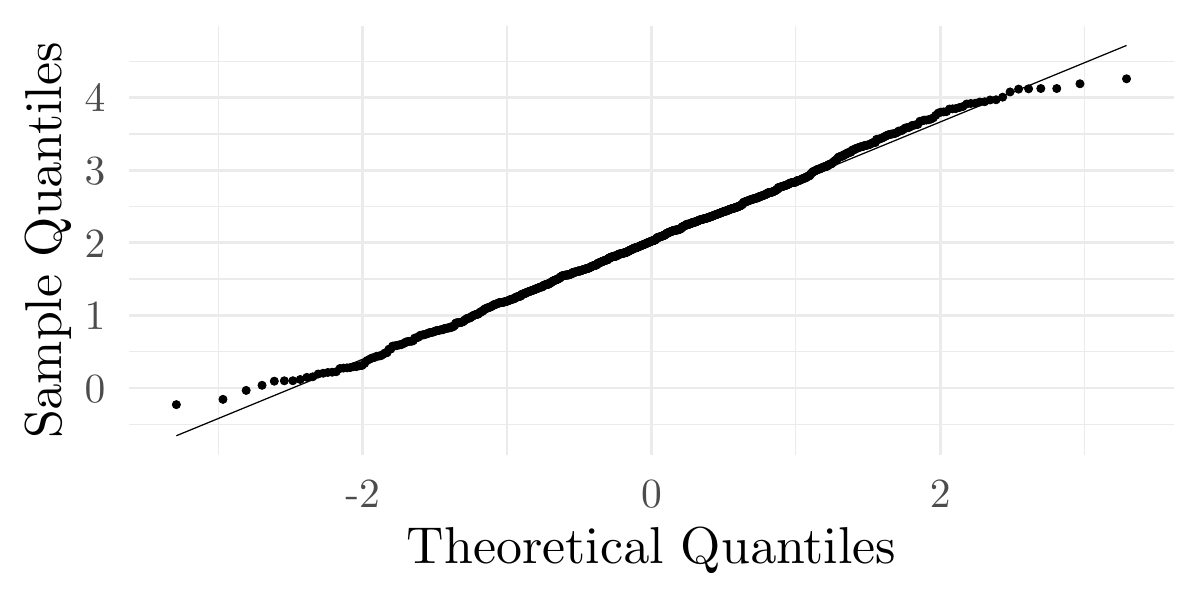}
    \end{subfigure}
    
    \medskip 
    
    \begin{subfigure}[b]{0.4\linewidth}
        \caption{Q-Q plot of $\taushat[1]$}
    \includegraphics[width=\linewidth]{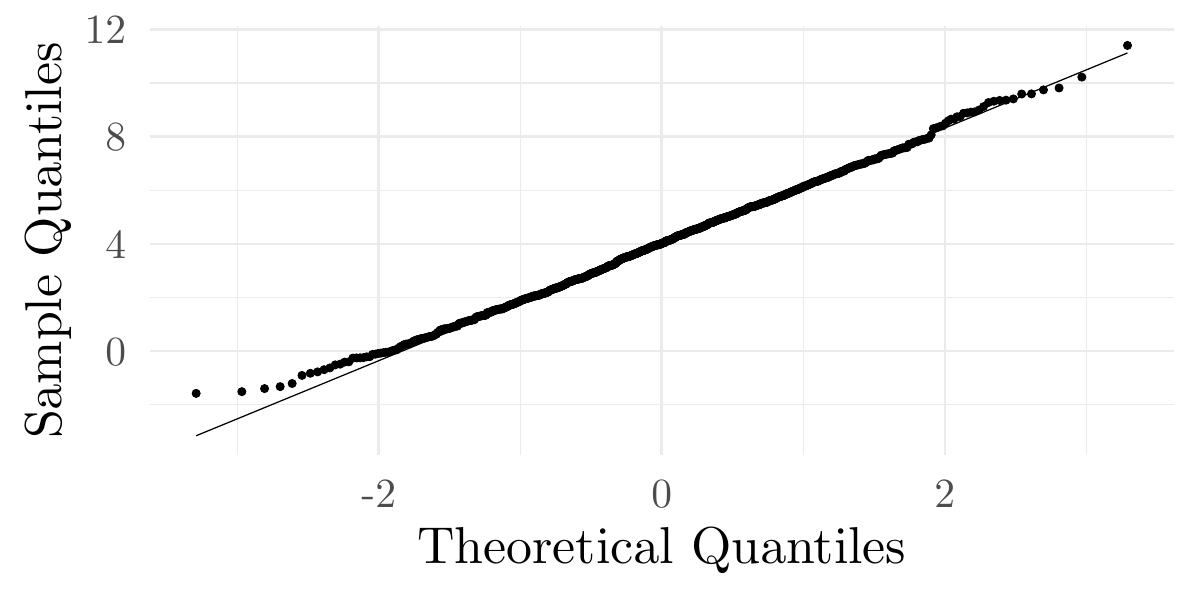}
    \end{subfigure}
    \quad
    \begin{subfigure}[b]{0.4\linewidth}
        \caption{Q-Q plot of $\taushat[0]$}
    \includegraphics[width=\linewidth]{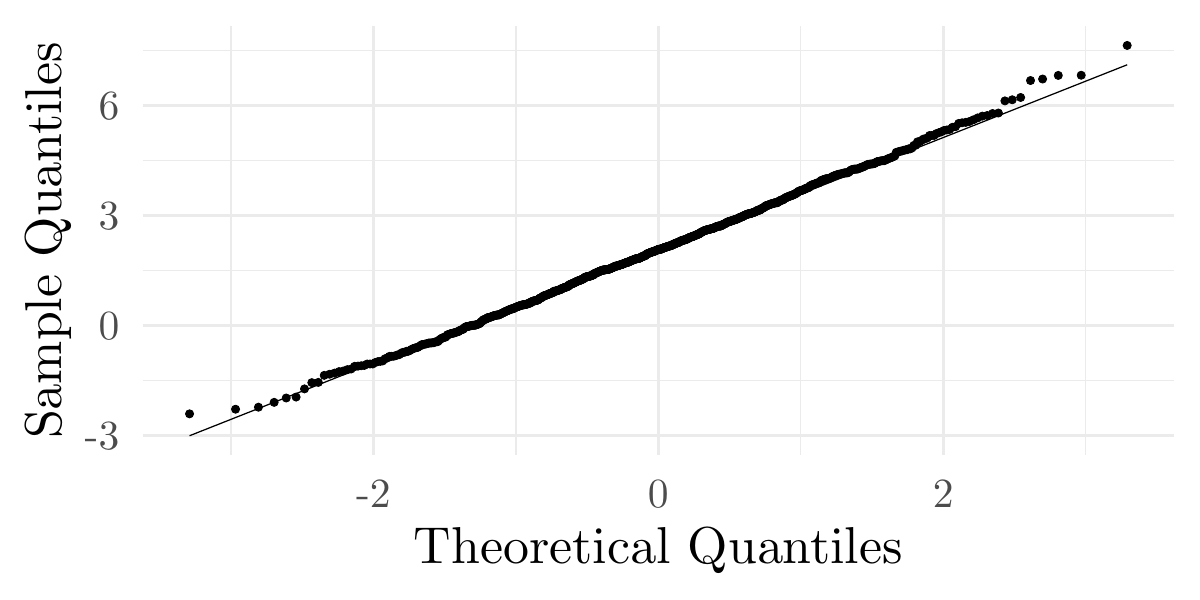}
    \end{subfigure}
    
    \caption{Q-Q plots of the estimators under $\mathbb T^*_1$ when $p=1$, $N=10$ and $T=480$. }\label{fig.mis}
\end{figure}

Table~\ref{table.multi} presents the simulation results in multi-center experiments with $G=48$ and $N_{[g]}=5$ under $\mathbb T^*_1$, showing similar findings to those in single-center experiments. The Horvitz--Thompson estimators exhibit asymptotic normality, and the variance estimation is reliable, with coverage probabilities close to the nominal level of 95\%. These findings further validate the effectiveness of the proposed estimation method in multi-center randomized experiments, particularly in handling both direct and spillover effects under diverse interference patterns.

\begin{table}[!h]
    \centering
    \caption{Simulation results in multi-center randomized experiments under $\mathbb T^*_1$}\label{table.multi}
    \begin{threeparttable}
        \begin{tabular}{ccccccccccc}
            \hline
            $N$ & $T$ & $p$ & Estimand & Value & Bias & $\mathrm{var}(\hat \tau)$ & $\widehat{\mathrm{var}}^U(\hat \tau)$ & CP\\
            \hline
            \multirow{4}{*}{$48\times 5$} & \multirow{4}{*}{120} & \multirow{4}{*}{2} & $\taud[q_1]$ & 6 & 0.00 & 0.38 & 0.37 & 0.945\\
            & & & $\taud[q_2]$ & 3 & -0.01 & 0.15 & 0.15 & 0.941\\
            & & & $\taus[1]$ & 6 & -0.02 & 0.77 & 0.80 & 0.959\\
            & & & $\taus[0]$ & 3 & -0.03 & 0.28 & 0.28 & 0.952\\
            \hline
            \multirow{4}{*}{$48\times 5$} & \multirow{4}{*}{120} & \multirow{4}{*}{3} & $\taud[q_1]$ & 6 & 0.00 & 0.56 & 0.55 & 0.952\\
            & & & $\taud[q_2]$ & 3 & 0.00 & 0.24 & 0.22 & 0.942\\
            & & & $\taus[1]$ & 6 & -0.05 & 1.20 & 1.20 & 0.954\\
            & & & $\taus[0]$ & 3 & -0.04 & 0.40 & 0.42 & 0.946\\
            \hline
            \multirow{4}{*}{$48\times 5$} & \multirow{4}{*}{480} & \multirow{4}{*}{2} & $\taud[q_1]$ & 6 & 0.00 & 0.11 & 0.12 & 0.941\\
            & & & $\taud[q_2]$ & 3 & 0.00 & 0.05 & 0.05 & 0.947\\
            & & & $\taus[1]$ & 6 & -0.02 & 0.26 & 0.25 & 0.939\\
            & & & $\taus[0]$ & 3 & -0.02 & 0.10 & 0.10 & 0.963\\
            \hline
            \multirow{4}{*}{$48\times 5$} & \multirow{4}{*}{480} & \multirow{4}{*}{3} & $\taud[q_1]$ & 6 & 0.01 & 0.16 & 0.18 & 0.965\\
            & & & $\taud[q_2]$ & 3 & 0.00 & 0.08 & 0.08 & 0.957\\
            & & & $\taus[1]$ & 6 & -0.01 & 0.36 & 0.38 & 0.957\\
            & & & $\taus[0]$ & 3 & -0.01 & 0.15 & 0.16 & 0.945\\
            \hline
        \end{tabular}
        \begin{tablenotes}
        \footnotesize
        \item[] Note: Value, true value; CP, coverage probability.
        \end{tablenotes}
    \end{threeparttable}
\end{table}

\begin{figure}[!h]
    \centering
    
    \begin{subfigure}[b]{0.45\linewidth}
        \caption{$N=10$ and $T$ varies (single-center)}
        \label{fig.compare.a}
    \includegraphics[width=\linewidth]{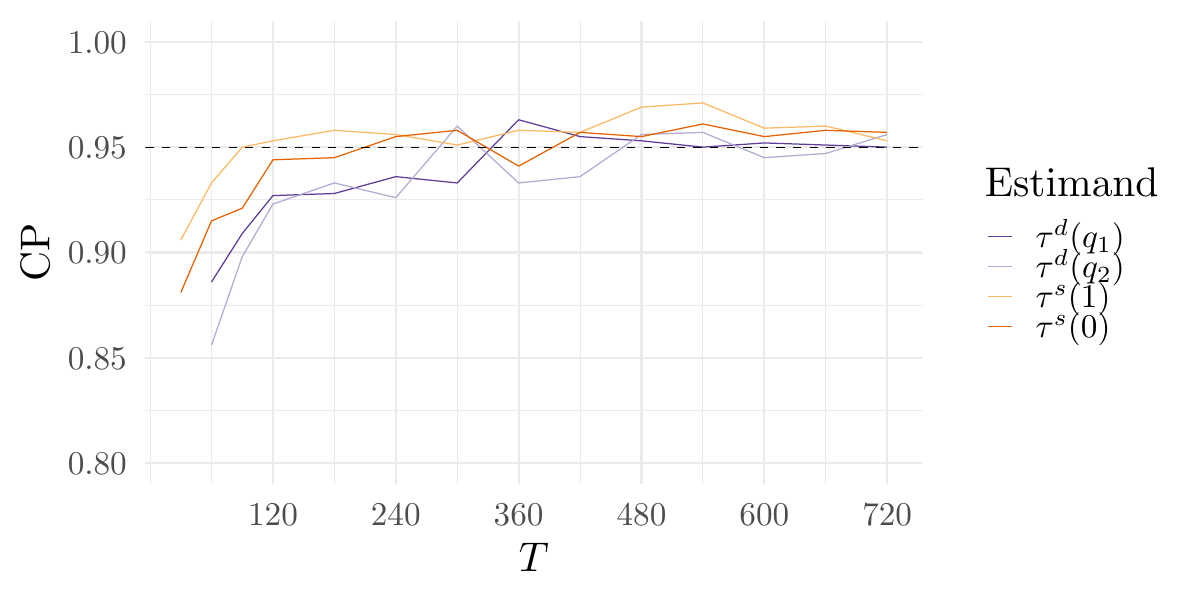}
    \end{subfigure}
    \quad 
    \begin{subfigure}[b]{0.45\linewidth}
        \caption{$T=480$ and $N$ varies (single-center)}
        \label{fig.compare.b}
    \includegraphics[width=\linewidth]{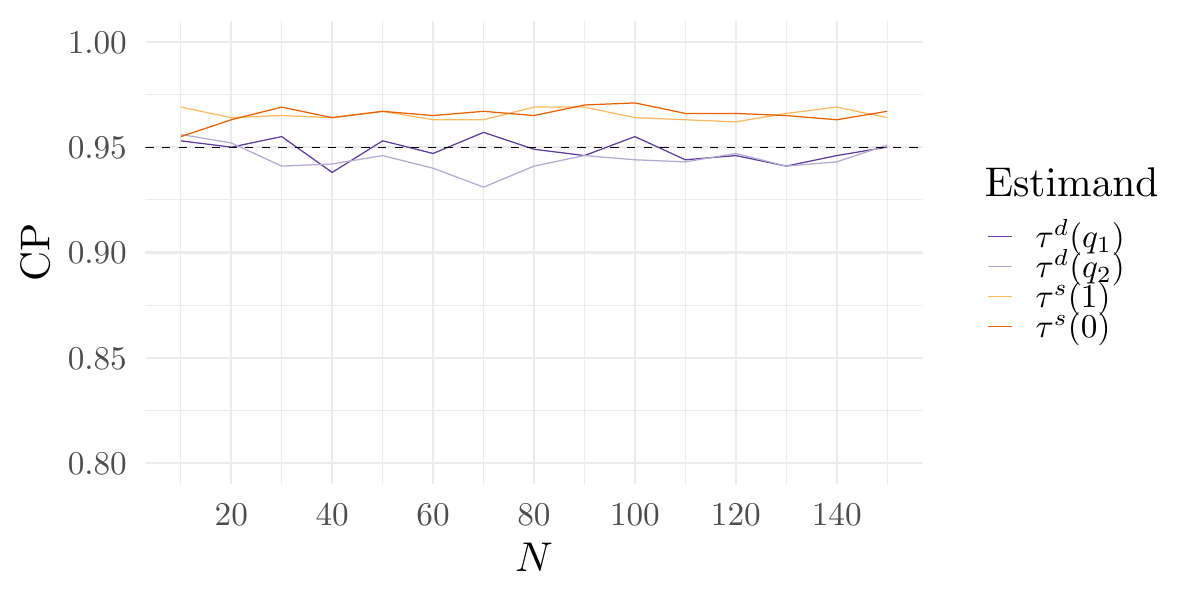}
    \end{subfigure}

    \medskip

    \begin{subfigure}[b]{0.45\linewidth}
        \caption{$N_{[g]}=5$ and $T$ varies (multi-center)}
        \label{fig.compare.c}
    \includegraphics[width=\linewidth]{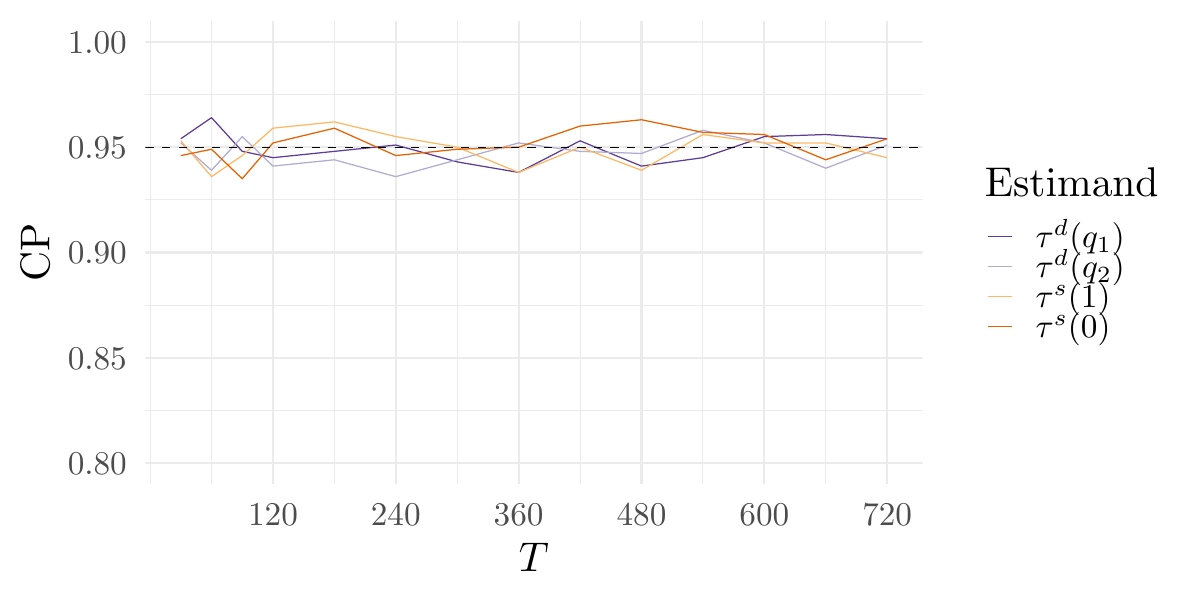}
    \end{subfigure}
    \quad
    \begin{subfigure}[b]{0.45\linewidth}
        \caption{$T=240$ and $N$ varies (multi-center)}
        \label{fig.compare.d}
    \includegraphics[width=\linewidth]{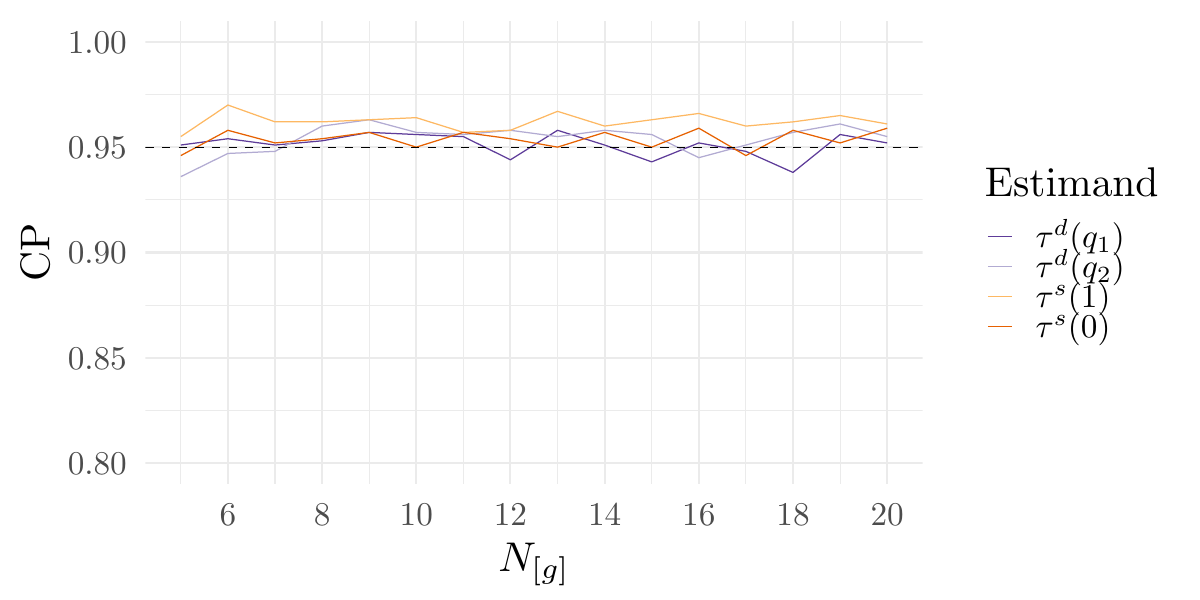}
    \end{subfigure}
    
    \caption{CP under $\mathbb T^*_1$ with different population size when $p=2$. }\label{fig.compare}
\end{figure}

Figure~\ref{fig.compare} compares the coverage probabilities (CPs) under $\mathbb T^*_1$ for different population sizes in single-center and multi-center experiments, providing insights into the impact of Assumption~\ref{assumption.NT}. Key observations include: (1) Multi-center experiments (\Cref{fig.compare.c} and \Cref{fig.compare.d}) exhibit significantly more stable CPs across various combinations of $ N $ and $ T $, indicating their adaptability to different experimental conditions and providing more reliable estimation results. (2) Single-center experiments show poor performance when $ T $ is small (\Cref{fig.compare.a}), with CPs falling far below the nominal level of 0.95. This underscores the limitations of single-center studies conducted over short time periods. (3) In single-center experiments (\Cref{fig.compare.a} and \Cref{fig.compare.b}), longer time periods ($ T $) are more beneficial than larger sample sizes ($ N $) for achieving higher CPs closer to the nominal level. These findings suggest that extending the duration of studies should be prioritized over increasing the number of participants in single-center designs whenever feasible.

\subsection{Order identification}

\begin{figure}[!h]
    \centering
    
    \begin{subfigure}[b]{0.45\linewidth}
        \caption{$N=240$ and $T$ varies (single-center)}
        \label{fig.order.a}
    \includegraphics[width=\linewidth]{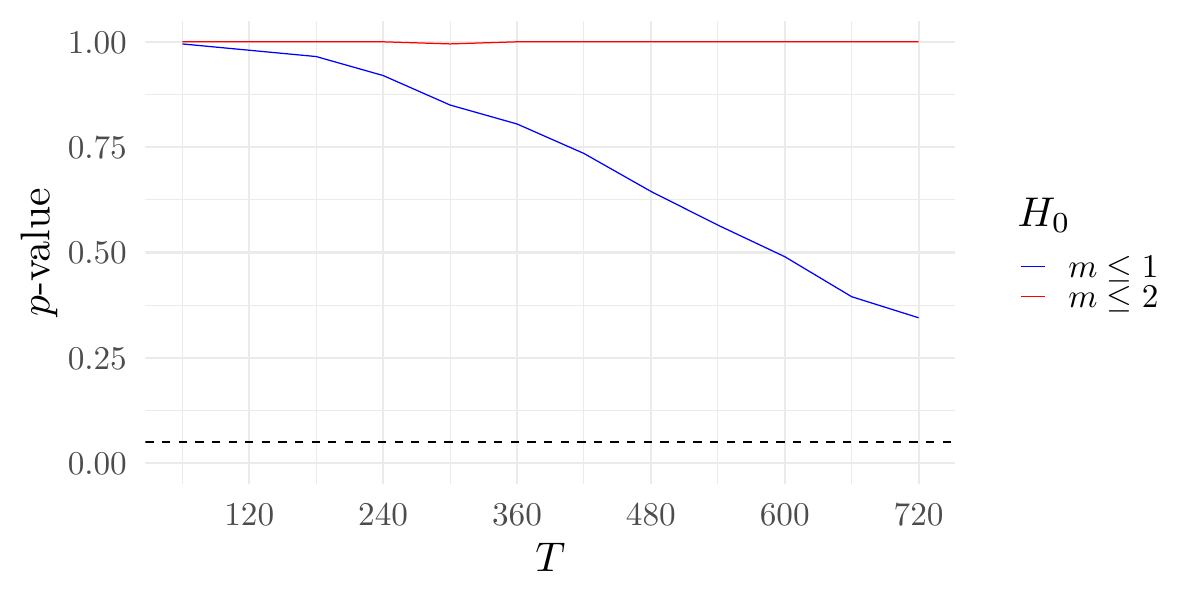}
    \end{subfigure}
    \quad
    \begin{subfigure}[b]{0.45\linewidth}
        \caption{$T=480$ and $N$ varies (single-center)}
        \label{fig.order.b}
    \includegraphics[width=\linewidth]{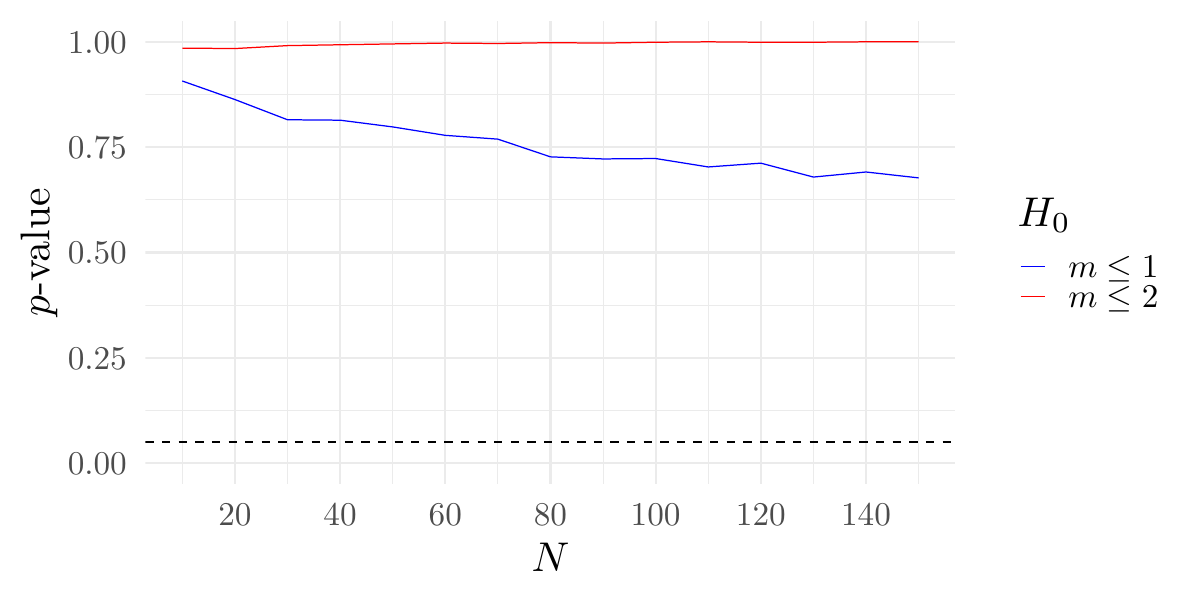}
    \end{subfigure}
    
    \medskip 

    \begin{subfigure}[b]{0.45\linewidth}
        \caption{$N_{[g]}=5$ and $T$ varies (multi-center)}
        \label{fig.order.c}
    \includegraphics[width=\linewidth]{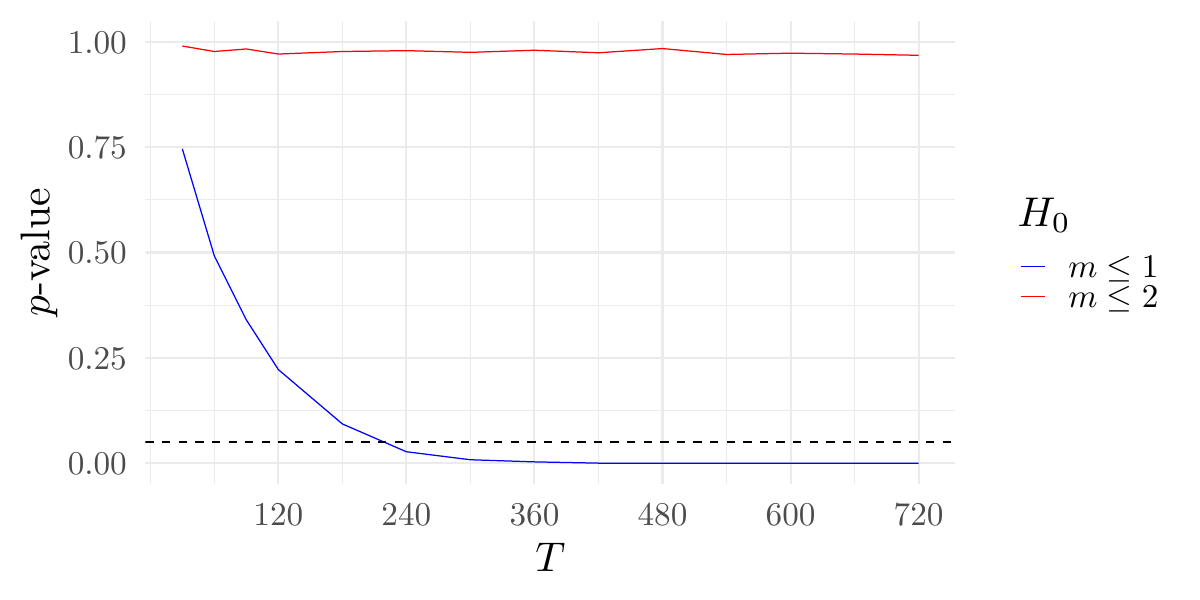}
    \end{subfigure}
    \quad
    \begin{subfigure}[b]{0.45\linewidth}
        \caption{$T=480$ and $N$ varies (multi-center)}
        \label{fig.order.d}
    \includegraphics[width=\linewidth]{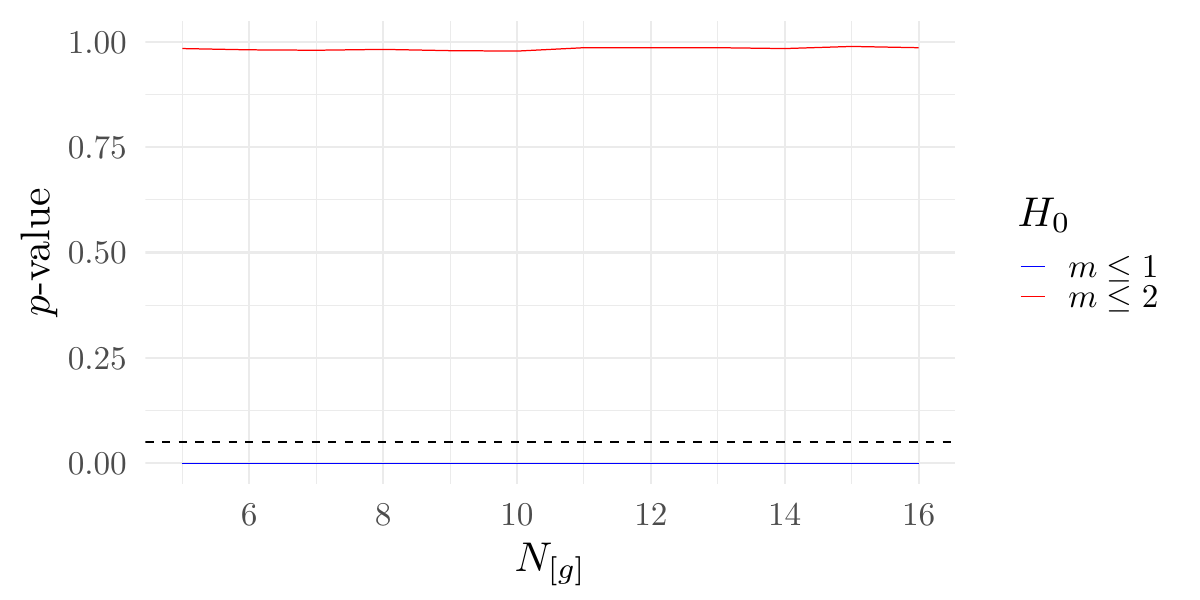}
    \end{subfigure}
    
    \caption{$p$-value from the Wald test for order identification under $\mathbb T^*_1$. }\label{fig.order}
\end{figure}

This section explores the identification of the carryover effects order based on the methods presented in Section~\ref{section.order}. Figure~\ref{fig.order} illustrates the $p$-values obtained from the Wald test for order identification across different experimental setups, where the true order is set at $m=2$. In single-center experiments (\Cref{fig.order.a} and \Cref{fig.order.b}), the $p$-values for the null hypothesis $ H_0: m \leq 2 $ consistently remain above the conventional significance level of 0.05, regardless of changes in $ T $ or $ N $. This result suggests that we cannot reject the hypothesis that the carryover effects order is at most 2. In contrast, for $ H_0: m \leq 1 $, the $p$-values show a decreasing trend as $ T $ or $ N $ increases (as shown in \Cref{fig.order.a} and \Cref{fig.order.b}). This suggests that larger sample sizes enhance the detection of carryover effects order; however, the rate of decline is relatively gradual, indicating the need for more samples to achieve greater statistical power.

Multi-center experiments (\Cref{fig.order.c} and \Cref{fig.order.d}), on the other hand, show a marked improvement in the ability to identify the carryover effects order. For $ H_0: m \leq 1 $, the $p$-values quickly fall below the 0.05 threshold as either $ T $ or $ N $ increases, allowing us to confidently reject the hypothesis that the carryover effects order is at most 1. This finding aligns with the true order of $ m = 2 $.

\section{Application based analysis}
\label{sec:empirical}

In this section, we analyze a real-world data comparing the performance of human traders and algorithmic traders in the stock markets.
The dataset, introduced by \citet{bojinov2019}, consists of trading data from 10 different markets across the US, Europe, and Asia in 2016. 
The experimental process involves two types of traders (human traders and algorithmic trades), labeled as ``A'' and ``B'', whose identities remain confidential due to privacy policies.
Whenever a stock needs to be traded, either trader ``A'' or ``B'' is randomly selected for the transaction. This allocation follows an independent and identically distributed (i.i.d.) Bernoulli process, with probabilities of selecting trader ``A'' set to $q_1=0.75$ or $q_2=0.5$.
Three of the ten markets change the selection probability in the middle of the year. 
The potential outcome of interest is slippage, defined as the difference between the quoted price of a stock and its actual traded price, which can also serve as an indicator of market liquidity. \Cref{assumption.nonanticipativity} and \Cref{assumption.carryover} are reasonable, as current slippage is usually not affected by future trader assignments, and may be influenced by recent trading history. These two assumptions are similarly employed in the analysis conducted by \citet{bojinov2019}, with an additional requirement that outcomes are unaffected by distant trading history. \citet{brogaard2023does} have demonstrated that different proportions of human traders relative to algorithmic traders can affect market quality, including liquidity and price efficiency. Consequently, \Cref{assumption.stratifiedinterference} is reasonable in this context, positing that trading outcomes are influenced by both the type of individual trader (treatment or control) and the ratio of traders, which can be represented by the selection probabilities.

To demonstrate the advantages of the proposed minimax optimal design, we generate a synthetic dataset based on real-world observations. We set $p=2$, which is also considered by \citet{bojinov2019}, and fit a model:
$Y_{i,t}^* = \alpha t + \sum_{\Delta t=0}^p \{\delta_q^{(\Delta t)} I_{i,t-\Delta t}(q=q_1) + \delta_z^{(\Delta t)} f_{i,t-\Delta t}(z=1) + \delta_{q,z}^{(\Delta t)}I_{i,t-\Delta t}(q=q_1)f_{i,t-\Delta t}(z=1)\}$,
where $Y_{i,t}^*$ denotes the slippage in market $i$ at time $t$, $f_{i,t}(z=1)$ represents the proportion of trader ``A'' in market $i$ at time $t$ and $I_{i,t}(q=q_1)$ is an indicator function for whether $q=q_1$ in market $i$ at time $t$.
Although the inclusion of the term $ t $ may seem counterintuitive, it is statistically significant in the model-fitting process. Let $\hat \alpha$, $\hat \delta_q^{(\Delta t)}$, $\hat \delta_z^{(\Delta t)}$, and $ \hat \delta_{q,z}^{(\Delta t)}$ represent the estimators for the coefficients of the fitted model.
In total, there are 247 time points; however, for simplicity, we analyze the last 244 time points in our study.
This choice facilitates the application of typical minimax optimal designs, $\mathbb{T}^*_1$ and $\mathbb{T}^*_2$, which are more straightforward when $T-4p$ is a multiple of $p$ and $T-(4p+2)$ is a multiple of $p+1$, respectively, as indicated by Theorem~\ref{theorem.minimaxdesign} and Corollary~\ref{corollary.minimaxdesign}.
The synthetic dataset is generated according to $\yit{q}{z}{1}= \hat \alpha t + \sum_{\Delta t=0}^p \{ \hat \delta_q^{(\Delta t)} I_{i,t-\Delta t}(q=q_1) + \hat \delta_z^{(\Delta t)} I_{i,t-\Delta t}(z=1) + \hat \delta_{q,z}^{(\Delta t)}I_{i,t-\Delta t}(q=q_1)I_{i,t-\Delta t}(z=1)\}+\varepsilon_{i,t}$, where $\varepsilon_{i,t}\stackrel{\text{i.i.d.}}{\sim}\mathcal N(0,1)$. 
Three objective functions, including $\mathcal{L}(1,0)$, $\mathcal L(0,1)$, and $\mathcal L(0.5,0.5)$, are calculated under the designs $\mathbb{T}^*_1$, $\mathbb{T}^*_2$, $\mathbb{T}^1$, and $\mathbb{T}^2$, as shown in \Cref{table.real}.
The design with the minimum risk is $\mathbb{T}^*_2$, which aligns with the findings presented in \Cref{corollary.minimaxdesign} when $\theta^*\in(1/p,(3p+2)/p^2]$.

\begin{table}[ht]
    \centering
    \caption{Value of the objective function under different designs}\label{table.real}
    \begin{threeparttable}
        \begin{tabular}{ccccccccc} 
            \hline
            $N$ & $T$ & Obj. Func. & $\theta^*$ & $\mathbb T^*$ & $\mathbb T^*_1$ & $\mathbb T^*_2$ & $\mathbb T^1$ & $\mathbb T^2$\\
            \hline
            \multirow{3}{*}{10} & \multirow{3}{*}{244} & $\mathcal L(1,0)$ & 1.763 & $\mathbb T^*_2$ & 0.265 & \textbf{0.225} & 0.981 & 0.239\\
            && $\mathcal L(0,1)$ & 0.902 & $\mathbb T^*_2$ & 0.528 & \textbf{0.510} & 1.230 & 0.537\\
            && $\mathcal L(0.5,0.5)$ & 1.221 & $\mathbb T^*_2$ & 0.396 & \textbf{0.368} & 1.106 & 0.388\\
            \hline
        \end{tabular}
        \begin{tablenotes}
        \footnotesize
        \item[] Note: ``Obj. Func.'' is short for ``Objective Function''.
        \end{tablenotes}
    \end{threeparttable}
\end{table}

Additionally, Table~\ref{table.finance} evaluates the normality of the Horvitz--Thompson estimators under the minimax optimal design $ \mathbb{T}^*_2$. The results confirm the consistency of the estimators, with negligible bias and conservative variance estimators. CPs are close to the nominal level, supporting the validity of the inferences.

\begin{table}[ht]
    \centering
    \caption{Simulation results in application based analysis}\label{table.finance}
    \begin{threeparttable}
        \begin{tabular}{cccccc}
            \hline
            Estimand & Value & Bias & $\mathrm{var}(\hat \tau)$ & $\widehat{\mathrm{var}}^U(\hat \tau)$ & CP \\
            \hline
            $\taud[0.75]$ & -0.994 & -0.009 & 0.143 & 0.154 & 0.932 \\
            $\taud[0.5]$ & -0.737 & -0.007 & 0.089 & 0.087 & 0.930 \\
            $\taus[1]$ & -0.547 & 0.007 & 0.108 & 0.118 & 0.946\\
            $\taus[0]$ & -0.290 & 0.009 & 0.416 & 0.396 & 0.956\\
            \hline
        \end{tabular}
        \begin{tablenotes}
        \footnotesize
        \item[] Note: The reported expectations and risks are based on 1,000 replications. Value, true value; CP, coverage probability.
        \end{tablenotes}
    \end{threeparttable}
\end{table}

\section{Conclusion}\label{sec:conclusion}


In conclusion, we propose a novel minimax optimal design that demonstrates significant potential across various domains of policy evaluation. The design can effectively account for spillover and carryover effects in diverse scenarios: from evaluating vaccination programs with herd immunity effects in public health, to optimizing incentive strategies on digital platforms, as well as assessing educational interventions where peer effects and learning persistence are crucial. Its ability to adapt to misspecified carryover effects makes it particularly valuable in these real-world applications.

From a theoretical perspective, our design enhances the estimation of direct and spillover effects by optimizing treatment time points to improve the reliability of causal effect estimations under worst-case scenarios. We develop a polynomial-time algorithm to implement this optimal design, making it computationally feasible for large-scale applications. Moreover, we establish the consistency and asymptotic normality of the Horvitz--Thompson estimators for both direct and spillover effects, providing theoretical guarantees for the design's performance. This combination of computational efficiency and theoretical robustness positions our design as a practical tool for complex policy evaluations. Future research could explore optimal designs by solving deterministic or stochastic optimization problems introduced by \citet{zhao2024}, incorporate time-varying covariates to improve efficiency, and investigate strategies to relax the assumptions related to stratified interference.

\section*{Supplementary Material}

The Supplementary Material provides the specific forms of variances and their estimators, additional simulation results, and proofs of the theoretical results in the main text.

\bibliographystyle{apalike.bst}



\newpage

\setcounter{page}{1}
\renewcommand{\thepage}{S\arabic{page}}
\setcounter{table}{0}
\renewcommand{\thesection}{\Alph{section}}
\setcounter{section}{0}
\renewcommand{\theassumption}{S\arabic{assumption}}
\setcounter{assumption}{0}
\renewcommand{\thetheorem}{S\arabic{theorem}}
\setcounter{theorem}{0}
\renewcommand{\theproposition}{S\arabic{proposition}}
\setcounter{proposition}{0}
\renewcommand{\thecondition}{S\arabic{condition}}
\setcounter{condition}{0}
\renewcommand{\thedefinition}{S\arabic{definition}}
\setcounter{definition}{0}
\renewcommand{\thelemma}{S\arabic{lemma}}
\setcounter{lemma}{0}
\renewcommand{\thecorollary}{S\arabic{corollary}}
\setcounter{corollary}{0}
\renewcommand{\thealgorithm}{S\arabic{algorithm}}
\setcounter{corollary}{0}
\renewcommand{\thefigure}{S\arabic{figure}}
\setcounter{figure}{0}
\renewcommand{\thetable}{S\arabic{table}}
\setcounter{table}{0}
\renewcommand{\theequation}{S\arabic{equation}}
\setcounter{equation}{0}

\begin{center}
\LARGE\bf Supplementary Material for ``Minimax Optimal Design with Spillover and Carryover Effects''
\end{center}

The Supplementary Material is organized as follows:

\Cref{sec:sm.variance} provides variances of the Horvitz--Thompson estimators under the minimax optimal design $\mathbb T^*$, their upper bounds, and their conservative estimators.

\Cref{sec:B} provides the exact values of the optimal $\R[q_1]$ under varying scenarios.

\Cref{sec:sm.theorem.r}--\ref{sec:sm.theorem.misspecified} provide proofs for the main theoretical results, along with a polynomial--time algorithm to identify the minimax optimal design. Specifically,
\Cref{sec:sm.theorem.r} presents the proof for \Cref{theorem.r};
\Cref{sec:sm.theorem.minimaxdesign} presents the proof for \Cref{theorem.minimaxdesign};
\Cref{sec:algorithm} introduces a polynomial--time algorithm to identify the minimax optimal design;
\Cref{sec:proof-lemma.optimal} presents the proof for \Cref{corollary.minimaxdesign};
\Cref{sec:sm.theorem.variance.sm} presents the proof for \Cref{theorem.variance.sm};
\Cref{sec:sm.corollary.upperbound} presents the proof for \Cref{corollary.upperbound};
\Cref{sec:sm.corollary.estimator} presents the proof for \Cref{corollary.estimator};
\Cref{sec:sm.theorem.CLT} presents the proof for \Cref{theorem.CLT};
\Cref{sec:sm.theorem.misspecified} presents the proof for \Cref{theorem.misspecified}. 

\Cref{sec:sm.sim} provides additional simulation results.


\section{Variances and their conservative estimators}
\label{sec:sm.variance}

Recall that under minimax optimal design $\mathbb T^*$, $(T-2\aop)/\bop=K-4\geq 0$ with $K\in\mathbb N$; the decision points include $t_1=1$, $t_k=\aop+1+(k-2)\bop$ for $k=1,\ldots,K-2$, and $t_{K-1}=T+1$.
Define
\begin{align*}
    \tilde Y_{i,1}(q\mathbf 1,z\mathbf 1)=\sum_{t=p+1}^{t_2-1}\yit{q}{z}{1},&\  
    \yik{q}{z}{1}=\sum_{t=t_k}^{t_{k+1}-1}\yit{q}{z}{1},\ \text{for}\ k=2,\ldots,K-2,\\
    \check Y_{i,1}(q\mathbf 1,z\mathbf 1)=0,&\
    \yikc{q}{z}{1}=\sum_{t=t_k}^{t_{k}+p-1}\yit{q}{z}{1},\ \text{for}\ k=2,\ldots,K-2,\\
    \yk{q}{z}{1}&=[\tilde Y_{1,k}(q\mathbf 1,z\mathbf 1),\ldots,\tilde Y_{N,k}(q\mathbf 1,z\mathbf 1)]^\top,\\
    \ykc{q}{z}{1}&=[\check Y_{1,k}(q\mathbf 1,z\mathbf 1),\ldots,\check Y_{N,k}(q\mathbf 1,z\mathbf 1)]^\top.
\end{align*}
Notably, $\yik{q}{z}{1}$ computes the sum of potential outcomes during the time period $[t_k,t_{k+1}-1]$, excluding the period $[1,p]$; whereas $\yikc{q}{z}{1}$ computes the sum of potential outcomes during the time period $[t_k,t_{k}+p-1]$, also excluding the period $[1,p]$. 
Denote $\eta_k=1$ when $k=1$ and $\eta_k=2$ when $k=2,\ldots,K-2$. Let $\mathbf J_N$ denote the $N \times N$ matrix of ones, $\mathbf I_N$ denote the $N \times N$ identity matrix and $\bar{q} = 1 - q$. Additionally, we define $q_{1,z} = zq_1 + (1 - z)(1 - q_1)$ and $q_{2,z} = zq_2 + (1 - z)(1 - q_2)$. Let $q_{1,1}$ and $q_{1,0}$ represent the probabilities of assigning a unit to the treatment and control, respectively, under the treated probability $q_1$. Similar arguments apply to $q_{2,1}$ and $q_{2,0}$. 

\begin{theorem}[Variances of the Horvitz--Thompson estimators under the minimax optimal design proposed in \Cref{corollary.minimaxdesign}]
    \label{theorem.variance.sm}
    Under Assumptions~\ref{assumption.nonanticipativity}--\ref{assumption.bounded}, $\R[q_1]=\R[q_2]=0.5$, and the minimax optimal design $\mathbb T^*$, the variances of $\taudhat$ and $\taushat$ are given by
    \begin{align*}
        \mathrm{var}\{\taudhat\} &= \frac{1}{N^2 (T-p)^2}(A^d + B^d + C^d + D^d + E^d + F^d + G^d), \ q = q_1, q_2,\\
        \mathrm{var}\{\taushat\} &= \frac{1}{N^2 (T-p)^2}(A^s + B^s + C^s + D^s + E^s + F^s + G^s), \ z=0,1,
    \end{align*}
    where
    \begin{align*}
        A^d&=\sum_{k=1}^{K-2}\{\yk{q}{}{1}\}^{\top}\{\mathbf J_N+2(q^{-1}-1)\mathbf I_N\}\{\yk{q}{}{1}\}\\
        &\quad +\sum_{k=2}^{K-2}\{\ykc{q}{}{1}\}^{\top}\{2\mathbf J_N+2(q^{-1}-1)(2q^{-1}+1)\mathbf I_N\}\{\ykc{q}{}{1}\},\\
        B^d&=\sum_{k=1}^{K-2}\{\yk{q}{}{0}\}^{\top}\{\mathbf J_N+2(\bar q^{-1}-1)\mathbf I_N\}\{\yk{q}{}{0}\}\\
        &\quad +\sum_{k=2}^{K-2}\{\ykc{q}{}{0}\}^{\top}\{2\mathbf J_N+2(\bar q^{-1}-1)(2\bar q^{-1}+1)\mathbf I_N\}\{\ykc{q}{}{0}\},\\
        C^d&=\sum_{k=1}^{K-2}\{\yk{q}{}{1}\}^{\top}\{-2\mathbf J_N+4\mathbf I_N\}\{\yk{q}{}{0}\}\\
        &\quad +\sum_{k=2}^{K-2}\{\ykc{q}{}{1}\}^{\top}\{-4\mathbf J_N+4\mathbf I_N\}\{\ykc{q}{}{0}\},\\
        D^d&=2\sum_{k=1}^{K-3}\{\yk{q}{}{1}\}^{\top}\{\mathbf J_N+2(q^{-1}-1)\mathbf I_N\}\{\ykcplus{q}{}{1}\},\\
        E^d&=2\sum_{k=1}^{K-3}\{\yk{q}{}{0}\}^{\top}\{\mathbf J_N+2(\bar q^{-1}-1)\mathbf I_N\}\{\ykcplus{q}{}{0}\},\\
        F^d&=2\sum_{k=1}^{K-3}\{\yk{q}{}{1}\}^{\top}\{-\mathbf J_N+2\mathbf I_N\}\{\ykcplus{q}{}{0}\},\\
        G^d&=2\sum_{k=1}^{K-3}\{\yk{q}{}{0}\}^{\top}\{-\mathbf J_N+2\mathbf I_N\}\{\ykcplus{q}{}{1}\},
    \end{align*}
    \begin{align*}
        A^s&=\sum_{k=1}^{K-2}\{\yk{q_1}{z}{1}\}^{\top}\{\mathbf J_N+2(q_{1,z}^{-1}-1)\mathbf I_N\}\{\yk{q_1}{z}{1}\}\\
        &\quad +\sum_{k=2}^{K-2}\{\ykc{q_1}{z}{1}\}^{\top}\{2\mathbf J_N+2(q_{1,z}^{-1}-1)(2q_{1,z}^{-1}+1)\mathbf I_N\}\{\ykc{q_1}{z}{1}\},\\
        B^s&=\sum_{k=1}^{K-2}\{\yk{q_2}{z}{1}\}^{\top}\{\mathbf J_N+2(q_{2,z}^{-1}-1)\mathbf I_N\}\{\yk{q_2}{z}{1}\},\\
        &\quad +\sum_{k=2}^{K-2}\{\ykc{q_2}{z}{1}\}^{\top}\{2\mathbf J_N+2(q_{2,z}^{-1}-1)(2q_{2,z}^{-1}+1)\mathbf I_N\}\{\ykc{q_2}{z}{1}\},\\
        C^s&=\sum_{k=1}^{K-2}\{\yk{q_1}{z}{1}\}^{\top}\{2\mathbf J_N\}\{\yk{q_2}{z}{1}\},\\
        D^s&=2\sum_{k=1}^{K-3}\{\yk{q_1}{z}{1}\}^{\top}\{\mathbf J_N+2(q_{1,z}^{-1}-1)\mathbf I_N\}\{\ykcplus{q_1}{z}{1}\},\\
        E^s&=2\sum_{k=1}^{K-3}\{\yk{q_2}{z}{1}\}^{\top}\{\mathbf J_N+2(q_{2,z}^{-1}-1)\mathbf I_N\}\{\ykcplus{q_2}{z}{1}\},\\
        F^s&=2\sum_{k=1}^{K-3}\{\yk{q_1}{z}{1}\}^{\top}\{\mathbf J_N\}\{\ykcplus{q_2}{z}{1}\},\\
        G^s&=2\sum_{k=1}^{K-3}\{\yk{q_2}{z}{1}\}^{\top}\{\mathbf J_N\}\{\ykcplus{q_1}{z}{1}\}.
    \end{align*}
\end{theorem}

Notably, the expressions for variances can be applied to all minimax optimal designs discussed in \Cref{corollary.minimaxdesign}.
Some of the terms in Theorem~\ref{theorem.variance.sm} are challenging to estimate directly because we cannot observe all potential outcomes simultaneously. To address this issue, we propose upper bounds for the variances based on the Cauchy--Schwarz inequality, denoted as $ \mathrm{var}^U\{\taudhat\} $ and $ \mathrm{var}^U\{\taushat\} $, which serve as upper bounds for $ \mathrm{var}\{\taudhat\} $ and $ \mathrm{var}\{\taushat\} $, respectively. These upper bounds can be unbiasedly estimable and calculated using the observed data. The following corollary provides the specific forms of these upper bounds.

\begin{corollary}[Upper bounds of variances]
    \label{corollary.upperbound}
    Under Assumptions~\ref{assumption.nonanticipativity}--\ref{assumption.bounded}, $\R[q_1]=\R[q_2]=0.5$, and the minimax optimal design $\mathbb T^*$, an estimable upper bound of $\mathrm{var}\{\taudhat\}$ is
    \begin{align*}
        \mathrm{var}^U\{\taudhat\}=\frac 1{N^2(T-p)^2}(A^d + B^d + \bar C^d + D^d + E^d + \bar F^d + \bar G^d),
    \end{align*}
    and an estimable upper bound of $\mathrm{var}\{\taushat\}$ is 
    \begin{align*}
        \mathrm{var}^U\{\taushat\}=\frac 1{N^2(T-p)^2}(A^s + B^s + \bar C^s + D^s + E^s + \bar F^s + \bar G^s),
    \end{align*}
    where $A^d$, $B^d$, $D^d$, $E^d$, $A^s$, $B^s$, $D^s$ and $E^s$ are defined in Theorem~\ref{theorem.variance.sm}, and $\bar C^d$, $\bar F^d$, $\bar G^d$, $\bar C^s$, $\bar F^s$ and $\bar G^s$ are defined as follows:
    \begin{align*}
        \bar C^d&= \sum_{k=1}^{K-2}\{\yk{q}{}{1}\}^{\top}\{-2(\mathbf J_N-\mathbf I_N)\}\{\yk{q}{}{0}\}\\
        &\quad +\sum_{k=2}^{K-2}\{\ykc{q}{}{1}\}^{\top}\{-4(\mathbf J_N-\mathbf I_N)\}\{\ykc{q}{}{0}\}\\
        &\quad +\sum_{k=1}^{K-2}\{\yk{q}{}{1}\}^{\top}\{\mathbf I_N\}\{\yk{q}{}{1}\}+\sum_{k=1}^{K-2}\{\yk{q}{}{0}\}^{\top}\{\mathbf I_N\}\{\yk{q}{}{0}\},\\
        \bar F^d&= 2\sum_{k=1}^{K-3}\{\yk{q}{}{1}\}^{\top}\{-(\mathbf J_N-\mathbf I_N)\}\{\ykcplus{q}{}{0}\}\\
        &\quad +\sum_{k=1}^{K-3}\{\yk{q}{}{1}\}^{\top}\{\mathbf I_N\}\{\yk{q}{}{1}\}+\sum_{k=2}^{K-2}\{\ykc{q}{}{0}\}^{\top}\{\mathbf I_N\}\{\ykc{q}{}{0}\},\\
        \bar G^d&= 2\sum_{k=1}^{K-3}\{\yk{q}{}{0}\}^{\top}\{-(\mathbf J_N-\mathbf I_N)\}\{\ykcplus{q}{}{1}\}\\
        &\quad +\sum_{k=2}^{K-2}\{\ykc{q}{}{1}\}^{\top}\{\mathbf I_N\}\{\ykc{q}{}{1}\}+\sum_{k=1}^{K-3}\{\yk{q}{}{0}\}^{\top}\{\mathbf I_N\}\{\yk{q}{}{0}\}.
    \end{align*}
    \begin{align*}
        \bar C^s&= \sum_{k=1}^{K-2}\{\yk{q_1}{z}{1}\}^{\top}\{\mathbf J_N\}\{\yk{q_1}{z}{1}\}+\sum_{k=1}^{K-2}\{\yk{q_2}{z}{1}\}^{\top}\{\mathbf J_N\}\{\yk{q_2}{z}{1}\},\\
        \bar F^s&=\sum_{k=1}^{K-3}\{\yk{q_1}{z}{1}\}^{\top}\{\mathbf J_N\}\{\yk{q_1}{z}{1}\}+\sum_{k=2}^{K-2}\{\ykc{q_2}{z}{1}\}^{\top}\{\mathbf J_N\}\{\ykc{q_2}{z}{1}\},\\
        \bar G^s&=\sum_{k=2}^{K-2}\{\ykc{q_1}{z}{1}\}^{\top}\{\mathbf J_N\}\{\ykc{q_1}{z}{1}\}+\sum_{k=1}^{K-3}\{\yk{q_2}{z}{1}\}^{\top}\{\mathbf J_N\}\{\yk{q_2}{z}{1}\}.
    \end{align*}
\end{corollary}

\Cref{corollary.upperbound} is based on Theorem~\ref{theorem.variance.sm}. In the expression for $\mathrm{var}\{\taudhat\}$, the terms $A^d$, $B^d$, $D^d$, and $E^d$ can be unbiasedly estimated, while $C^d$, $F^d$, and $G^d$ can only be conservatively estimated using $\hat{\bar C}^d$, $\hat{\bar F}^d$, and $\hat{\bar G}^d$. A similar result holds for $\mathrm{var}\{\taushat\}$. \Cref{corollary.estimator} provides the explicit formulas for the conservative variance estimators.

\begin{theorem}[Conservative variance estimators]
    \label{corollary.estimator}
    Under Assumptions~\ref{assumption.nonanticipativity}--\ref{assumption.bounded}, $\R[q_1]=\R[q_2]=0.5$, and the minimax optimal design $\mathbb T^*$, the asymptotically conservative variance estimators of $\mathrm{var}\{\taudhat\}$ and $\mathrm{var}\{\taushat\}$ are given by
    \begin{align*}
        \widehat{\mathrm{var}}^U\{\taudhat\}&=\frac 1{N^2(T-p)^2}(\hat A^d + \hat B^d + \hat {\bar C}^d +\hat  D^d +\hat  E^d + \hat {\bar F}^d + \hat {\bar G}^d),\\
        \widehat{\mathrm{var}}^U\{\taushat\}&=\frac 1{N^2(T-p)^2}(\hat A^s +\hat  B^s + \hat {\bar C}^s +\hat  D^s +\hat  E^s + \hat {\bar F}^s + \hat {\bar G}^s).
    \end{align*}
    Define $\Ik{q}{z}{1}$ as the $N \times 1$ matrix whose $i$-th element is given by $\Iik{q}{z}{1} = I_{i,t_{k+1}-1}(q \mathbf{1}, z \mathbf{1})$. Define $\Ikc{q}{z}{1}$ as the $N \times 1$ matrix whose $i$-th element is given by $\check I_{i,1}(q\mathbf 1,z\mathbf 1) = I_{i,p+1}(q \mathbf{1}, z \mathbf{1})$ and $\Iikc{q}{z}{1} = I_{i,t_k}(q \mathbf{1}, z \mathbf{1})$ for $k \geq 2$. Let $\ykd{q}{}{1} = \yk{q}{}{1} - \ykc{q}{}{1}$. Note that $\circ$ denotes the Hadamard product, and we define the observed version $\ykco = \sum_{q=q_1, q_2; z=0, 1} \{\ykc{q}{z}{1} \circ \Ikc{q}{z}{1}\}$ and $\ykdo = \sum_{q=q_1, q_2; z=0, 1} \{\ykd{q}{z}{1} \circ \Ik{q}{z}{1}\}$. Then we can obtain
    
    \begin{align*}
        \hat A^d&= \sum_{k=1}^{K-2}\{\ykco\circ \Ikc{q}{}{1}\}^{\top}\left\{\frac{3(\mathbf J_N-\mathbf I_N)}{2^{-\eta_k}q^{2\eta_k}}+\frac{(4q^{-2}-1)\mathbf I_N}{2^{-\eta_k}q^{\eta_k}}\right\}\{\ykco\circ \Ikc{q}{}{1}\}\\
        & + \sum_{k=1}^{K-2}\{\ykdo\circ \Ik{q}{}{1}\}^{\top}\left\{\frac{2(\mathbf J_N-\mathbf I_N)}{2^{-\eta_k}q^{\eta_k+1}}+\frac{(4q^{-1}-2)\mathbf I_N}{2^{-\eta_k}q^{\eta_k}}\right\}\{\ykco\circ \Ikc{q}{}{1}\} \\
        & + \sum_{k=1}^{K-2}\{\ykdo\circ \Ik{q}{}{1}\}^{\top}\left\{\frac{\mathbf J_N-\mathbf I_N}{2^{-1}q^{2}}+\frac{(2q^{-1}-1)\mathbf I_N}{2^{-1}q}\right\}\{\ykdo\circ \Ik{q}{}{1}\},\\
        \hat B^d&= \sum_{k=1}^{K-2}\{\ykco\circ \Ikc{q}{}{0}\}^{\top}\left\{\frac{3(\mathbf J_N-\mathbf I_N)}{2^{-\eta_k}\bar q^{2\eta_k}}+\frac{(4\bar q^{-2}-1)\mathbf I_N}{2^{-\eta_k}\bar q^{\eta_k}}\right\}\{\ykco\circ \Ikc{q}{}{0}\}\\
        & + \sum_{k=1}^{K-2}\{\ykdo\circ \Ik{q}{}{0}\}^{\top}\left\{\frac{2(\mathbf J_N-\mathbf I_N)}{2^{-\eta_k}\bar q^{\eta_k+1}}+\frac{(4\bar q^{-1}-2)\mathbf I_N}{2^{-\eta_k}\bar q^{\eta_k}}\right\}\{\ykco\circ \Ikc{q}{}{0}\} \\
        & + \sum_{k=1}^{K-2}\{\ykdo\circ \Ik{q}{}{0}\}^{\top}\left\{\frac{\mathbf J_N-\mathbf I_N}{2^{-1}\bar q^{2}}+\frac{(2\bar q^{-1}-1)\mathbf I_N}{2^{-1}\bar q}\right\}\{\ykdo\circ \Ik{q}{}{0}\},\\
        \hat {\bar C}^d&= \sum_{k=1}^{K-2}\{\ykco\circ \Ikc{q}{}{1}\}^{\top}\left\{\frac{-6(\mathbf J_N-\mathbf I_N)}{2^{-\eta_k}q^{\eta_k}\bar q^{\eta_k}}\right\}\{\ykco\circ \Ikc{q}{}{0}\}\\
        &+\sum_{k=1}^{K-2}\{\ykco\circ \Ikc{q}{}{1}\}^{\top}\left\{\frac{-2(\mathbf J_N-\mathbf I_N)}{2^{-\eta_k}q^{\eta_k}\bar q}\right\}\{\ykdo\circ \Ik{q}{}{0}\}\\
        &+\sum_{k=1}^{K-2}\{\ykdo\circ \Ik{q}{}{1}\}^{\top}\left\{\frac{-2(\mathbf J_N-\mathbf I_N)}{2^{-\eta_k}q\bar q^{\eta_k}}\right\}\{\ykco\circ \Ikc{q}{}{0}\}\\
        &+\sum_{k=1}^{K-2}\{\ykdo\circ \Ik{q}{}{1}\}^{\top}\left\{\frac{-2(\mathbf J_N-\mathbf I_N)}{2^{-1}q\bar q}\right\}\{\ykdo\circ \Ik{q}{}{0}\}\\
        &+\sum_{k=1}^{K-2}\{\ykco\circ \Ikc{q}{}{1}\}^{\top}\left\{\frac{\mathbf I_N}{2^{-\eta_k}q^{\eta_k}}\right\}\{\ykco\circ \Ikc{q}{}{1}\}\\
        &+\sum_{k=1}^{K-2}\{\ykdo\circ \Ik{q}{}{1}\}^{\top}\left\{\frac{2\mathbf I_N}{2^{-\eta_k}q^{\eta_k}}\right\}\{\ykco\circ \Ikc{q}{}{1}\}\\
        &+\sum_{k=1}^{K-2}\{\ykdo\circ \Ik{q}{}{1}\}^{\top}\left\{\frac{\mathbf I_N}{2^{-1}q}\right\}\{\ykdo\circ \Ik{q}{}{1}\}\\
        &+\sum_{k=1}^{K-2}\{\ykco\circ \Ikc{q}{}{0}\}^{\top}\left\{\frac{\mathbf I_N}{2^{-\eta_k}\bar q^{\eta_k}}\right\}\{\ykco\circ \Ikc{q}{}{0}\}\\
        &+\sum_{k=1}^{K-2}\{\ykdo\circ \Ik{q}{}{0}\}^{\top}\left\{\frac{2\mathbf I_N}{2^{-\eta_k}\bar q^{\eta_k}}\right\}\{\ykco\circ \Ikc{q}{}{0}\}\\
        &+\sum_{k=1}^{K-2}\{\ykdo\circ \Ik{q}{}{0}\}^{\top}\left\{\frac{\mathbf I_N}{2^{-1}\bar q}\right\}\{\ykdo\circ \Ik{q}{}{0}\},\\
    \end{align*}
    \begin{align*}
        \hat D^d&=2\sum_{k=1}^{K-3}\{\ykco\circ\Ikc{q}{}{1}\}^{\top}\left\{\frac{\mathbf J_N-\mathbf I_N}{2^{-\eta_k-1}q^{\eta_k+\eta_{k+1}}}+\frac{(2q^{-1}-1)\mathbf I_N}{2^{-\eta_k-1}q^{\eta_k+1}}\right\}\{\ykcoplus\circ\Ikcplus{q}{}{1}\}\\
        &+2\sum_{k=1}^{K-3}\{\ykdo\circ\Ik{q}{}{1}\}^{\top}\left\{\frac{\mathbf J_N-\mathbf I_N}{2^{-2}q^{3}}+\frac{(2q^{-1}-1)\mathbf I_N}{2^{-2}q^{2}}\right\}\{\ykcoplus\circ\Ikcplus{q}{}{1}\},\\
        \hat E^d&=2\sum_{k=1}^{K-3}\{\ykco\circ\Ikc{q}{}{0}\}^{\top}\left\{\frac{\mathbf J_N-\mathbf I_N}{2^{-\eta_k-1}\bar q^{\eta_k+\eta_{k+1}}}+\frac{(2\bar q^{-1}-1)\mathbf I_N}{2^{-\eta_k-1}\bar q^{\eta_k+1}}\right\}\{\ykcoplus\circ\Ikcplus{q}{}{0}\}\\
        &+2\sum_{k=1}^{K-3}\{\ykdo\circ\Ik{q}{}{0}\}^{\top}\left\{\frac{\mathbf J_N-\mathbf I_N}{2^{-2}\bar q^{3}}+\frac{(2\bar q^{-1}-1)\mathbf I_N}{2^{-2}\bar q^{2}}\right\}\{\ykcoplus\circ\Ikcplus{q}{}{0}\},\\
        \hat {\bar F}^d&=2\sum_{k=1}^{K-3}\{\ykc{q}{}{1}\circ\Ikc{q}{}{1}\}^{\top}\left\{\frac{-(\mathbf J_N-\mathbf I_N)}{2^{-\eta_k-1}q^{\eta_k}\bar q^{\eta_{k+1}}}\right\}\{\ykcoplus\circ\Ikcplus{q}{}{0}\}\\
        &+2\sum_{k=1}^{K-3}\{\ykdo\circ\Ik{q}{}{1}\}^{\top}\left\{\frac{-(\mathbf J_N-\mathbf I_N)}{2^{-2}q\bar q^{2}}\right\}\{\ykcoplus\circ\Ikcplus{q}{}{0}\}\\
        &+\sum_{k=1}^{K-3}\{\ykco\circ\Ikc{q}{}{1}\}^{\top}\left\{\frac{\mathbf I_N}{2^{-\eta_k}q^{\eta_k}}\right\}\{\ykco\circ\Ikc{q}{}{1}\} \\
        &+\sum_{k=1}^{K-3}\{\ykco\circ\Ikc{q}{}{1}\}^{\top}\left\{\frac{2\mathbf I_N}{2^{-\eta_k}q^{\eta_k}}\right\}\{\ykdo\circ\Ik{q}{}{1}\}\\
        &+\sum_{k=1}^{K-3}\{\ykdo\circ\Ik{q}{}{1}\}^{\top}\left\{\frac{\mathbf I_N}{2^{-1}q}\right\}\{\ykdo\circ\Ik{q}{}{1}\}\\
        &+\sum_{k=2}^{K-2}\{\ykco\circ\Ikc{q}{}{0}\}^{\top}\left\{\frac{\mathbf I_N}{2^{-\eta_k}\bar q^{\eta_k}}\right\}\{\ykco\circ\Ikc{q}{}{0}\}, \\
        \hat {\bar G}^d&=2\sum_{k=1}^{K-3}\{\ykco\circ\Ikc{q}{}{0}\}^{\top}\left\{\frac{-(\mathbf J_N-\mathbf I_N)}{2^{-\eta_k-1}q^{\eta_{k+1}}\bar q^{\eta_{k}}}\right\}\{\ykcoplus\circ\Ikcplus{q}{}{1}\}\\
        &+2\sum_{k=1}^{K-3}\{\ykdo\circ\Ik{q}{}{0}\}^{\top}\left\{\frac{-(\mathbf J_N-\mathbf I_N)}{2^{-2}q^2\bar q}\right\}\{\ykcoplus\circ\Ikcplus{q}{}{1}\}\\
        &+\sum_{k=1}^{K-3}\{\ykco\circ\Ikc{q}{}{0}\}^{\top}\left\{\frac{\mathbf I_N}{2^{-\eta_k}\bar q^{\eta_k}}\right\}\{\ykco\circ\Ikc{q}{}{0}\} \\
        &+\sum_{k=1}^{K-3}\{\ykco\circ\Ikc{q}{}{0}\}^{\top}\left\{\frac{2\mathbf I_N}{2^{-\eta_k}\bar q^{\eta_k}}\right\}\{\ykdo\circ\Ik{q}{}{0}\}\\
        &+\sum_{k=1}^{K-3}\{\ykdo\circ\Ik{q}{}{0}\}^{\top}\left\{\frac{\mathbf I_N}{2^{-1}\bar q}\right\}\{\ykdo\circ\Ik{q}{}{0}\}\\
        &+\sum_{k=2}^{K-2}\{\ykco\circ\Ikc{q}{}{1}\}^{\top}\left\{\frac{\mathbf I_N}{2^{-\eta_k}q^{\eta_k}}\right\}\{\ykco\circ\Ikc{q}{}{1}\}, \\
    \end{align*}
    \begin{align*}
        \hat A^s&= \sum_{k=1}^{K-2}\{\ykco\circ \Ikc{q_1}{z}{1}\}^{\top}\left\{\frac{3(\mathbf J_N-\mathbf I_N)}{2^{-\eta_k}q_{1,z}^{2\eta_k}}+\frac{(4q_{1,z}^{-2}-1)\mathbf I_N}{2^{-\eta_k}q_{1,z}^{\eta_k}}\right\}\{\ykco\circ \Ikc{q_1}{z}{1}\}\\
        & + \sum_{k=1}^{K-2}\{\ykdo\circ \Ik{q_1}{z}{1}\}^{\top}\left\{\frac{2(\mathbf J_N-\mathbf I_N)}{2^{-\eta_k}q_{1,z}^{\eta_k+1}}+\frac{(4q_{1,z}^{-1}-2)\mathbf I_N}{2^{-\eta_k}q_{1,z}^{\eta_k}}\right\}\{\ykco\circ \Ikc{q_1}{z}{1}\} \\
        & + \sum_{k=1}^{K-2}\{\ykdo\circ \Ik{q_1}{z}{1}\}^{\top}\left\{\frac{\mathbf J_N-\mathbf I_N}{2^{-1}q_{1,z}^{2}}+\frac{(2q_{1,z}^{-1}-1)\mathbf I_N}{2^{-1}q_{1,z}}\right\}\{\ykdo\circ \Ik{q_1}{z}{1}\},\\
        \hat B^s&= \sum_{k=1}^{K-2}\{\ykco\circ \Ikc{q_2}{z}{1}\}^{\top}\left\{\frac{3(\mathbf J_N-\mathbf I_N)}{2^{-\eta_k}q_{2,z}^{2\eta_k}}+\frac{(4q_{2,z}^{-2}-1)\mathbf I_N}{2^{-\eta_k}q_{2,z}^{\eta_k}}\right\}\{\ykco\circ \Ikc{q_2}{z}{1}\}\\
        & + \sum_{k=1}^{K-2}\{\ykdo\circ \Ik{q_2}{z}{1}\}^{\top}\left\{\frac{2(\mathbf J_N-\mathbf I_N)}{2^{-\eta_k}q_{2,z}^{\eta_k+1}}+\frac{(4q_{2,z}^{-1}-2)\mathbf I_N}{2^{-\eta_k}q_{2,z}^{\eta_k}}\right\}\{\ykco\circ \Ikc{q_2}{z}{1}\} \\
        & + \sum_{k=1}^{K-2}\{\ykdo\circ \Ik{q_2}{z}{1}\}^{\top}\left\{\frac{\mathbf J_N-\mathbf I_N}{2^{-1}q_{2,z}^{2}}+\frac{(2q_{2,z}^{-1}-1)\mathbf I_N}{2^{-1}q_{2,z}}\right\}\{\ykdo\circ \Ik{q_2}{z}{1}\},\\
        \hat C^s&= \sum_{k=1}^{K-2}\{\ykco\circ \Ikc{q_1}{z}{1}\}^{\top}\left\{\frac{\mathbf J_N-\mathbf I_N}{2^{-\eta_k}q_{1,z}^{2\eta_k}}+\frac{\mathbf I_N}{2^{-\eta_k}q_{1,z}^{\eta_k}}\right\}\{\ykco\circ \Ikc{q_1}{z}{1}\}\\
        & + \sum_{k=1}^{K-2}\{\ykdo\circ \Ik{q_1}{z}{1}\}^{\top}\left\{\frac{2(\mathbf J_N-\mathbf I_N)}{2^{-\eta_k}q_{1,z}^{\eta_k+1}}+\frac{2\mathbf I_N}{2^{-\eta_k}q_{1,z}^{\eta_k}}\right\}\{\ykco\circ \Ikc{q_1}{z}{1}\} \\
        & + \sum_{k=1}^{K-2}\{\ykdo\circ \Ik{q_1}{z}{1}\}^{\top}\left\{\frac{\mathbf J_N-\mathbf I_N}{2^{-1}q_{1,z}^{2}}+\frac{\mathbf I_N}{2^{-1}q_{1,z}}\right\}\{\ykdo\circ \Ik{q_1}{z}{1}\},\\
        &+ \sum_{k=1}^{K-2}\{\ykco\circ \Ikc{q_2}{z}{1}\}^{\top}\left\{\frac{\mathbf J_N-\mathbf I_N}{2^{-\eta_k}q_{2,z}^{2\eta_k}}+\frac{\mathbf I_N}{2^{-\eta_k}q_{2,z}^{\eta_k}}\right\}\{\ykco\circ \Ikc{q_2}{z}{1}\}\\
        & + \sum_{k=1}^{K-2}\{\ykdo\circ \Ik{q_2}{z}{1}\}^{\top}\left\{\frac{2(\mathbf J_N-\mathbf I_N)}{2^{-\eta_k}q_{2,z}^{\eta_k+1}}+\frac{2\mathbf I_N}{2^{-\eta_k}q_{2,z}^{\eta_k}}\right\}\{\ykco\circ \Ikc{q_2}{z}{1}\} \\
        & + \sum_{k=1}^{K-2}\{\ykdo\circ \Ik{q_2}{z}{1}\}^{\top}\left\{\frac{\mathbf J_N-\mathbf I_N}{2^{-1}q_{2,z}^{2}}+\frac{\mathbf I_N}{2^{-1}q_{2,z}}\right\}\{\ykdo\circ \Ik{q_2}{z}{1}\},\\
        \hat D^s&=2\sum_{k=1}^{K-3}\{\ykco\circ\Ikc{q_1}{z}{1}\}^{\top}\left\{\frac{\mathbf J_N-\mathbf I_N}{2^{-\eta_k-1}q_{1,z}^{\eta_k+\eta_{k+1}}}+\frac{(2q_{1,z}^{-1}-1)\mathbf I_N}{2^{-\eta_k-1}q_{1,z}^{\eta_k+1}}\right\}\{\ykcoplus\circ\Ikcplus{q_1}{z}{1}\}\\
        &+2\sum_{k=1}^{K-3}\{\ykdo\circ\Ik{q_1}{z}{1}\}^{\top}\left\{\frac{\mathbf J_N-\mathbf I_N}{2^{-2}q_{1,z}^{3}}+\frac{(2q_{1,z}^{-1}-1)\mathbf I_N}{2^{-2}q_{1,z}^{2}}\right\}\{\ykcoplus\circ\Ikcplus{q_1}{z}{1}\},\\
    \end{align*}
    \begin{align*}
        \hat E^s&=2\sum_{k=1}^{K-3}\{\ykco\circ\Ikc{q_2}{z}{1}\}^{\top}\left\{\frac{\mathbf J_N-\mathbf I_N}{2^{-\eta_k-1}q_{2,z}^{\eta_k+\eta_{k+1}}}+\frac{(2q_{2,z}^{-1}-1)\mathbf I_N}{2^{-\eta_k-1}q_{2,z}^{\eta_k+1}}\right\}\{\ykcoplus\circ\Ikcplus{q_2}{z}{1}\}\\
        &+2\sum_{k=1}^{K-3}\{\ykdo\circ\Ik{q_2}{z}{1}\}^{\top}\left\{\frac{\mathbf J_N-\mathbf I_N}{2^{-2}q_{2,z}^{3}}+\frac{(2q_{2,z}^{-1}-1)\mathbf I_N}{2^{-2}q_{2,z}^{2}}\right\}\{\ykcoplus\circ\Ikcplus{q_2}{z}{1}\},\\
        \hat {\bar F}^s&=\sum_{k=1}^{K-3}\{\ykco\circ \Ikc{q_1}{z}{1}\}^{\top}\left\{\frac{\mathbf J_N-\mathbf I_N}{2^{-\eta_k}q_{1,z}^{2\eta_k}}+\frac{\mathbf I_N}{2^{-\eta_k}q_{1,z}^{\eta_k}}\right\}\{\ykco\circ \Ikc{q_1}{z}{1}\}\\
        & + \sum_{k=1}^{K-3}\{\ykdo\circ \Ik{q_1}{z}{1}\}^{\top}\left\{\frac{2(\mathbf J_N-\mathbf I_N)}{2^{-\eta_k}q_{1,z}^{\eta_k+1}}+\frac{2\mathbf I_N}{2^{-\eta_k}q_{1,z}^{\eta_k}}\right\}\{\ykco\circ \Ikc{q_1}{z}{1}\} \\
        & + \sum_{k=1}^{K-3}\{\ykdo\circ \Ik{q_1}{z}{1}\}^{\top}\left\{\frac{\mathbf J_N-\mathbf I_N}{2^{-1}q_{1,z}^{2}}+\frac{\mathbf I_N}{2^{-1}q_{1,z}}\right\}\{\ykdo\circ \Ik{q_1}{z}{1}\}\\
        &+\sum_{k=2}^{K-2}\{\ykco\circ \Ikc{q_2}{z}{1}\}^{\top}\left\{\frac{\mathbf J_N-\mathbf I_N}{2^{-\eta_k}q_{2,z}^{2\eta_k}}+\frac{\mathbf I_N}{2^{-\eta_k}q_{2,z}^{\eta_k}}\right\}\{\ykco\circ \Ikc{q_2}{z}{1}\},\\
        \hat {\bar G}^s&= \sum_{k=1}^{K-3}\{\ykco\circ \Ikc{q_2}{z}{1}\}^{\top}\left\{\frac{\mathbf J_N-\mathbf I_N}{2^{-\eta_k}q_{2,z}^{2\eta_k}}+\frac{\mathbf I_N}{2^{-\eta_k}q_{2,z}^{\eta_k}}\right\}\{\ykco\circ \Ikc{q_2}{z}{1}\}\\
        & + \sum_{k=1}^{K-3}\{\ykdo\circ \Ik{q_2}{z}{1}\}^{\top}\left\{\frac{2(\mathbf J_N-\mathbf I_N)}{2^{-\eta_k}q_{2,z}^{\eta_k+1}}+\frac{2\mathbf I_N}{2^{-\eta_k}q_{2,z}^{\eta_k}}\right\}\{\ykco\circ \Ikc{q_2}{z}{1}\} \\
        & + \sum_{k=1}^{K-3}\{\ykdo\circ \Ik{q_2}{z}{1}\}^{\top}\left\{\frac{\mathbf J_N-\mathbf I_N}{2^{-1}q_{2,z}^{2}}+\frac{\mathbf I_N}{2^{-1}q_{2,z}}\right\}\{\ykdo\circ \Ik{q_2}{z}{1}\}\\
        &+\sum_{k=2}^{K-2}\{\ykco\circ \Ik{q_1}{z}{1}\}^{\top}\left\{\frac{\mathbf J_N-\mathbf I_N}{2^{-\eta_k}q_{1,z}^{2\eta_k}}+\frac{\mathbf I_N}{2^{-\eta_k}q_{1,z}^{\eta_k}}\right\}\{\ykco\circ \Ik{q_1}{z}{1}\}. 
    \end{align*}
    
\end{theorem}


\section{Exact optimal $\R[q_1]$}
\label{sec:B}

For fixed $N$ and $\R[q_1]+\R[q_2]\neq 1$, the exact optimal $\R[q_1]$ may vary across different scenarios. \Cref{fig.r} shows the optimal $\R[q_1]$ under different objective functions (combined risks), different population sizes $N$ and different combinations of $\mathcal J_j$ (including $\mathcal J_1:\mathcal J_2=2:1$, $\mathcal J_1:\mathcal J_2:\mathcal J_3=3:2:1$ and $\mathcal J_1:\mathcal J_2:\mathcal J_3:\mathcal J_4=4:3:2:1$) when $q_1=0.8$ and $q_2=0.5$. They all converge towards 0.5 as $N$ tends to infinity.

\begin{figure}[htbp]
    \centering
    \includegraphics[width=\linewidth]{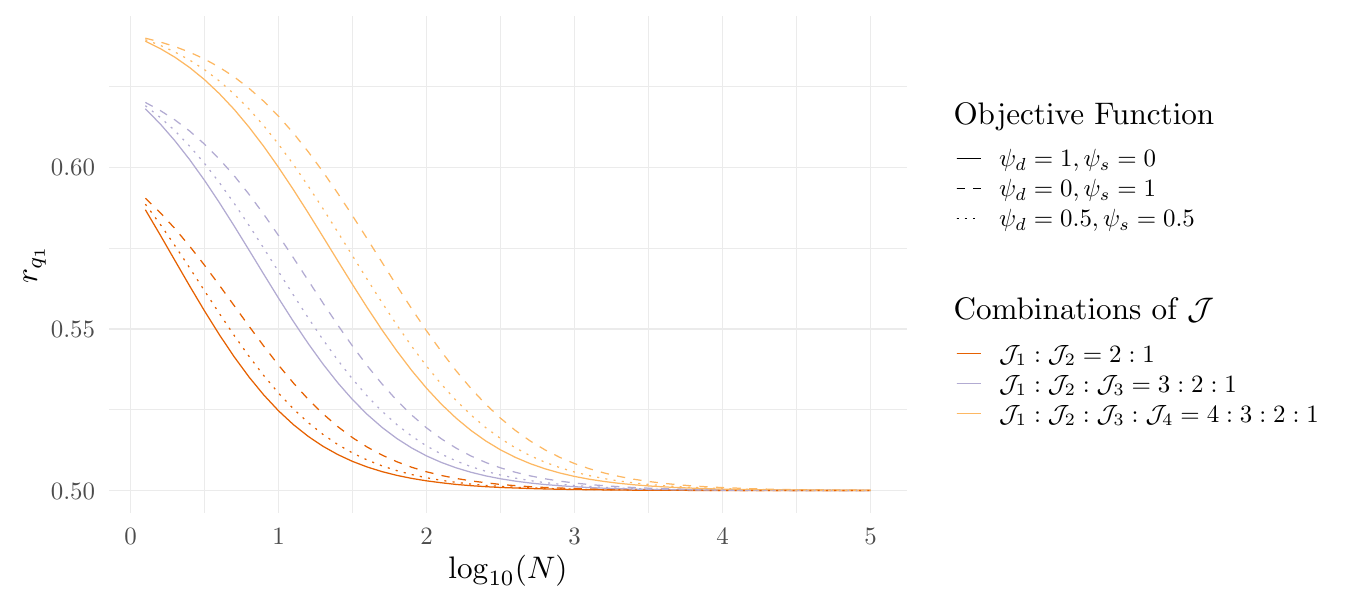}
    \caption{Optimal $\R[q_1]$ under different $N$ and combinations of $\mathcal J_j$. }\label{fig.r}
\end{figure}

\section{Proof of Theorem~\ref{theorem.r}}
\label{sec:sm.theorem.r}

Recall that $\mathbb T= \{t_0,\ldots,t_L \}$ is the set of decision points, $\R[q_1]$ and $\R[q_2]$ are the probabilities of choosing the treated probability from $q_1$ and $q_2$, $\mathbb{Q} = \{Q_{t_0}, Q_{t_1}, \ldots, Q_{t_L}\}$ is the set of treated probabilities, $\mathbb{Y} = \{Y_{i,t}(\mathbf q_{(t-p):t},\mathbf z_{i,(t-p):t})\mid i\in [N], \ t\in [T], \ \mathbf q_{(t-p):t}\in \{q_1,q_2\}^{p+1}, \ \mathbf z_{i,(t-p):t}\in \{0,1\}^{p+1}\}$ is the set of all potential outcomes.

We first introduce the following notation:
\begin{itemize}
    \item [1.] $\mathcal F_{\mathbb T}(t)=\max\{j\mid j\in \mathbb T,j\leq t\}$: This represents the maximum index $j$ from the set of randomization decision points $\mathbb T$ such that $j$ is less than or equal to $t$. In simpler terms, it denotes the randomization time point that determines the treatment assignment at time $t$.
    \item [2.] $\mathcal F_{\mathbb T}^p(t)=\{j\mid \exists i\in\{t-p,\ldots,t\},\text{ such that } j=\mathcal F_{\mathbb T}(i)\}$: This set includes all indices $j$ for which there exists some $i$ in the range ${t-p, \ldots, t}$ such that $j = \mathcal F_{\mathbb T}(i)$. In other words, it consists of the randomization time points in $\mathbb T$ that collectively determine the treatment assignments for the time periods ${t-p, \ldots, t}$. It is worth noting that $\mathcal F_{\mathbb T}^p(t)$ is not empty for all $t$.
    \item [3.] $J_t=|\mathcal F_{\mathbb T}^p(t)|$: This represents the number of elements in $\mathcal F_{\mathbb T}^p(t)$, denoting how many different points collectively determine the treatment assignment for the time period $t$. It is always greater than or equal to 1 and less than or equal to $p+1$.
    \item [4.] $O_{\mathbb T}(t,t')=\mathcal F_{\mathbb T}^p(t)\cap \mathcal F_{\mathbb T}^p(t')$: This set includes the common randomization points between the time periods ${t-p, \ldots, t}$ and ${t'-p, \ldots, t'}$. In other words, it is the set of overlapping randomization points that affect both of these time periods.
    \item [5.] $J_{t,t'}^{\circ}=|O_{\mathbb T}(t,t')|$: This represents the number of overlapping randomization points between the time periods ${t-p, \ldots, t}$ and ${t'-p, \ldots, t'}$ that determine treatment assignments. It is always less than or equal to both $J_t$ and $J_{t'}$.
\end{itemize}

Recall that the objective function is given by
\[
\mathcal{L}(\psi_d, \psi_s) = \psi_d \{\mathrm{risk}^d(q_1) + \mathrm{risk}^d(q_2)\} + \psi_s \{\mathrm{risk}^s(1) + \mathrm{risk}^s(0)\}=\psi_d \mathcal L(1,0)+\psi_s \mathcal L(0,1).
\]
This expression represents the weighted sum of the direct effect risk and the spillover effect risk. Here, $\psi_d$ and $\psi_s$ are the weights assigned to the direct and spillover risks, respectively. As a result, our primary goal is to obtain the theoretical results for $\mathcal L(1,0)$ and $\mathcal L(0,1)$; subsequently, we combine these results to derive the results for $\mathcal L(\psi_d,\psi_s)$.

Recall that $q_{1,z} = zq_1 + (1 - z)(1 - q_1)$ and $q_{2,z} = zq_2 + (1 - z)(1 - q_2)$. Given $\R[q_1], \R[q_2]$ and $\mathbb T$,
\begin{align}
\label{eq:pr}
    \Pit{q_1}{z}{1}=\R[q_1]^{J_t}q_{1,z}^{J_t},\quad \Pit{q_2}{z}{1}=\R[q_2]^{J_t}q_{2,z}^{J_t},\quad z=0,1.
\end{align}
Define
\begin{align*}
    \mathbf D_{i,t}(\mathbb T,q)&=\yit{q}{}{1}\left\{\frac{\Iit{q}{}{1}}{\R^{J_t}q^{J_t}}-1\right\}-\yit{q}{}{0}\left\{\frac{\Iit{q}{}{0}}{\R^{J_t}\bar q^{J_t}}-1\right\},\\
    \mathbf S_{i,t}(\mathbb T,z)&=\yit{q_1}{z}{1}\left\{\frac{\Iit{q_1}{z}{1}}{\R[q_1]^{J_t}q_{1,z}^{J_t}}-1\right\}-\yit{q_2}{z}{1}\left\{\frac{\Iit{q_2}{z}{1}}{\R[q_2]^{J_t}q_{2,z}^{J_t}}-1\right\},\\
    \mathbf D_t(\mathbb T,q)&=\frac 1N\sum_{i=1}^N\mathbf D_{i,t}(\mathbb T,q),\quad
    \mathbf S_t(\mathbb T,z)=\frac 1N\sum_{i=1}^N\mathbf S_{i,t}(\mathbb T,z) .
\end{align*}
Then
\begin{align*}
    \taudhat-\taud&=\frac 1{T-p}\sum_{t=p+1}^T\frac 1N\sum_{i=1}^N\mathbf D_{i,t}(\mathbb T,q)=\frac 1{T-p}\sum_{t=p+1}^T\mathbf D_t(\mathbb T,q),\\
    \taushat-\taus&=\frac 1{T-p}\sum_{t=p+1}^T\frac 1N\sum_{i=1}^N\mathbf S_{i,t}(\mathbb T,z)=\frac 1{T-p}\sum_{t=p+1}^T\mathbf S_t(\mathbb T,z).
\end{align*}
Note that $\taudhat$ and $\taushat$ depend on  $\R[q_1], \R[q_2]$, $\mathbb T$, and $\mathbb Y$. For simplicity, we will ignore these dependencies unless it leads to confusion.

\begin{lemma}[Properties of $\mathbf D_t$ and $\mathbf S_t$]
    \label{lemma.properties}
    Under Assumptions~\ref{assumption.nonanticipativity}--\ref{assumption.bounded}, $\mathbf D_t$ and $\mathbf S_t$ satisfy

    (1.1) $E \{ \mathbf D_t(\mathbb T,q) \} = 0$;

    (1.2) for $q=q_1,q_2$,
    \begin{align*}
        & N^2E\left\{\mathbf D_t(\mathbb T,q)\right\}^2\\
        =& (\R^{-J_t}-1)\left\{\sum_{i=1}^N\yit{q}{}{1}-\sum_{i=1}^N\yit{q}{}{0}\right\}^2+2\R^{-J_t}\sum_{i=1}^N\left\{\yit{q}{}{1}\yit{q}{}{0}\right\}\\
        & +\R^{-J_t}(q^{-J_t}-1)\sum_{i=1}^N\left\{\yit{q}{}{1}\right\}^2+\R^{-J_t}(\bar q^{-J_t}-1)\sum_{i=1}^N\left\{\yit{q}{}{0}\right\}^2;
    \end{align*}
    
    (1.3) for $q=q_1,q_2$ and $t\neq t'$, 
    \begin{align*}
        &N^2E\left\{\mathbf D_t(\mathbb T,q)\mathbf D_{t'}(\mathbb T,q)\right\}\\
        =& (\R^{-J_{t,t'}^{\circ}}-1)\left\{\sum_{i=1}^N\yit{q}{}{1}-\sum_{i=1}^N\yit{q}{}{0}\right\}\left\{\sum_{i=1}^N\yitp{q}{}{1}-\sum_{i=1}^N\yitp{q}{}{0}\right\}\\
        & +\R^{-J_{t,t'}^{\circ}}\sum_{i=1}^N\left\{\yit{q}{}{1}\yitp{q}{}{0}\right\}+\R^{-J_{t,t'}^{\circ}}\sum_{i=1}^N\left\{\yit{q}{}{0}\yitp{q}{}{1}\right\}\\
        & +\R^{-J_{t,t'}^{\circ}}(q^{-J_{t,t'}^{\circ}}-1)\sum_{i=1}^N\left\{\yit{q}{}{1}\yitp{q}{}{1}\right\}\\
        & +\R^{-J_{t,t'}^{\circ}}(\bar q^{-J_{t,t'}^{\circ}}-1)\sum_{i=1}^N\left\{\yit{q}{}{0}\yitp{q}{}{0}\right\};
    \end{align*}

    (2.1) $E\{\mathbf S_t(\mathbb T,z)\}=0$;

    (2.2) for $z=0,1$, 
    \begin{align*}
        &N^2E\left\{\mathbf S_t(\mathbb T,z)\right\}^2\\
        =&(\R[q_1]^{-J_t}-1)\left\{\sum_{i=1}^N\yit{q_1}{z}{1}\right\}^2+(\R[q_2]^{-J_t}-1)\left\{\sum_{i=1}^N\yit{q_2}{z}{1}\right\}^2\\
        & +2\left\{\sum_{i=1}^N\yit{q_1}{z}{1}\right\}\left\{\sum_{i=1}^N\yit{q_2}{z}{1}\right\}\\
        & +\R[q_1]^{-J_t}(q_{1,z}^{-J_t}-1)\sum_{i=1}^N\left\{\yit{q_1}{z}{1}\right\}^2+\R[q_2]^{-J_t}(q_{2,z}^{-J_t}-1)\sum_{i=1}^N\left\{\yit{q_2}{z}{1}\right\}^2;
    \end{align*}

    (2.3) for $z=0,1$ and $t\neq t'$, 
    \begin{align*}
        &N^2E\left\{\mathbf S_t(\mathbb T,z)\mathbf S_{t'}(\mathbb T,z)\right\}\\
        =&(\R[q_1]^{-J_{t,t'}^{\circ}}-1)\left\{\sum_{i=1}^N\yit{q_1}{z}{1}\right\}\left\{\sum_{i=1}^N\yitp{q_1}{z}{1}\right\}\\
        & +(\R[q_2]^{-J_{t,t'}^{\circ}}-1)\left\{\sum_{i=1}^N\yit{q_2}{z}{1}\right\}\left\{\sum_{i=1}^N\yitp{q_2}{z}{1}\right\}\\
        & +\left\{\sum_{i=1}^N\yit{q_1}{z}{1}\right\}\left\{\sum_{i=1}^N\yitp{q_2}{z}{1}\right\}+\left\{\sum_{i=1}^N\yitp{q_1}{z}{1}\right\}\left\{\sum_{i=1}^N\yit{q_2}{z}{1}\right\}\\
        & +\R[q_1]^{-J_{t,t'}^{\circ}}(q_{1,z}^{-J_{t,t'}^{\circ}}-1)\sum_{i=1}^N\left\{\yit{q_1}{z}{1}\yitp{q_1}{z}{1}\right\}\\
        & +\R[q_2]^{-J_{t,t'}^{\circ}}(q_{2,z}^{-J_{t,t'}^{\circ}}-1)\sum_{i=1}^N\left\{\yit{q_2}{z}{1}\yitp{q_2}{z}{1}\right\}.
    \end{align*}
\end{lemma}

\begin{proof}[Proof of Lemma~\ref{lemma.properties}]
    (1.1) By \Cref{eq:pr}, it is easy to see that $E\{\mathbf D_{i,t}(\mathbb T,q)\}=0$, which further implies that $E\{\mathbf D_t(\mathbb T,q)\}=0$. 
    
    (1.2) For $i\neq j$,
    \begin{align*}
        E\{\mathbf D_{i,t}(\mathbb T,q)\}^2=&(q^{-J_t}\R^{-J_t}-1)\{\yit{q}{}{1}\}^2+(\bar q^{-J_t}\R^{-J_t}-1)\{\yit{q}{}{0}\}^2\\
        &+2\yit{q}{}{1}\yit{q}{}{0},\\
        E\{\mathbf D_{i,t}(\mathbb T,q)\mathbf D_{j,t}(\mathbb T,q)\}=&(\R^{-J_t}-1)\{\yit{q}{}{1}-\yit{q}{}{0}\}\{\yjt{q}{}{1}-\yjt{q}{}{0}\}.
    \end{align*}
    The result follows from 
    \begin{align*}
    N^2E\{\mathbf D_t(\mathbb T,q)\}^2=\sum_{i=1}^NE\{\mathbf D_{i,t}(\mathbb T,q)\}^2+\sum_{i\neq j}E\{\mathbf D_{i,t}(\mathbb T,q)\mathbf D_{j,t}(\mathbb T,q)\}.
    \end{align*}

    (1.3) For $t\neq t'$ and $J_{t,t'}^{\circ}\neq 0$,
    \begin{align*}
        E\{\mathbf D_{i,t}(\mathbb T,q)\mathbf D_{i,t'}(\mathbb T,q)\}
        =&(q^{-J_{t,t'}^{\circ}}\R^{-J_{t,t'}^{\circ}}-1)\yit{q}{}{1}\yitp{q}{}{1}\\
        &+(\bar q^{-J_{t,t'}^{\circ}}\R^{-J_{t,t'}^{\circ}}-1)\yit{q}{}{0}\yitp{q}{}{0}\\
        &+\yit{q}{}{1}\yitp{q}{}{0}+\yit{q}{}{0}\yitp{q}{}{1},\\
        E\{\mathbf D_{i,t}(\mathbb T,q)\mathbf D_{j,t'}(\mathbb T,q)\}
        =&(\R^{-J_{t,t'}^{\circ}}-1)\{\yit{q}{}{1}-\yit{q}{}{0}\}\{\yjtp{q}{}{1}-\yjtp{q}{}{0}\}.
    \end{align*}
    The result follows from 
    \begin{align*}
    N^2E\{\mathbf D_t(\mathbb T,q)\mathbf D_{t'}(\mathbb T,q)\}=\sum_{i=1}^NE\{\mathbf D_{i,t}(\mathbb T,q)\mathbf D_{i,t'}(\mathbb T,q)\}+\sum_{i\neq j}E\{\mathbf D_{i,t}(\mathbb T,q)\mathbf D_{j,t'}(\mathbb T,q)\}.
    \end{align*}

    (2.1) By \Cref{eq:pr}, it is easy to see that $E\{\mathbf S_{i,t}(\mathbb T,z)\}=0$, which further implies that $E\{\mathbf S_t(\mathbb T,z)\}=0$, $z=0,1$.

    (2.2) For $i\neq j$,
    \begin{align*}
        E\{\mathbf S_{i,t}(\mathbb T,z)\}^2=&(q_{1,z}^{-J_t}\R[q_1]^{-J_t}-1)\{\yit{q_1}{z}{1}\}^2+(q_{2,z}^{-J_t}\R[q_2]^{-J_t}-1)\{\yit{q_2}{z}{1}\}^2\\
        &+2\yit{q_1}{z}{1}\yit{q_2}{z}{1},\\
        E\{\mathbf S_{i,t}(\mathbb T,z)\mathbf S_{j,t}(\mathbb T,z)\}=&(\R[q_1]^{-J_t}-1)\yit{q_1}{z}{1}\yjt{q_1}{z}{1}+(\R[q_2]^{-J_t}-1)\yit{q_2}{z}{1}\yjt{q_2}{z}{1}\\
        &+\yit{q_1}{z}{1}\yjt{q_2}{z}{1}+\yit{q_2}{z}{1}\yjt{q_1}{z}{1}.
    \end{align*}
    The result follows from 
    \begin{align*}
    N^2E\{\mathbf S_t(\mathbb T,z)\}^2=\sum_{i=1}^NE\{\mathbf S_{i,t}(\mathbb T,z)\}^2+\sum_{i\neq j}E\{\mathbf S_{i,t}(\mathbb T,z)\mathbf S_{j,t}(\mathbb T,z)\}.
    \end{align*}

    (2.3) For $t\neq t'$ and $J_{t,t'}^{\circ}\neq 0$,
    \begin{align*}
        E\{\mathbf S_{i,t}(\mathbb T,z)\mathbf S_{i,t'}(\mathbb T,z)\}=&(q_{1,z}^{-J_{t,t'}^{\circ}}\R[q_1]^{-J_{t,t'}^{\circ}}-1)\yit{q_1}{z}{1}\yitp{q_1}{z}{1}\\
        &+(q_{2,z}^{-J_{t,t'}^{\circ}}\R[q_2]^{-J_{t,t'}^{\circ}}-1)\yit{q_2}{z}{1}\yitp{q_2}{z}{1}\\
        &+\yit{q_1}{z}{1}\yitp{q_2}{z}{1}+\yit{q_2}{z}{1}\yitp{q_1}{z}{1},\\
        E\{\mathbf S_{i,t}(\mathbb T,z)\mathbf S_{j,t'}(\mathbb T,z)\}=&(\R[q_1]^{-J_{t,t'}^{\circ}}-1)\yit{q_1}{z}{1}\yjtp{q_1}{z}{1}\\
        &+(\R[q_2]^{-J_{t,t'}^{\circ}}-1)\yit{q_2}{z}{1}\yjtp{q_2}{z}{1}\\
        &+\yit{q_1}{z}{1}\yjtp{q_2}{z}{1}+\yit{q_2}{z}{1}\yjtp{q_1}{z}{1}.
    \end{align*}
    The result follows from
    \begin{align*}
    N^2E\{\mathbf S_t(\mathbb T,z)\mathbf S_{t'}(\mathbb T,z)\}= \sum_{i=1}^NE\{\mathbf S_{i,t}(\mathbb T,z)\mathbf S_{i,t'}(z)\}+\sum_{i\neq j}E\{\mathbf S_{i,t}(\mathbb T,z)\mathbf S_{j,t'}(\mathbb T,z)\}.
    \end{align*}
\end{proof}

Recall that $B$ is the bound of potential outcomes as stated in \Cref{assumption.bounded}.
\begin{lemma}
    \label{lemma.Y}
    Under Assumptions~\ref{assumption.nonanticipativity}--\ref{assumption.bounded}, given $\mathbb T^*$, $\R[q_1]$ and $\R[q_2]$, we have 

    (1) when $N\geq (1- \max\{\R[q_1],\R[q_2]\})^{-1}$, $\mathop{\arg\max}_{\mathbb Y\in \mathcal Y}\mathcal L(1,0)=\{Y_{i,t}(\mathbf q_{(t-p):t},\mathbf z_{i,(t-p):t})\mid \yit{q_1}{}{1}=C_1,\yit{q_1}{}{0}=-C_1,\yit{q_2}{}{1}=C_2,\yit{q_2}{}{0}=-C_2\}$ with $C_1=\pm B$, $C_2=\pm B$; when $N\leq (1-\min\{\R[q_1]^{p+1},\R[q_2]^{p+1}\})^{-1}$, $\mathop{\arg\max}_{\mathbb Y\in \mathcal Y}\mathcal L(1,0)=\{Y_{i,t}(\mathbf q_{(t-p):t},\mathbf z_{i,(t-p):t})\mid \yit{q_1}{}{1}=C_1,\yit{q_1}{}{0}=C_1,\yit{q_2}{}{1}=C_2,\yit{q_2}{}{0}=C_2\}$ with $C_1=\pm B$ and $C_2=\pm B$;

    (2) $\mathop{\arg \max}_{\mathbb Y\in \mathcal Y}\mathcal L(0,1)=\{Y_{i,t}(\mathbf q_{(t-p):t},\mathbf z_{i,(t-p):t})\mid \yit{q_1}{}{1}=C_3,\yit{q_1}{}{0}=C_4,\yit{q_2}{}{1}=C_3,\yit{q_2}{}{0}=C_4\}$ with $C_3=\pm B$ and $C_4=\pm B$;

    (3) for any $\psi_d>0$ and $\psi_s>0$, when $N\geq (1- \max\{\R[q_1],\R[q_0]\})^{-1}$, $\mathop{\arg \max}_{\mathbb Y\in \mathcal Y}\mathcal L(\psi_d,\psi_s)=\{Y_{i,t}(\mathbf q_{(t-p):t},\mathbf z_{i,(t-p):t})\mid \yit{q_1}{}{1}=C,\yit{q_1}{}{0}=-C,\yit{q_2}{}{1}=C,\yit{q_2}{}{0}=-C\}$ with $C=\pm B$; when $N\leq (1-\min\{\R[q_1]^{p+1},\R[q_0]^{p+1}\})^{-1}$, $\mathop{\arg \max}_{\mathbb Y\in \mathcal Y}\mathcal L(\psi_d,\psi_s)=\{Y_{i,t}(\mathbf q_{(t-p):t},\mathbf z_{i,(t-p):t})\mid \yit{q_1}{}{1}=C,\yit{q_1}{}{0}=C,\yit{q_2}{}{1}=C,\yit{q_2}{}{0}=C\}$ with $C=\pm B$.
\end{lemma}

\begin{proof}[Proof of Lemma~\ref{lemma.Y}]
    First of all, recall that 
    \begin{align*}
        \mathrm{risk}^d(q)&=E\{\taudhat-\taud\}^2=\frac 1{(T-p)^2}\left[\sum_{t=p+1}^TE\left\{\mathbf D_t(\mathbb T,q)\right\}^2+\sum_{t\neq t'}E\{\mathbf D_t(\mathbb T,q)\mathbf D_{t'}(\mathbb T,q)\}\right],\\
        \mathrm{risk}^s(z)&=E\{\taushat-\taus\}^2=\frac 1{(T-p)^2}\left[\sum_{t=p+1}^TE\left\{\mathbf S_t(\mathbb T,z)\right\}^2+\sum_{t\neq t'}E\{\mathbf S_t(\mathbb T,z)\mathbf S_{t'}(\mathbb T,z)\}\right].
    \end{align*}
    In the proof, we will respectively maximize $E\left\{\mathbf D_t(\mathbb T,q)\right\}^2$, $E\{\mathbf D_t(\mathbb T,q)\mathbf D_{t'}(\mathbb T,q)\}$, $E\left\{\mathbf S_t(\mathbb T,z)\right\}^2$ and $E\{\mathbf S_t(\mathbb T,z)\mathbf S_{t'}(\mathbb T,z)\}$ under Assumptions~\ref{assumption.nonanticipativity}--\ref{assumption.bounded}.

    (1) Recall that
    \begin{align*}
        & N^2E\left\{\mathbf D_t(\mathbb T,q)\right\}^2\\
        =& (\R^{-J_t}-1)\left\{\sum_{i=1}^N\yit{q}{}{1}-\sum_{i=1}^N\yit{q}{}{0}\right\}^2+2\R^{-J_t}\sum_{i=1}^N\left\{\yit{q}{}{1}\yit{q}{}{0}\right\}\\
        & +\R^{-J_t}(q^{-J_t}-1)\sum_{i=1}^N\left\{\yit{q}{}{1}\right\}^2+\R^{-J_t}(\bar q^{-J_t}-1)\sum_{i=1}^N\left\{\yit{q}{}{0}\right\}^2.
    \end{align*}
    The last two terms in \Cref{eq:lemmaY} are maximized when $|\yit{q}{}{1}|=|\yit{q}{}{0}|=B$ for specific $t$ and $i=1,\ldots,N$, where $B$ represents the bound of potential outcomes as stated in \Cref{assumption.bounded}. For simplicity, let us denote $a_i = \yit{q}{}{1}$ and $b_i = \yit{q}{}{0}$. The remaining term can be expressed as:
    \begin{align}
    \label{eq:lemmaY}
    (\R^{-J_t} - 1) \left( \sum_{i=1}^N a_i - \sum_{i=1}^N b_i \right)^2 + 2 \R^{-J_t} \sum_{i=1}^N a_i b_i.
    \end{align}
    Without loss of generality, we can assume that $\sum_{i=1}^N a_i - \sum_{i=1}^N b_i \geq 0$. If this is not the case, we can simply switch $\{a_i\}$ and $\{b_i\}$ to satisfy this condition with the value of \Cref{eq:lemmaY} unchanged.

    Next, consider two new sets defined as $\{a_i'\}$ and $\{b_i'\}$:
    \begin{align*}
        a_i'=a_i,\quad b_i'=b_i,\quad &\text{if}\quad a_i\geq b_i,\quad \max(a_i,b_i)\geq 0;\\
        a_i'=-b_i,\quad b_i'=-a_i,\quad &\text{if}\quad a_i\geq b_i,\quad \max(a_i,b_i)< 0;\\
        a_i'=b_i,\quad b_i'=a_i,\quad &\text{if}\quad a_i< b_i,\quad \max(a_i,b_i)\geq 0;\\
        a_i'=-a_i,\quad b_i'=-b_i,\quad &\text{if}\quad a_i< b_i,\quad \max(a_i,b_i)< 0.
    \end{align*}
    Then, $\{a_i'\}$ and $\{b_i'\}$ satisfies $a_i'\geq b_i'$ and $a_i'\geq 0$. Additionally, $a_i'-b_i'\geq a_i-b_i$ implies  $\sum_{i=1}^Na_i'-\sum_{i=1}^Nb_i'\geq \sum_{i=1}^Na_i-\sum_{i=1}^Nb_i\geq 0$, and $a_i'b_i'=a_ib_i$ implies $\sum_{i=1}^Na_i'b_i'=\sum_{i=1}^Na_ib_i$.
    Then we obtain $\{a_i'\}$ and $\{b_i'\}$ with $a_i'\geq b_i'$ that yield a value of \Cref{eq:lemmaY} at least as large as that of the original sets.

    Next, additionally consider two new sets defined as $\{a_i''\}$ and $\{b_i''\}$: $a_i''=B$ and $b_i''=B-(a_i'-b_i')$. We have $( \sum_{i=1}^Na_i''-\sum_{i=1}^Nb_i'' ) - ( \sum_{i=1}^Na_i'-\sum_{i=1}^Nb_i' ) = 0$ given that $a_i''-b_i''=a_i'-b_i'$; and $\sum_{i=1}^Na_i''b_i''\geq \sum_{i=1}^Na_i'b_i$ given that $a_i''b_i''-a_i'b_i'=Bb_i''-Bb_i'+Bb_i'-a_i'b_i'= B^2-Ba_i'+Bb_i'-a_i'b_i'=(B-a_i')(B+b_i')\geq 0$. Then, we obtain two new sets $\{a_i''\}$ and $\{b_i''\}$ with $a_i''=B$ and $b_i''\leq B$ that yield a value of \Cref{eq:lemmaY} at least as large as that of the original sets.

    Next, we maximize \Cref{eq:lemmaY} based on the sets $\{a_i''\}$ and $\{b_i''\}$. For simplicity, we define $x = \sum_{i=1}^N b_i''$. \Cref{eq:lemmaY} can be expressed as
    \begin{align*}
        f(x) = (\R^{-J_t} - 1)(NB - x)^2 + 2 \R^{-J_t} B x.
    \end{align*}
    Expanding this yields
    \begin{align*}
        f(x) = (\R^{-J_t} - 1)x^2 - 2B(N \R^{-J_t} - N - 2\R^{-J_t})x + (\R^{-J_t} - 1)N^2B^2,
    \end{align*}
    where $x \in [-NB, NB]$. Given that $J_t \geq 1$ and $\R \in (0, 1)$, we have $\R^{-J_t} - 1 > 0$.

    We can calculate that $f(-NB) = \{4N^2(\R^{-J_t} - 1) - 2N \R^{-J_t}\} B^2$
    and $f(NB) = 2N \R^{-J_t} B^2$. When $f(-NB)\geq f(NB)$, i.e., $N \geq (1 - \R)^{-1}$, the maximum is achieved when $x=-NB$, which corresponds to $b_i=-B$ for $i=1,\ldots,B$; when $f(-NB)\leq f(NB)$ for $J_t=1,\ldots,p+1$, i.e., $N \leq (1 - \R^{p+1})^{-1}$, the maximum is achieved when $x=NB$, which corresponds to $b_i=B$ for $i=1,\ldots,N$.

    In summary, when $N \geq (1 - \R)^{-1}$, we have demonstrated that any original sets $\{a_i\}$ and $\{b_i\}$ yield a value of \Cref{eq:lemmaY} that is at most equal to that obtained when $\yit{q}{}{1} = C$ and $\yit{q}{}{0} = -C$ for specific $t$ and $i=1,\ldots,N$, where $C = \pm B$. This is justified by the observation that the values of \Cref{eq:lemmaY} remain the same for $C = B$ and $C = -B$.

    Similarly, when $N \leq  (1 - \R^{p+1})^{-1}$, we can show that any original sets $\{a_i\}$ and $\{b_i\}$ yield a value of \Cref{eq:lemmaY} that is at most equal to that achieved when $\yit{q}{}{1} = \yit{q}{}{0} = C$ for specific $t$ and $i=1,\ldots,N$, where $C=\pm B$. 
    

    For the interaction term $\mathbf D_t(\mathbb T,q)\mathbf D_{t'}(\mathbb T,q)$, recall that
    \begin{align*}
        &N^2E\left\{\mathbf D_t(\mathbb T,q)\mathbf D_{t'}(\mathbb T,q)\right\}\\
        =& (\R^{-J_{t,t'}^{\circ}}-1)\left\{\sum_{i=1}^N\yit{q}{}{1}-\sum_{i=1}^N\yit{q}{}{0}\right\}\left\{\sum_{i=1}^N\yitp{q}{}{1}-\sum_{i=1}^N\yitp{q}{}{0}\right\}\\
        & +\R^{-J_{t,t'}^{\circ}}\sum_{i=1}^N\left\{\yit{q}{}{1}\yitp{q}{}{0}\right\}+\R^{-J_{t,t'}^{\circ}}\sum_{i=1}^N\left\{\yit{q}{}{0}\yitp{q}{}{1}\right\}\\
        & +\R^{-J_{t,t'}^{\circ}}(q^{-J_{t,t'}^{\circ}}-1)\sum_{i=1}^N\left\{\yit{q}{}{1}\yitp{q}{}{1}\right\}\\
        & +\R^{-J_{t,t'}^{\circ}}(\bar q^{-J_{t,t'}^{\circ}}-1)\sum_{i=1}^N\left\{\yit{q}{}{0}\yitp{q}{}{0}\right\}.
    \end{align*}
    The last two terms are maximized when $\yit{q}{z}{1}=\yitp{q}{z}{1}$ and $|\yit{q}{z}{1}|=|\yitp{q}{z}{1}|=B$ for $t$, $t'$, $z=0,1$ and $i=1,\ldots,N$. For simplicity, let us denote $a_i = \yit{q}{}{1}$, $b_i = \yit{q}{}{0}$, $c_i = \yitp{q}{}{1}$ and $d_i = \yitp{q}{}{0}$. The remaining term can be expressed as
    \begin{align*}
    &(\R^{-J_{t,t'}^{\circ}} - 1) \left( \sum_{i=1}^N a_i - \sum_{i=1}^N b_i \right)\left( \sum_{i=1}^N c_i - \sum_{i=1}^N d_i \right) + \R^{-J_{t,t'}^{\circ}} \sum_{i=1}^N a_i d_i + \R^{-J} \sum_{i=1}^N b_i c_i\\
    =&-\frac 12(\R^{-J_{t,t'}^{\circ}}-1)\left(\sum_{i=1}^Na_i-\sum_{i=1}^Nc_i\right)^2-\frac 12(\R^{-J_{t,t'}^{\circ}}-1)\left(\sum_{i=1}^Nb_i-\sum_{i=1}^Nd_i\right)^2\\
    &+\frac 12(\R^{-J_{t,t'}^{\circ}}-1)\left(\sum_{i=1}^Na_i-\sum_{i=1}^Nd_i\right)^2+\R^{-J_{t,t'}^{\circ}} \sum_{i=1}^N a_i d_i\\
    &+\frac 12(\R^{-J_{t,t'}^{\circ}}-1)\left(\sum_{i=1}^Nc_i-\sum_{i=1}^Nb_i\right)^2+ \R^{-J_{t,t'}^{\circ}} \sum_{i=1}^N b_i c_i.
    \end{align*}
    The terms in the second line are maximized when $\sum_{i=1}^Na_i=\sum_{i=1}^Nc_i$ and $\sum_{i=1}^Nb_i=\sum_{i=1}^Nd_i$; the terms in the third line can be interpreted as half of \Cref{eq:lemmaY} with $b_i$ substituted by $d_i$ and $J_t$ substituted by $J_{t,t'}^{\circ}$; the terms in the last line can be similarly interpreted as half of \Cref{eq:lemmaY} with $a_i$ substituted by $c_i$ and $J_t$ substituted by $J_{t,t'}^{\circ}$. Based on the results above, this function is maximized when $a_i=c_i=C$, $b_i=d_i=-C$ with $C=\pm B$ when $N\geq (1-\R)^{-1}$; it is maximized when $a_i=c_i=C$, $b_i=d_i=C$ with $C=\pm B$ when $N\leq (1-\R^{p+1})^{-1}$.

    In summary, when $N\geq (1-\R)^{-1}$, $\mathrm{risk}^d(q)$ is maximized under $\yit{q_1}{}{1}=C_1,\yit{q_1}{}{0}=-C_1,\yit{q_2}{}{1}=C_2,\yit{q_2}{}{0}=-C_2$ with $C_1=\pm B$, $C_2=\pm B$ and $B$ is the bound of potential outcomes as stated in \Cref{assumption.bounded}; when $N\leq (1-\R^{p+1})^{-1}$, $\mathrm{risk}^d(q)$ is minimized under $\yit{q_1}{}{1}=C_1,\yit{q_1}{}{0}=C_1,\yit{q_2}{}{1}=C_2,\yit{q_2}{}{0}=C_2$ with $C_1=\pm B$ and $C_2=\pm B$. 

    Combining results of $q_1$ and $q_2$, we can obtain that when $N\geq (1- \max\{\R[q_1],\R[q_2]\})^{-1}$, $\mathop{\arg\max}_{\mathbb Y\in\mathcal Y}\mathcal L(1,0)=\{Y_{i,t}(\mathbf q_{(t-p):t},\mathbf z_{i,(t-p):t})\mid \yit{q_1}{}{1}=C_1,\yit{q_1}{}{0}=-C_1,\yit{q_2}{}{1}=C_2,\yit{q_2}{}{0}=-C_2\}$ with $C_1=\pm B$, $C_2=\pm B$; when $N\leq (1-\min\{\R[q_1]^{p+1},\R[q_2]^{p+1}\})^{-1}$, $\mathop{\arg\max}_{\mathbb Y\in \mathcal Y}\mathcal L(1,0)=\{Y_{i,t}(\mathbf q_{(t-p):t},\mathbf z_{i,(t-p):t})\mid \yit{q_1}{}{1}=C_1,\yit{q_1}{}{0}=C_1,\yit{q_2}{}{1}=C_2,\yit{q_2}{}{0}=C_2\}$ with $C_1=\pm B$ and $C_2=\pm B$.

    (2) Recall that
    \begin{align*}
        &N^2E\left\{\mathbf S_t(\mathbb T,z)\right\}^2\\
        =&(\R[q_1]^{-J_t}-1)\left\{\sum_{i=1}^N\yit{q_1}{z}{1}\right\}^2+(\R[q_2]^{-J_t}-1)\left\{\sum_{i=1}^N\yit{q_2}{z}{1}\right\}^2\\
        & +2\left\{\sum_{i=1}^N\yit{q_1}{z}{1}\right\}\left\{\sum_{i=1}^N\yit{q_2}{z}{1}\right\}\\
        & +\R[q_1]^{-J_t}(q_{1,z}^{-J_t}-1)\sum_{i=1}^N\left\{\yit{q_1}{z}{1}\right\}^2+\R[q_2]^{-J_t}(q_{2,z}^{-J_t}-1)\sum_{i=1}^N\left\{\yit{q_2}{z}{1}\right\}^2.
    \end{align*}
    All terms in this function are maximized under $\{\yit{q}{z}{1}:\yit{q_1}{z}{1}=\yit{q_2}{z}{1}=C_3\}$ with $C_3=\pm B$. The proof can be extended for $E\{\mathbf S_{i,t}(\mathbb T,z)\mathbf S_{i,t'}(\mathbb T,z)\}$. Combining results of $z=0$ and $z=1$, we obtain that $\mathop{\arg \max}_{\mathbb Y\in \mathcal Y}\mathcal L(0,1)=\{Y_{i,t}(\mathbf q_{(t-p):t},\mathbf z_{i,(t-p):t})\mid \yit{q_1}{}{1}=C_3,\yit{q_1}{}{0}=C_4,\yit{q_2}{}{1}=C_3,\yit{q_2}{}{0}=C_4\}$ with $C_3=\pm B$ and $C_4=\pm B$.
    
    (3) With $\mathcal L(\psi_d,\psi_s)=\psi_d\mathcal L(1,0)+\psi_s\mathcal L(0,1)$, where $\psi_d>0$ and $\psi_s>0$, the maximum is obtained when $\mathbb Y$ satisfies the conditions in both (1) and (2).
\end{proof}

Now, we can prove \Cref{theorem.r}.
\begin{proof}[Proof of \Cref{theorem.r}]
    (1) When $N\geq (1- \max\{\R[q_1],\R[q_2]\})^{-1}$ and $\mathbb Y^* \in  \arg \max_{\mathbb Y\in \mathcal Y}\mathcal L(1,0)$ as shown in \Cref{lemma.Y}, we have
    \begin{align*}
        &E\{\mathbf D_t(\mathbb T,q_1)\}^2+E\{\mathbf D_t(\mathbb T,q_2)\}^2\\
        =&4(\R[q_1]^{-J_t}+\R[q_2]^{-J_t}-2)B^2+\left\{\R[q_1]^{-J_t}(q_1^{-J_t}+\bar q_1^{-J_t}-4)+\R[q_2]^{-J_t}(q_2^{-J_t}+\bar q_2^{-J_t}-4)\right\}\frac{B^2}N,\\
        &E\{\mathbf D_t(\mathbb T,q_1)\mathbf D_{t'}(\mathbb T,q_1)\}+E\{\mathbf D_t(\mathbb T,q_2)\mathbf D_{t'}(\mathbb T,q_2)\}\\
        =&4(\R[q_1]^{-J_{t,t'}^{\circ}}+\R[q_2]^{-J_{t,t'}^{\circ}}-2)B^2+\left\{\R[q_1]^{-J_{t,t'}^{\circ}}(q_1^{-J_{t,t'}^{\circ}}+\bar q_1^{-J_{t,t'}^{\circ}}-4)+\R[q_2]^{-J_{t,t'}^{\circ}}(q_2^{-J_{t,t'}^{\circ}}+\bar q_2^{-J_{t,t'}^{\circ}}-4)\right\}\frac{B^2}N,
    \end{align*}
    while when $N\leq (1-\min\{\R[q_1]^{p+1},\R[q_2]^{p+1}\})^{-1}$, we have
    \begin{align*}
        &E\{\mathbf D_t(\mathbb T,q_1)\}^2+E\{\mathbf D_t(\mathbb T,q_2)\}^2\\
        =&\left\{\R[q_1]^{-J_t}(q_1^{-J_t}+\bar q_1^{-J_t})+\R[q_2]^{-J_t}(q_2^{-J_t}+\bar q_2^{-J_t})\right\}\frac{B^2}N,\\
        &E\{\mathbf D_t(\mathbb T,q_1)\mathbf D_{t'}(\mathbb T,q_1)\}+E\{\mathbf D_t(\mathbb T,q_2)\mathbf D_{t'}(\mathbb T,q_2)\}\\
        =&\left\{\R[q_1]^{-J_{t,t'}^{\circ}}(q_1^{-J_{t,t'}^{\circ}}+\bar q_1^{-J_{t,t'}^{\circ}})+\R[q_2]^{-J_{t,t'}^{\circ}}(q_2^{-J_{t,t'}^{\circ}}+\bar q_2^{-J_{t,t'}^{\circ}})\right\}\frac{B^2}N.
    \end{align*}
    Letting $J=J_t$ or $J=J_{t,t'}^{\circ}$, $E\{\mathbf D_t(\mathbb T,q_1)\}^2+E\{\mathbf D_t(\mathbb T,q_2)\}^2$ or $E\{\mathbf D_t(\mathbb T,q_1)\mathbf D_{t'}(\mathbb T,q_1)\}+E\{\mathbf D_t(\mathbb T,q_2)\mathbf D_{t'}(\mathbb T,q_2)\}$ becomes
    \begin{align}
    \label{eq:r1}
        4(\R[q_1]^{-J}+\R[q_2]^{-J}-2)B^2+\left\{\R[q_1]^{-J}(q_1^{-J}+\bar q_1^{-J}-4)+\R[q_2]^{-J}(q_2^{-J}+\bar q_2^{-J}-4)\right\}\frac{B^2}N
    \end{align}
    when $N\geq (1-\max\{\R[q_1],\R[q_2]\})^{-1}$ and
    \begin{align}
    \label{eq:r2}
        \left\{\R[q_1]^{-J}(q_1^{-J}+\bar q_1^{-J})+\R[q_2]^{-J}(q_2^{-J}+\bar q_2^{-J})\right\}\frac{B^2}N
    \end{align}
    when $N\leq (1-\min\{\R[q_1]^{p+1},\R[q_2]^{p+1}\})^{-1}$.
    Given design $\mathbb T$, define $\mathcal{J}_j$ as the sum of the number of time points $t$ such that $J_t = j$ and the number of pairs $(t, t')$ with $t \neq t'$ such that $J_{t,t'}^{\circ} = j$. Combining terms from $j=1$ to $p+1$, and further letting $\zeta_j=\R[q_1]^{-j}q_1^{-j} + \R[q_1]^{-j}\bar{q}_1^{-j} + \R[q_2]^{-j}q_2^{-j} + \R[q_2]^{-j}\bar{q}_2^{-j}$, we know
    \begin{align*}
    \mathop{\max}_{\mathbb Y\in\mathcal Y}\mathcal L(1,0)=\sum_{j=1}^{p+1}\mathcal J_j \left\{ 4(\R[q_1]^{-j} + \R[q_2]^{-j})B^2(1-N^{-1}) - 8B^2 + \zeta_j B^2/N \right\}
    \end{align*}
    when $N \geq (1 - \max\{\R[q_1], \R[q_2]\})^{-1}$, and
    \begin{align*}
    \mathop{\max}_{\mathbb Y\in\mathcal Y}\mathcal L(1,0)=\sum_{j=1}^{p+1}\mathcal J_j \left\{ \zeta_j B^2/N \right\}
    \end{align*}
    when $N \leq (1 - \min\{\R[q_1]^{p+1}, \R[q_2]^{p+1}\})^{-1}$.

    (2) When $\mathbb Y^* \in \arg\max_{\mathbb Y\in \mathcal Y}\mathcal L(0,1)$ as shown in \Cref{lemma.Y}, we have
    \begin{align*}
        &E\{\mathbf S_t(\mathbb T,1)\}^2+E\{\mathbf S_t(\mathbb T,0)\}^2\\
        =&2(\R[q_1]^{-J_t}+\R[q_2]^{-J_t})B^2(1-N^{-1})+(\R[q_1]^{-J_t}q_1^{-J_t}+\R[q_1]^{-J_t}\bar q_1^{-J_t}+\R[q_2]^{-J_t}q_2^{-J_t}+\R[q_2]^{-J_t}\bar q_2^{-J_t})B^2/N,
        \\
        &E\{\mathbf S_t(\mathbb T,1)\mathbf S_{t'}(\mathbb T,1)\}+E\{\mathbf S_t(\mathbb T,0)\mathbf S_{t'}(\mathbb T,0)\}\\=&2(\R[q_1]^{-J^{\circ}_{t,t'}}+\R[q_2]^{-J^{\circ}_{t,t'}})B^2(1-N^{-1})\\
        &+(\R[q_1]^{-J^{\circ}_{t,t'}}q_1^{-J^{\circ}_{t,t'}}+\R[q_1]^{-J^{\circ}_{t,t'}}\bar q_1^{-J^{\circ}_{t,t'}}+\R[q_2]^{-J^{\circ}_{t,t'}}q_2^{-J^{\circ}_{t,t'}}+\R[q_2]^{-J^{\circ}_{t,t'}}\bar q_2^{-J^{\circ}_{t,t'}})B^2/N.
    \end{align*}
    Similarly letting $J=J_t$ or $J=J_{t,t'}^{\circ}$, $E\{\mathbf S_t(\mathbb T,1)\}^2+E\{\mathbf S_t(\mathbb T,0)\}^2$ or $E\{\mathbf S_t(\mathbb T,1)\mathbf S_{t'}(\mathbb T,1)\}+E\{\mathbf S_t(\mathbb T,0)\mathbf S_{t'}(\mathbb T,0)\}$ becomes
    \begin{align}
    \label{eq:r3}
        2(\R[q_1]^{-J}+\R[q_2]^{-J})B^2(1-N^{-1})+(\R[q_1]^{-J}q_1^{-J}+\R[q_1]^{-J}\bar q_1^{-J}+\R[q_2]^{-J}q_2^{-J}+\R[q_2]^{-J}\bar q_2^{-J})B^2/N.
    \end{align}
    Further considering $\zeta_j=\R[q_1]^{-j}q_1^{-j} + \R[q_1]^{-j}\bar{q}_1^{-j} + \R[q_2]^{-j}q_2^{-j} + \R[q_2]^{-j}\bar{q}_2^{-j}$, we have
    \begin{align*}
    \mathop{\max}_{\mathbb Y\in\mathcal Y}\mathcal L(0,1)=\sum_{j=1}^{p+1}\mathcal J_j \left\{ 2(\R[q_1]^{-j} + \R[q_2]^{-j})B^2(1 - N^{-1}) + \zeta_j B^2/N \right\}.
    \end{align*}

    (3) Based on $\mathcal L(\psi_d,\psi_s)=\psi_d\mathcal L(1,0)+\psi_s\mathcal L(0,1)$, and $\mathcal L(1,0)$ and $\mathcal L(0,1)$ can reach their maximum simultaneously, we have
    \begin{align*}
    \mathop{\max}_{\mathbb Y\in\mathcal Y}\mathcal L(\psi_d,\psi_s)&=\sum_{j=1}^{p+1}\mathcal J_j \left\{ (4\psi_d+2\psi_s)(\R[q_1]^{-j} + \R[q_2]^{-j})B^2(1-N^{-1}) - 8\psi_dB^2 + \zeta_j B^2/N \right\}\\
    &=B^2 \sum_{j=1}^{p+1}\mathcal J_j \left\{ (4\psi_d+2\psi_s)(\R[q_1]^{-j} + \R[q_2]^{-j})(1-N^{-1}) - 8\psi_d + \zeta_j N^{-1} \right\}
    \end{align*}
    when $N \geq (1 - \max\{\R[q_1], \R[q_2]\})^{-1}$, and
    \begin{align*}
    \mathop{\max}_{\mathbb Y\in\mathcal Y}\mathcal L(\psi_d,\psi_s)&=\sum_{j=1}^{p+1}\mathcal J_j \left\{2\psi_s(\R[q_1]^{-j}+\R[q_2]^{-j})B^2(1-N^{-1})+ \zeta_j B^2/N\right\}\\
    &=B^2\sum_{j=1}^{p+1}\mathcal J_j \left\{2\psi_s(\R[q_1]^{-j}+\R[q_2]^{-j})(1-N^{-1})+ \zeta_j N^{-1}\right\}
    \end{align*}
    when $N \leq (1 - \min\{\R[q_1]^{p+1}, \R[q_2]^{p+1}\})^{-1}$. 
    
    (i) As $N \to \infty$, the term with $N^{-1}$ becomes negligible, allowing us to focus on the remaining terms. For each $j$, these remaining terms are proportional to $\R[q_1]^{-j} + \R[q_2]^{-j}$. This expression reaches its minimum when $\R[q_1] = \R[q_2] = 0.5$. 
    
    (ii) When $q_1 + q_2 = 1$, although the term with $N^{-1}$ is non-negligible, the term $\zeta_j$ is proportional to $\R[q_1]^{-j} + \R[q_2]^{-j}$. Similar to the previous case, this expression is minimized at $\R[q_1] = \R[q_2] = 0.5$. 
\end{proof}

\section{Proof of Theorem~\ref{theorem.minimaxdesign}}
\label{sec:sm.theorem.minimaxdesign}

When potential outcomes $\mathbb Y^*  \in \arg \max_{\mathbb Y\in\mathcal Y}\mathcal L(\psi_d,\psi_s)$ shown in \Cref{lemma.Y} and $\R[q_1]=\R[q_2]=0.5$, the expectations in \Cref{lemma.properties} can be rewritten as
\begin{align}
    \label{eq:e}
    \begin{aligned}
        E\{\mathbf D_t(\mathbb T,q)\}^2&=4(2^{J_t}-1)B^2+2^{J_t}(q^{-J_t}+\bar q^{-J_t}-4)\frac{B^2}N,\quad \text{when}\quad  N\geq 2;\\
        E\{\mathbf D_t(\mathbb T,q)\mathbf D_{t'}(\mathbb T,q)\}&=4(2^{J_{t,t'}^{\circ}}-1)B^2+2^{J_{t,t'}^{\circ}}(q^{-J_{t,t'}^{\circ}}+\bar q^{-J_{t,t'}^{\circ}}-4)\frac{B^2}N,\quad \text{when}\quad  N\geq 2;\\
        E\{\mathbf D_t(\mathbb T,q)\}^2&=2^{J_t}(q^{-J_t}+\bar q^{-J_t})\frac{B^2}N,\quad \text{when}\quad N=1;\\
        E\{\mathbf D_t(\mathbb T,q)\mathbf D_{t'}(\mathbb T,q)\}&=2^{J_{t,t'}^{\circ}}(q^{-J_{t,t'}^{\circ}}+\bar q^{-J_{t,t'}^{\circ}})\frac{B^2}N,\quad \text{when}\quad  N=1;\\
        E\{\mathbf S_t(\mathbb T,z)\}^2&=2^{J_t+1}B^2+2^{J_t}(q_{1,z}^{-J_t}+q_{2,z}^{-J_t}-2)\frac{B^2}N,\\
        E\{\mathbf S_t(\mathbb T,z)\mathbf S_{t'}(\mathbb T,z)\}&=2^{J_{t,t'}+1}B^2+2^{J_{t,t'}}(q_{1,z}^{-J_{t,t'}}+q_{2,z}^{-J_{t,t'}}-2)\frac{B^2}N.
    \end{aligned}
\end{align}
All quantities in \Cref{eq:e} increase as $J_t$ or $J_{t,t'}$ increases, given that $q_1,q_2\in(0,1)$. In the remaining part of the proof, we define $\alpha_J(q)$ and $\beta_J(z)$ to simplify the notation: 
\begin{align}
    \label{eq:ab}
    \begin{aligned}
        \alpha_J(q)&=4(2^{J}-1)B^2+2^{J}(q^{-J}+\bar q^{-J}-4)\frac{B^2}N,\quad \text{when}\quad  N\geq 2;\\
        \alpha_J(q)&=2^{J}(q^{-J}+\bar q^{-J})\frac{B^2}N,\quad \text{when}\quad N=1;\\
        \beta_J(z)&=2^{J+1}B^2+2^{J}(q_{1,z}^{-J}+q_{2,z}^{-J}-2)\frac{B^2}N.
    \end{aligned}
\end{align}
Notably, $\alpha_J(q)$ equals to $E\{\mathbf D_t(\mathbb T,q)^2\}$ when $J_t=J$ and $E\{\mathbf D_t(\mathbb T,q)\mathbf D_{t'}(\mathbb T,q)\}$ when $J_{t,t'}^{\circ}=J$. Similarly, $\beta_J(z)$ equals to $E\{\mathbf S_t(\mathbb T,z)^2\}$ when $J_t=J$ and $E\{\mathbf S_t(\mathbb T,z)\mathbf S_{t'}(\mathbb T,z)\}$ when $J_{t,t'}^{\circ}=J$.

Additionally, since the objective function $\mathcal L(\psi_d,\psi_s)$ is simply a weighted sum of $\mathcal L(1,0)$ and $\mathcal L(0,1)$, we first analyze $\mathcal L(1,0)$ and $\mathcal L(0,1)$. When the results from these two components are identical, the results can be directly extended to $\mathcal L(\psi_d,\psi_s)$.

\begin{lemma}
    \label{lemma.t1tl}
    Under Assumptions~\ref{assumption.nonanticipativity}--\ref{assumption.bounded} and $\R[q_1]=\R[q_2]=0.5$, any design $\mathbb T$ that minimizes $\max_{\mathbb Y\in \mathcal Y}\mathcal L(\psi_d,\psi_s)$ must satisfy $t_1\geq p+2$ and $t_L\leq T-p$.
\end{lemma}

\begin{proof}[Proof of \Cref{lemma.t1tl}]
    Suppose there exists a minimax optimal design $\mathbb T=\{t_0=1,t_1,t_2,\ldots,t_L\}$ such that $t_1\leq p+1$. Then we define $\tilde{\mathbb T}=\{t_0,t_2,t_3,\ldots,t_L\}$ and we want to prove that under $\tilde{\mathbb T}$, the objective function will decrease compared to $\mathbb T$, which suggests that $\mathbb T$ is not the minimax optimal design.

    (1) We first consider $E\{\mathbf D_t(\mathbb T,q)\}^2-E\{\mathbf D_t(\tilde{\mathbb T},q)\}^2$ and $E\{\mathbf S_t(\mathbb T,z)\}^2-E\{\mathbf S_t(\tilde{\mathbb T},z)\}^2$. We divide $\{p+1,\ldots,T\}$ into $[p+1,t_1+p-1]$ and $[t_1+p,T]$.

    (1.1) For $t\in[p+1,t_1+p-1]$. Under this circumstance, $t_1\in \mathcal F_{\mathbb T}^p(t)$ but $t_1\notin \mathcal F_{\tilde{\mathbb T}}^p(t)$. It means $\mathcal F_{\mathbb T}^p(t)-\{t_1\}=\mathcal F_{\tilde{\mathbb T}}^p(t)$ and $\tilde J_t=J_t-1$, where $\tilde J_t=|\mathcal F_{\tilde{\mathbb T}}^p(t)|$. Based on \Cref{eq:e} and $\tilde J_t<J_t$, we have $E\{\mathbf D_t(\mathbb T,q)\}^2-E\{\mathbf D_t(\tilde{\mathbb T},q)\}^2>0$ and $E\{\mathbf S_t(\mathbb T,z)\}^2-E\{\mathbf S_t(\tilde{\mathbb T},z)\}^2>0$. 
    
    (1.2) For $t\geq t_1+p$, $|\mathcal F_{\mathbb T}^p(t)|=|\mathcal F_{\tilde{\mathbb T}}^p(t)|$, because either (i) $\mathcal F_{\mathbb T}(t-p)=t_1$ and in this case $\mathcal F_{\tilde {\mathbb T}}(t-p)=t_0$, which means $\mathcal F_{\mathbb T}^p(t)-\{t_1\}=\mathcal F_{\tilde{\mathbb T}}^p(t)-\{t_0\}$; or (ii) $\mathcal F_{\mathbb T}(t-p)\geq t_2$, in which case $\mathcal F_{\mathbb T}^p(t)=\mathcal F_{\tilde{\mathbb T}}^p(t)$. Under both cases, we have $J_t=\tilde J_t$, which implies that $E\{\mathbf D_t(\mathbb T,q)\}^2-E\{\mathbf D_t(\tilde{\mathbb T},q)\}^2=0$ and $E\{\mathbf S_t(\mathbb T,z)\}^2-E\{\mathbf S_t(\tilde{\mathbb T},z)\}^2=0$. 
    
    Based on (1.1) and (1.2), we have
    \begin{align*}
        \sum_{t=p+1}^TE\{\mathbf D_t(\mathbb T,q)\}^2-\sum_{t=p+1}^TE\{\mathbf D_t(\tilde{\mathbb T},q)\}^2&>0,\\
        \sum_{t=p+1}^TE\{\mathbf S_t(\mathbb T,z)\}^2-\sum_{t=p+1}^TE\{\mathbf S_t(\tilde{\mathbb T},z)\}^2&>0.
    \end{align*}

    (2) Consider $E\{\mathbf D_t(\mathbb T,q)\mathbf D_{t'}(\mathbb T,q)\}-E\{\mathbf D_t(\tilde{\mathbb T},q)\mathbf D_{t'}(\tilde{\mathbb T},q)\}$ and $E\{\mathbf S_t(\mathbb T,z)\mathbf S_{t'}(\mathbb T,z)\}-E\{\mathbf S_t(\tilde{\mathbb T},z)\mathbf S_{t'}(\tilde{\mathbb T},z)\}$. 

    (2.1) For any $t$ and $t'$ such that $p+1\leq t<t'\leq t_1+p-1$, $t_1\in O_{\mathbb T}(t,t')$ but $t_1\notin O_{\tilde {\mathbb T}}(t,t')$, where $O_{\mathbb T}(t,t')=\mathcal F_{\mathbb T}^p(t)\cap \mathcal F_{\mathbb T}^p(t')$. Then, $\tilde J_{t,t'}^{\circ}=J_{t,t'}^{\circ}-1$. As a result, $E\{\mathbf D_t(\mathbb T,q)\mathbf D_t'(\mathbb T,q)\}-E\{\mathbf D_t(\tilde{\mathbb T},q)\mathbf D_t'(\tilde{\mathbb T},q)\}>0$ and $E\{\mathbf S_t(\mathbb T,z)\mathbf S_t'(\mathbb T,z)\}-E\{\mathbf S_t(\tilde{\mathbb T},z)\mathbf S_t'(\tilde{\mathbb T},z)\}>0$. 
    
    (2.2) For any $p+1\leq t<t'\leq T$ and $t'\geq t_1+p$, either (i) $\mathcal F_{\mathbb T}(t'-p)=t_1$ and in this case $\mathcal F_{\mathbb T}(t'-p)=t_0$, which implies that $O_{\mathbb T}(t,t')-\{t_1\}=O_{\tilde{\mathbb T}}(t,t')-\{t_0\}$; or (ii) $\mathcal F_{\mathbb T}(t'-p)=t_2$, in which case $O_{\mathbb T}(t,t')=O_{\tilde{\mathbb T}}(t,t')$. In both cases, $|O_{\mathbb T}(t,t')|=|O_{\tilde{\mathbb T}}(t,t')|$, which implies that $J_{t,t'}^{\circ}=\tilde J_{t,t'}^{\circ}$. Then, $E\{\mathbf D_t(\mathbb T,q)\mathbf D_t'(\mathbb T,q)\}-E\{\mathbf D_t(\tilde{\mathbb T},q)\mathbf D_t'(\tilde{\mathbb T},q)\}=0$ and $E\{\mathbf S_t(\mathbb T,z)\mathbf S_t'(\mathbb T,z)\}-E\{\mathbf S_t(\tilde{\mathbb T},z)\mathbf S_t'(\tilde{\mathbb T},z)\}=0$.

    Based on (2.1) and (2.2), we have
    \begin{align*}
        \sum_{p+1\leq t<t'\leq T}E\{\mathbf D_t(\mathbb T,q)\mathbf D_t'(\mathbb T,q)\}-\sum_{p+1\leq t<t'\leq T}E\{\mathbf D_t(\tilde{\mathbb T},q)\mathbf D_t'(\tilde{\mathbb T},q)\}&> 0,\\
        \sum_{p+1\leq t<t'\leq T}E\{\mathbf S_t(\mathbb T,z)\mathbf S_t'(\mathbb T,z)\}-\sum_{p+1\leq t<t'\leq T}E\{\mathbf S_t(\tilde{\mathbb T},z)\mathbf S_t'(\tilde{\mathbb T},z)\}&> 0.
    \end{align*}
    As a result, either $\mathcal L(0,1)$ or $\mathcal L(1,0)$ will decrease under $\tilde {\mathbb T}$ compared to $\mathbb T$. This result can extended to $\mathcal L(\psi_d,\psi_s)$ for it is the weighted sum of $\mathcal L(0,1)$ and $\mathcal L(1,0)$.
    Similarly, we can prove that $t_L\leq T-p$.
\end{proof}

\begin{lemma}
    \label{lemma.tp}
    Under Assumptions~\ref{assumption.nonanticipativity}--\ref{assumption.bounded} and $\R[q_1]=\R[q_2]=0.5$, any design $\mathbb T$ that minimizes $\max_{\mathbb Y\in \mathcal Y}\mathcal L(\psi_d,\psi_s)$ must satisfy $t_{l+1}-t_{l-1}\geq p$, for $l=1,\ldots,L$.
\end{lemma}

\begin{proof}[Proof of \Cref{lemma.tp}]
    By Lemma~\ref{lemma.t1tl}, when $l=1$, $t_2-t_0>t_1-t_0\geq p+1$; when $l=L$, $t_{L+1}-t_{L-1}>t_{L+1}-t_L\geq p+1$. For $2\leq l\leq L-1$, we will prove the result by contradiction.

    Suppose there exists a minimax optimal design $\mathbb T$, and $\exists 2\leq l\leq L-1$, s.t. $t_{l+1} - t_{l-1}\leq p-1$. Denote $\mathbb L=\{l\in\{2:L-1\}\mid t_{l+1}-t_{l-1}\leq p-1\}$ and let $l'=\max_{l \in \mathbb L} l$. We know that $l'\leq L-1$, $t_{l'+1}-t_{l'-1}\leq p-1$ and $t_{l'+2}-t_{l'}\geq p$. We then construct another design $\tilde{\mathbb T}=\{\tilde t_0=1,\tilde t_1=t_1,\ldots,\tilde t_{l'-1}=t_{l'-1},\tilde t_{l'}=t_{l'+1},\ldots,\tilde t_{L-1}=t_L\}$.

    (1) We first consider $E\{\mathbf D_t(\mathbb T,q)\}^2-E\{\mathbf D_t(\tilde{\mathbb T},q)\}^2$ and $E\{\mathbf S_t(\mathbb T,z)\}^2-E\{\mathbf S_t(\tilde{\mathbb T},z)\}^2$.

    (1.1) When $t\leq t_{l'}-1$, $\mathcal F_{\mathbb T}^p(t)=\mathcal F_{\tilde{\mathbb T}}^p(t)$ and $J_t=\tilde J_t$, which further implies that $E\{\mathbf D_t(\mathbb T,q)\}^2-E\{\mathbf D_t(\tilde {\mathbb T},q)\}^2=0$ and $E\{\mathbf S_t(\mathbb T,z)\}^2-E\{\mathbf S_t(\tilde {\mathbb T},z)\}^2=0$.

    (1.2) When $t\leq t_{l'}+p-1$, $t_{l'}\in \mathcal F_{\mathbb T}^p(t)$ but $t_{l'}\notin \mathcal F_{\tilde{\mathbb T}}^p(t)$. We then have $\mathcal F_{\mathbb T}^p(t)-\{t_{l'}\}=\mathcal F_{\tilde{\mathbb T}}^p(t)$ ($J_t=\tilde J_t+1$) and $|\mathcal F_{\tilde{\mathbb T}}^p(t)|\geq 1$ ($\tilde J_t\geq 1$). By \Cref{eq:e}, $E\{\mathbf D_t(\mathbb T,q)\}^2-E\{\mathbf D_t(\tilde {\mathbb T},q)\}^2>0$ and $E\{\mathbf S_t(\mathbb T,z)\}^2-E\{\mathbf S_t(\tilde {\mathbb T},z)\}^2>0$.

    (1.3) When $t_{l'}+p\leq t\leq T$, either $\mathcal F_{\mathbb T}(t-p)=t_{l'}$ or $\mathcal F_{\mathbb T}(t-p)\geq t_{l'+1}$. In both cases, $|\mathcal F_{\mathbb T}^p(t)|=|\mathcal F_{\tilde{\mathbb T}}^p(t)|$ ($J_t=\tilde J_t$), which implies that $E\{\mathbf D_t(\mathbb T,q)\}^2-E\{\mathbf D_t(\tilde {\mathbb T},q)\}^2=0$ and $E\{\mathbf S_t(\mathbb T,z)\}^2-E\{\mathbf S_t(\tilde {\mathbb T},z)\}^2=0$. Combing the above arguments, we have 
    \begin{align*}
        \sum_{t=p+1}^TE\{\mathbf D_t(\mathbb T,q)\}^2-\sum_{t=p+1}^TE\{\mathbf D_t(\tilde {\mathbb T},q)\}^2&>0,\\
        \sum_{t=p+1}^TE\{\mathbf S_t(\mathbb T,z)\}^2-\sum_{t=p+1}^TE\{\mathbf S_t(\tilde {\mathbb T},z)\}^2&>0.
    \end{align*}

    (2) We then consider the cross-product terms, which are associated with $O_{\mathbb T}(t,t')$ and $O_{\tilde{\mathbb T}}(t,t')$. The time period $[p+1,T]$ can be divided by $t_{l'-1}$, $t_{l'}$ and $t_{l'}+p$ into four sub time periods: $[p+1,t_{l'-1}-1]$, $[t_{l'-1},t_{l'}-1]$, $[t_{l'},t_{l'}+p-1]$ and $[t_{l'}+p,T]$. Without loss of generality, we assume $t<t'$. Define $\tilde J_{t,t'}^{\circ}=|O_{\tilde{\mathbb T}}(t,t')|$.

    (2.1) When $t\in[p+1,t_{l'-1}-1]$, for any $t'>t$, the overlapping points are unchanged, and we have $J_{t,t'}^{\circ}=\tilde J_{t,t'}^{\circ}$. It implies that $E\{\mathbf D_t(\mathbb T,q)\mathbf D_{t'}(\mathbb T,q)\}-E\{\mathbf D_t(\tilde{\mathbb T},q)\mathbf D_{t'}(\tilde{\mathbb T},q)\}=0$ and $E\{\mathbf S_t(\mathbb T,z)\mathbf S_{t'}(\mathbb T,z)\}-E\{\mathbf S_t(\tilde{\mathbb T},z)\mathbf S_{t'}(\tilde{\mathbb T},z)\}=0$.

    (2.2) When $t\in[t_{l'-1},t_{l'}-1]$ and $t'\in[t+1,t_{l'}+p-1]$, $t_{l'-1}$ is in both $O_{\mathbb T}(t,t')$ and $O_{\tilde{\mathbb T}}(t,t')$ but $t_{l'}$ is in neither, which suggests that $J_{t,t'}^{\circ}=\tilde J_{t,t'}^{\circ}$ and then $E\{\mathbf D_t(\mathbb T,q)\mathbf D_{t'}(\mathbb T,q)\}-E\{\mathbf D_t(\tilde{\mathbb T},q)\mathbf D_{t'}(\tilde{\mathbb T},q)\}=0$, $E\{\mathbf S_t(\mathbb T,z)\mathbf S_{t'}(\mathbb T,z)\}-E\{\mathbf S_t(\tilde{\mathbb T},z)\mathbf S_{t'}(\tilde{\mathbb T},z)\}=0$.

    (2.3) When $t\in[t_{l'-1},t_{l'}-1]$ and $t'\in[t_{l'}+p,t_{l'+1}+p-1]$, $t_{l'-1}$ is in $O_{\tilde{\mathbb T}}(t,t')$ but not in $O_{\mathbb T}(t,t')$, which means that $O_{\mathbb T}(t,t')=O_{\tilde{\mathbb T}}(t,t')-\{t_{l'-1}\}$ ($J_{t,t'}=\tilde J_{t,t'}-1$). Also, $\forall a\in \mathcal F_{\tilde{\mathbb T}}^p(t')$, we have $a\geq t_{l'-1}$, given that $t'\geq t_{l'}+p>t_{l'-1}+p$; $\forall a\in \mathcal F_{\tilde{\mathbb T}}^p(t)$, we have $a\leq t_{l'-1}$, given that $t\leq t_{l'}-1$; and $t_{l'-1}\in O_{\tilde{\mathbb T}}(t,t')$, given that $t_{l'-1}\in \mathcal F_{\tilde{\mathbb T}}^p(t)$ when $t=t_{l'-1}$ and $t_{l'-1}\in \mathcal F_{\tilde{\mathbb T}}^p(t')$ when $t'=t_{l'}+p$. Thus, $O_{\tilde{\mathbb T}}(t,t')=\{t_{l'-1}\}$ and $O_{\mathbb T}(t,t')=\emptyset$. Therefore,
    \begin{align*}
        E\{\mathbf D_t(\mathbb T,q)\mathbf D_{t'}(\mathbb T,q)\}-E\{\mathbf D_t(\tilde{\mathbb T},q)\mathbf D_{t'}(\tilde{\mathbb T},q)\}&=-\alpha_1(q),\\
        E\{\mathbf S_t(\mathbb T,z)\mathbf S_{t'}(\mathbb T,z)\}-E\{\mathbf S_t(\tilde{\mathbb T},z)\mathbf S_{t'}(\tilde{\mathbb T},z)\}&=-\beta_1(z).
    \end{align*}

    (2.4) When $t\in[t_{l'-1},t_{l'}-1]$ and $t'\in[t_{l'+1}+p,T]$, both $O_{\mathbb T}(t,t')$ and $O_{\mathbb{\mathbb T}}(t,t')$ are empty sets, and then $E\{\mathbf D_t(\mathbb T,q)\mathbf D_{t'}(\mathbb T,q)\}-E\{\mathbf D_t(\tilde{\mathbb T},q)\mathbf D_{t'}(\tilde{\mathbb T},q)\}=0$ and $E\{\mathbf S_t(\mathbb T,z)\mathbf S_{t'}(\mathbb T,z)\}-E\{\mathbf S_t(\tilde{\mathbb T},z)\mathbf S_{t'}(\tilde{\mathbb T},z)\}=0$.

    (2.5) When $t\in[t_{l'},t_{l'}+p-1]$ and $t'\in[t,t_{l'}+p-1]$, $t_{l'}\in O_{\mathbb T}(t,t')$ but $t_{l'}\notin O_{\tilde{\mathbb T}}(t,t')$. The other overlapping points are unchanged. Furthermore, $t_{l'-1}\in O_{\tilde{\mathbb T}}(t,t')$. As a result, $J_{t,t'}^{\circ}-1=\tilde J_{t,t'}^{\circ}\geq 1$. Further notice that $\alpha_{J+1}(q)-\alpha_J(q)$ and $\beta_{J+1}(z)-\beta_J(z)$ increase with $J$ when $J\geq 1$. Thus,
    \begin{align*}
        E\{\mathbf D_t(\mathbb T,q)\mathbf D_{t'}(\mathbb T,q)\}-E\{\mathbf D_t(\tilde{\mathbb T},q)\mathbf D_{t'}(\tilde{\mathbb T},q)\}&\geq \alpha_2(q)-\alpha_1(q),\\
        E\{\mathbf S_t(\mathbb T,z)\mathbf S_{t'}(\mathbb T,z)\}-E\{\mathbf S_t(\tilde{\mathbb T},z)\mathbf S_{t'}(\tilde{\mathbb T},z)\}&\geq \beta_2(z)-\beta_1(z).
    \end{align*}

    (2.6) When $t\in[t_{l'},t_{l'}+p-1]$ and $t'\in[t_{l'}+p,T]$, either $\mathcal F_{\mathbb T}^p(t'-p)=t_{l'}$ or $\mathcal F_{\mathbb T}^p(t'-p)\geq t_{l'+1}$. Under both cases, $|O_{\mathbb T}(t,t')|=|O_{\tilde{\mathbb T}}(t,t')|$ and as a result, $E\{\mathbf D_t(\mathbb T,q)\mathbf D_{t'}(\mathbb T,q)\}-E\{\mathbf D_t(\tilde{\mathbb T},q)\mathbf D_{t'}(\tilde{\mathbb T},q)\}=0$ and $E\{\mathbf S_t(\mathbb T,z)\mathbf S_{t'}(\mathbb T,z)\}-E\{\mathbf S_t(\tilde{\mathbb T},z)\mathbf S_{t'}(\tilde{\mathbb T},z)\}=0$.

    (2.7) When $t_{l'}+p\leq t<t'\leq T$, either $\mathcal F_{\mathbb T}^p(t'-p)=t_{l'}$ or $\mathcal F_{\mathbb T}^p(t'-p)\geq t_{l'+1}$. Under both cases, $|O_{\mathbb T}(t,t')|=|O_{\tilde{\mathbb T}}(t,t')|$ and as a result, $E\{\mathbf D_t(\mathbb T,q)\mathbf D_{t'}(\mathbb T,q)\}-E\{\mathbf D_t(\tilde{\mathbb T},q)\mathbf D_{t'}(\tilde{\mathbb T},q)\}=0$ and $E\{\mathbf S_t(\mathbb T,z)\mathbf S_{t'}(\mathbb T,z)\}-E\{\mathbf S_t(\tilde{\mathbb T},z)\mathbf S_{t'}(\tilde{\mathbb T},z)\}=0$.

    The main difference occurs in (2.3) and (2.5). Combing (2.1)--(2.7), we have 
    \begin{align*}
        &\sum_{p+1\leq t<t'\leq T}E\{\mathbf D_t(\mathbb T,q)\mathbf D_{t'}(\mathbb T,q)\}-\sum_{p+1\leq t<t'\leq T}E\{\mathbf D_t(\tilde{\mathbb T},q)\mathbf D_{t'}(\tilde{\mathbb T},q)\}\\
        =&\sum_{t_{l'-1}\leq t\leq t_l'-1,t_{l'}+p\leq t'\leq t_{l'+1}+p-1}\left[E\{\mathbf D_t(\mathbb T,q)\mathbf D_{t'}(\mathbb T,q)\}-E\{\mathbf D_t(\tilde{\mathbb T},q)\mathbf D_{t'}(\tilde{\mathbb T},q)\}\right]\\
        &+\sum_{t_{l'}\leq t<t'\leq t_{l'}+p-1}\left[E\{\mathbf D_t(\mathbb T,q)\mathbf D_{t'}(\mathbb T,q)\}-E\{\mathbf D_t(\tilde{\mathbb T},q)\mathbf D_{t'}(\tilde{\mathbb T},q)\}\right]\\
        \geq  &-(t_{l'}-t_{l'-1})(t_{l'+1}-t_{l'})\alpha_1(q)+\frac{p(p-1)}2\{\alpha_2(q)-\alpha_1(q)\}\\
        >&\frac{p(p-1)}4\{2\alpha_2(q)-3\alpha_1(q)\}>0,
    \end{align*}
    and 
    \begin{align*}
        &\sum_{p+1\leq t<t'\leq T}E\{\mathbf S_t(\mathbb T,z)\mathbf S_{t'}(\mathbb T,z)\}-\sum_{p+1\leq t<t'\leq T}E\{\mathbf S_t(\tilde{\mathbb T},z)\mathbf S_{t'}(\tilde{\mathbb T},z)\}\\
        &=\sum_{t_{l'-1}\leq t\leq t_l'-1,t_{l'}+p\leq t'\leq t_{l'+1}+p-1}\left[E\{\mathbf S_t(\mathbb T,z)\mathbf S_{t'}(\mathbb T,z)\}-E\{\mathbf S_t(\tilde{\mathbb T},z)\mathbf S_{t'}(\tilde{\mathbb T},z)\}\right]\\
        &\quad +\sum_{t_{l'}\leq t<t'\leq t_{l'}+p-1}\left[E\{\mathbf S_t(\mathbb T,z)\mathbf S_{t'}(\mathbb T,z)\}-E\{\mathbf S_t(\tilde{\mathbb T},z)\mathbf S_{t'}(\tilde{\mathbb T},z)\}\right]\\
        &\geq -(t_{l'}-t_{l'-1})(t_{l'+1}-t_{l'})\beta_1(z)+\frac{p(p-1)}2\{\beta_2(z)-\beta_1(z)\}>\frac{p(p-1)}4\{2\beta_2(z)-3\beta_1(z)\}\\
        &>0.
    \end{align*}
    The inequality holds as $(t_{l'}-t_{l'-1})(t_{l'+1}-t_{l'})< p(p-1)/4$ when $t_{l'+1}-t_{l'-1}\leq p-1$. Therefore, $\mathrm{risk}^d(\mathbb T,q)>\mathrm{risk}^d(\tilde{\mathbb T},q)$ and $\mathrm{risk}^s(\mathbb T,q)>\mathrm{risk}^s(\tilde{\mathbb T},z)$, which is a contradiction.

    As a result, either $\mathcal L(0,1)$ or $\mathcal L(1,0)$ will decrease under $\tilde {\mathbb T}$ compared to $\mathbb T$. This result can extended to $\mathcal L(\psi_d,\psi_s)$ for it is the weighted sum of $\mathcal L(0,1)$ and $\mathcal L(1,0)$.
\end{proof}

Based on \Cref{lemma.t1tl}--\ref{lemma.tp}, any minimax optimal design must satisfy $t_1\geq p+2$, $t_L\leq T-p$ and $t_{l+1}-t_{l-1}\geq p$, for $l=1,\ldots,L$. Define $\mathcal T=\{\mathbb T\mid t_1\geq p+2,t_L\leq T-p,\text{ and }t_{l+1}-t_{l-1}\geq p\text{ for }l=1,\ldots,L\}$. Then solving the minimax problem $\min_{\mathbb T}\max_{\mathbb Y\in\mathcal L}\mathcal L(\psi_d,\psi_s)$ is equivalent to solving $\min_{\mathbb T\in \mathcal T}\max_{\mathbb Y\in\mathcal L}\mathcal L(\psi_d,\psi_s)$.

\begin{lemma}
    \label{lemma.riskexpression}
    Under Assumptions~\ref{assumption.nonanticipativity}--\ref{assumption.bounded} and $\R[q_1]=\R[q_2]=0.5$, when $\mathbb T\in\mathcal T$ and $\mathbb Y\in \arg \max_{\mathbb Y\in\mathcal Y}\mathcal L(\psi_d,\psi_s)$, $\mathrm{risk}^d(q)$ satisfies
    \begin{align*}
        &(T-p)^2\mathrm{risk}^d(q)
        =\left\{\sum_{l=0}^{L}(t_{l+1}-t_l)^2+(L-1)p^2+2p(t_L-t_1)\right\}\alpha_1(q)\\
        &+Lp^2\{\alpha_2(q)-3\alpha_1(q)\}+\left(\sum_{l=2}^{L}\left[\{(p-t_l+t_{l-1})^+\}^2\right]\right)\{\alpha_3(q)-2\alpha_2(q)+\alpha_1(q)\}
    \end{align*}
    where $\alpha_J(q)$ ($J=1,2,3$) are defined in \Cref{eq:ab} and 
    $\mathrm{risk}^s(z)$ satisfies
    \begin{align*}
        &(T-p)^2\mathrm{risk}^s(z)
        =\left\{\sum_{l=0}^{L}(t_{l+1}-t_l)^2+(L-1)p^2+2p(t_L-t_1)\right\}\beta_1(z)\\
        &+Lp^2\{\beta_2(z)-2\beta_1(z)\}+\left(\sum_{l=2}^{L}\left[\{(p-t_l+t_{l-1})^+\}^2\right]\right)\{\beta_3(z)-2\beta_2(z)+\beta_1(z)\},
    \end{align*}
    where $\beta_J(z)$ ($J=1,2,3$) are defined in \Cref{eq:ab}.
\end{lemma}

\begin{proof}[Proof of \Cref{lemma.riskexpression}]
    When $\mathbb T\in \mathcal T$, $\mathrm{risk}^d(q)$ satisfies
    \begin{align*}
        &(T-p)^2\mathrm{risk}^d(q)=(T-p)^2E\{\taudhat-\taud\}^2=E\left\{\sum_{t=p+1}^T\mathbf D_{t}(\mathbb T,q)\sum_{t'=p+1}^T\mathbf D_{t'}(\mathbb T,q)\right\}\\
        =&\sum_{\substack{p+1 \leq t, t' \leq T \\ \min \left\{t, t'\right\} \leq t_1-1}} E\left\{\mathbf{D}_t(\mathbb T,q) \mathbf{D}_{t'}(\mathbb T,q)\right\}+\sum_{l=1}^{L-1}\sum_{\substack{t_l\leq t,t'\leq T\\ \min(t,t')\leq t_{l+1}-1}}E\left\{\mathbf D_t(\mathbb T,q)\mathbf D_{t'}(\mathbb T,q)\right\}\\
        &+\sum_{t_L\leq t,t'\leq T}E\{\mathbf D_t(\mathbb T,q)\mathbf D_{t'}(\mathbb T,q)\},
    \end{align*}
    and $\mathrm{risk}^s(z)$ satisifies
    \begin{align*}
        &(T-p)^2\mathrm{risk}^s(z)=(T-p)^2E\{\taushat-\taus\}^2=E\left\{\sum_{t=p+1}^T\mathbf S_{t}(\mathbb T,z)\sum_{t'=p+1}^T\mathbf S_{t'}(\mathbb T,z)\right\}\\
        =&\sum_{\substack{p+1 \leq t, t' \leq T \\ \min \left\{t, t'\right\} \leq t_1-1}} E\left\{\mathbf{S}_t(\mathbb T,z) \mathbf{S}_{t'}(\mathbb T,z)\right\}+\sum_{l=1}^{L-1}\sum_{\substack{t_l\leq t,t'\leq T\\ \min(t,t')\leq t_{l+1}-1}}E\left\{\mathbf S_t(\mathbb T,z)\mathbf S_{t'}(\mathbb T,z)\right\}\\
        &+\sum_{t_L\leq t,t'\leq T}E\{\mathbf S_t(\mathbb T,z)\mathbf S_{t'}(\mathbb T,z)\}.
    \end{align*}
    We consider the risk function under three cases. Notably, by \Cref{lemma.t1tl} and \Cref{lemma.tp}, any minimax optimal design $\mathbb T\in \mathcal T$ satisfies $t_1\geq p+2$, $t_L\leq T-p$ and $t_{l+1}-t_{l-1}\geq p$, for $l=1,\ldots,L$.

    (1) For any $t,t'$ such that $p+1\leq \min(t,t')\leq t_1-1$, $p+1\leq \max(t,t')\leq t_1+p-1$, we have $J_{t,t'}^{\circ}=1$; for any $t,t'$ such that $p+1\leq \min(t,t')\leq t_1-1$, $t_1+p\leq \max(t,t')\leq T$, we have $J_{t,t'}^{\circ}=0$, $E\{\mathbf D_t(\mathbb T,q)\mathbf D_{t'}(\mathbb T,q)\}=0$, and $E\{\mathbf S_t(\mathbb T,z)\mathbf S_{t'}(\mathbb T,z)\}=0$. Therefore,
    \begin{align*}
        \sum_{\substack{p+1 \leq t, t' \leq T \\ \min \left\{t, t'\right\} \leq t_1-1}} E\left\{\mathbf{D}_t(\mathbb T,q) \mathbf{D}_{t'}(\mathbb T,q)\right\}=\{(t_1-1)^2-p^2\}\alpha_1(q),\\
        \sum_{\substack{p+1 \leq t, t' \leq T \\ \min \left\{t, t'\right\} \leq t_1-1}} E\left\{\mathbf{S}_t(\mathbb T,z) \mathbf{S}_{t'}(\mathbb T,z)\right\}=\{(t_1-1)^2-p^2\}\beta_1(z).
    \end{align*}
    
    (2) For any $l\in[L-1]$, we consider $t_l-t_{l-1}$ and $t_{l+1}-t_l$. We need to consider the values of $E\{\mathbf D_t(\mathbb T,q)\mathbf D_{t'}(\mathbb T,q)\}$ and $E\{\mathbf S_t(\mathbb T,z)\mathbf S_{t'}(\mathbb T,z)\}$ when $t_l\leq \min(t,t')\leq t_{l+1}-1$ and $t_l\leq \max(t,t')\leq T$.

    (2.1) When $t_l-t_{l-1}\geq p$ and $t_{l+1}-t_l\geq p$, for $t_l\leq \min(t,t')\leq t_l+p-1$ and $t_l\leq \max(t,t')\leq t_l+p-1$, we have $J_{t,t'}^{\circ}=2$; for $t_l\leq \min(t,t')\leq t_{l+1}-1$ and $t_l+p\leq \max(t,t')\leq t_{l+1}+p-1$, we have $J_{t,t'}^{\circ}=1$; for $t_l\leq \min(t,t')\leq t_{l+1}-1$ and $t_{l+1}+p\leq \max(t,t')\leq T$, we have $J_{t,t'}^{\circ}=0$. In these cases, we have
    \begin{align*}
        \sum_{\substack{t_l\leq t,t'\leq T\\ \min(t,t')\leq t_{l+1}-1}}E\{\mathbf D_t(\mathbb T,q)\mathbf D_{t'}(\mathbb T,q)\}=\{(p+t_{l+1}-t_l)^2-2p^2\}\alpha_1(q)+p^2\alpha_2(q),\\
        \sum_{\substack{t_l\leq t,t'\leq T\\ \min(t,t')\leq t_{l+1}-1}}E\{\mathbf S_t(\mathbb T,z)\mathbf S_{t'}(\mathbb T,z)\}=\{(p+t_{l+1}-t_l)^2-2p^2\}\beta_1(z)+p^2\beta_2(z).
    \end{align*}
    
    (2.2) When $t_l-t_{l-1}\geq p$ and $t_{l+1}-t_l< p$, for $t_l\leq \min(t,t')\leq t_{l+1}-1$ and $t_l\leq \max(t,t')\leq t_l+p+1$, we have $J_{t,t'}^{\circ}=2$; for $t_l\leq \min(t,t')\leq t_{l+1}-1$ and $t_l+p\leq \max(t,t')\leq t_{l+1}+p+1$, we have $J_{t,t'}^{\circ}=1$; for $t_l\leq \min(t,t')\leq t_{l+1}-1$ and $t_{l+1}-p\leq \max(t,t')\leq T$, we have $J_{t,t'}^{\circ}=0$. In these cases, we have
    \begin{align*}
        &\sum_{\substack{t_l\leq t,t'\leq T\\ \min(t,t')\leq t_{l+1}-1}}E\{\mathbf D_t(\mathbb T,q)\mathbf D_{t'}(\mathbb T,q)\}\\
        =&\{(p+t_{l+1}-t_l)^2-2p^2+(p-t_{l+1}+t_l)^2\}\alpha_1(q)+\{p^2-(p-t_{l+1}+t_l)^2\}\alpha_2(q),\\
        &\sum_{\substack{t_l\leq t,t'\leq T\\ \min(t,t')\leq t_{l+1}-1}}E\{\mathbf S_t(\mathbb T,z)\mathbf S_{t'}(\mathbb T,z)\}\\
        =&\{(p+t_{l+1}-t_l)^2-2p^2+(p-t_{l+1}+t_l)^2\}\beta_1(z)+\{p^2-(p-t_{l+1}+t_l)^2\}\beta_2(z).
    \end{align*}

    (2.3) When $t_l-t_{l-1}<p$ and $t_{l+1}-t_l\geq p$, for $t_l\leq \min(t,t')\leq t_{l-1}+p-1$ and $t_l\leq \max(t,t')\leq t_{l-1}+p-1$, we have $J_{t,t'}^{\circ}=3$; for $t_l\leq \min(t,t')\leq t_l+p-1$ and $t_{l-1}+p\leq \max(t,t')\leq t_l+p-1$, we have $J_{t,t'}^{\circ}=2$; for $t_l\leq \min(t,t')\leq t_{l+1}-1$ and $t_l+p\leq \max(t,t')\leq t_{l+1}+p-1$, we have $J_{t,t'}^{\circ}=1$; for $t_l\leq \min(t,t')\leq t_{l+1}-1$ and $t_{l+1}+p\leq \max(t,t')\leq T$, we have $J_{t,t'}^{\circ}=0$. In these cases, we have
    \begin{align*}
        &\sum_{\substack{t_l\leq t,t'\leq T\\ \min(t,t')\leq t_{l+1}-1}}E\{\mathbf D_t(\mathbb T,q)\mathbf D_{t'}(\mathbb T,q)\}\\
        =&\{(p+t_{l+1}-t_l)^2-2p^2\}\alpha_1(q)+\{p^2-(p-t_l+t_{l-1})^2\}\alpha_2(q)+(p-t_l+t_{l-1})^2\alpha_3(q),\\
        &\sum_{\substack{t_l\leq t,t'\leq T\\ \min(t,t')\leq t_{l+1}-1}}E\{\mathbf S_t(\mathbb T,z)\mathbf S_{t'}(\mathbb T,z)\}\\
        =&\{(p+t_{l+1}-t_l)^2-2p^2\}\beta_1(z)+\{p^2-(p-t_l+t_{l-1})^2\}\beta_2(z)+(p-t_l+t_{l-1})^2\beta_3(z).
    \end{align*}

    (2.4) When $t_l-t_{l-1}<p$ and $t_{l+1}-t_l<p$, for $t_l\leq \min(t,t')\leq t_{l-1}+p-1$ and $t_l\leq \max(t,t')\leq t_{l-1}+p-1$, we have $J_{t,t'}^{\circ}=3$; for $t_l\leq \min(t,t')\leq t_{l+1}-1$ and $t_{l-1}+p\leq \max(t,t')\leq t_l+p-1$, we have $J_{t,t'}^{\circ}=2$; for $t_l\leq \min(t,t')\leq t_{l+1}-1$ and $t_l+p\leq \max(t,t')\leq t_{l+1}+p-1$, we have $J_{t,t'}^{\circ}=1$; for $t_l\leq \min(t,t')\leq t_{l+1}-1$ and $t_{l+1}-p\leq \max(t,t')\leq T$, we have $J_{t,t'}^{\circ}=0$. In these cases, we have
    \begin{align*}
        &\sum_{\substack{t_l\leq t,t'\leq T\\ \min(t,t')\leq t_{l+1}-1}}E\{\mathbf D_t(\mathbb T,q)\mathbf D_{t'}(\mathbb T,q)\}
        =\{(p+t_{l+1}-t_l)^2-2p^2+(p-t_{l+1}+t_l)^2\}\alpha_1(q)\\
        &+\{p^2-(p-t_l+t_{l-1})^2-(p-t_{l+1}+t_l)^2\}\alpha_2(q)+(p-t_l+t_{l-1})^2\alpha_3(q),\\
        &\sum_{\substack{t_l\leq t,t'\leq T\\ \min(t,t')\leq t_{l+1}-1}}E\{\mathbf S_t(\mathbb T,z)\mathbf S_{t'}(\mathbb T,z)\}
        =\{(p+t_{l+1}-t_l)^2-2p^2+(p-t_{l+1}+t_l)^2\}\beta_1(z)\\
        &+\{p^2-(p-t_l+t_{l-1})^2-(p-t_{l+1}+t_l)^2\}\beta_2(z)+(p-t_l+t_{l-1})^2\beta_3(z).
    \end{align*}

    Under these four sub-circumstances (2.1)--(2.4), we have, for any $l=1,\ldots,L-1$,
    \begin{align*}
        &\sum_{\substack{t_l\leq t,t'\leq T\\ \min(t,t')\leq t_{l+1}-1}}E\{\mathbf D_t(\mathbb T,q)\mathbf D_{t'}(\mathbb T,q)\}
        =[(p+t_{l+1}-t_l)^2-2p^2+\{(p-t_{l+1}+t_l)^+\}^2]\alpha_1(q)\\
        &+[p^2-\{(p-t_l+t_{l-1})^+\}^2-\{(p-t_{l+1}+t_l)^+\}^2]\alpha_2(q)+\{(p-t_l+t_{l-1})^+\}^2\alpha_3(q),
    \end{align*}
    and 
    \begin{align*}
        &\sum_{\substack{t_l\leq t,t'\leq T\\ \min(t,t')\leq t_{l+1}-1}}E\{\mathbf S_t(\mathbb T,z)\mathbf S_{t'}(\mathbb T,z)\}
        =[(p+t_{l+1}-t_l)^2-2p^2+\{(p-t_{l+1}+t_l)^+\}^2]\beta_1(z)\\
        &+[p^2-\{(p-t_l+t_{l-1})^+\}^2-\{(p-t_{l+1}+t_l)^+\}^2]\beta_2(z)+\{(p-t_l+t_{l-1})^+\}^2\beta_3(z).
    \end{align*}

    (3) Finally, we consider $t_L\leq \min(t,t')\leq T$ and $t_L\leq \max(t,t')\leq T$.

    (3.1) When $t_L-t_{L-1}\geq p$, for $t_L\leq \min(t,t')\leq t_L+p-1$ and $t_L\leq \max(t,t')\leq t_L+p-1$, we have $J_{t,t'}^{\circ}=2$; for $t_L\leq \min(T,T')\leq T$ and $t_L\leq \max(t,t')\leq T$, we have $J_{t,t'}^{\circ}=1$; otherwise, $J_{t,t'}^{\circ}=0$. In these cases, we have
    \begin{align*}
        \sum_{t_L\leq t,t'\leq T}E\{\mathbf D_t(\mathbb T,q)\mathbf D_{t'}(\mathbb T,q)\}=\{(T+1-t_L)^2-p^2\}\alpha_1(q)+p^2\alpha_2(q),\\
        \sum_{t_L\leq t,t'\leq T}E\{\mathbf S_t(\mathbb T,z)\mathbf S_{t'}(\mathbb T,z)\}=\{(T+1-t_L)^2-p^2\}\beta_1(z)+p^2\beta_2(z).
    \end{align*}

    (3.2) When $t_L-t_{L-1}<p$, for $t_L\leq \min(t,t')\leq t_{L-1}+p-1$ and $t_L\leq \max(t,t')\leq t_{L-1}+p-1$, we have $J_{t,t'}^{\circ}=3$; for $t_L\leq \min(t,t')\leq t_L+p-1$ and $t_{L-1}+p\leq \max(t,t')\leq t_L+p-1$, we have $J_{t,t'}^{\circ}=2$; for $t_L\leq \min(t,t')\leq T$ and $t_L+p\leq \max(t,t')\leq T$, we have $J_{t,t'}^{\circ}=1$; otherwise, $J_{t,t'}^{\circ}=0$. In these cases, we have
    \begin{align*}
        &\sum_{t_L\leq t,t'\leq T}E\{\mathbf D_t(\mathbb T,q)\mathbf D_{t'}(\mathbb T,q)\}\\
        =&\{(T+1-t_L)^2-p^2\}\alpha_1(q)+\{p^2-(p-t_L+t_{L-1})^2\}\alpha_2(q)+(p-t_L+t_{L-1})^2\alpha_3(q).
    \end{align*}

    Combing (3.1)--(3.2), we have
    \begin{align*}
        &\sum_{t_L\leq t,t'\leq T}E\{\mathbf D_t(\mathbb T,q)\mathbf D_{t'}(\mathbb T,q)\}
        =\{(T+1-t_L)^2-p^2\}\alpha_1(q)\\
        &+[p^2-\{(p-t_L+t_{L-1})^+\}^2]\alpha_2(q)
        +\{(p-t_L+t_{L-1})^+\}^2\alpha_3(q),
    \end{align*}
    \begin{align*}
        &\sum_{t_L\leq t,t'\leq T}E\{\mathbf S_t(\mathbb T,z)\mathbf S_{t'}(\mathbb T,z)\}
        =\{(T+1-t_L)^2-p^2\}\beta_1(z)\\
        &+[p^2-\{(p-t_L+t_{L-1})^+\}^2]\beta_2(z)
        +\{(p-t_L+t_{L-1})^+\}^2\beta_3(z).
    \end{align*}

    Therefore,
    \begin{align*}
        &(T-p)^2\mathrm{risk}^d(q)
        =\left\{\sum_{l=0}^{L}(t_{l+1}-t_l)^2+(L-1)p^2+2p(t_L-t_1)\right\}\alpha_1(q)\\
        &+Lp^2\{\alpha_2(q)-2\alpha_1(q)\}
        +\left(\sum_{l=1}^{L-1}\left[\{(p-t_{l+1}+t_{l})^+\}^2\right]\right)\{\alpha_3(q)-2\alpha_2(q)+\alpha_1(q)\},
    \end{align*}
    \begin{align*}
        &(T-p)^2\mathrm{risk}^s(z)
        =\left\{\sum_{l=0}^{L}(t_{l+1}-t_l)^2+(L-1)p^2+2p(t_L-t_1)\right\}\beta_1(z)\\
        &+Lp^2\{\beta_2(z)-2\beta_1(z)\}
        +\left(\sum_{l=1}^{L-1}\left[\{(p-t_{l+1}+t_{l})^+\}^2\right]\right)\{\beta_3(z)-2\beta_2(z)+\beta_1(z)\}.
    \end{align*}
    Note that, when $N\geq 2$,
    \begin{align*}
        \alpha_1(q)&=\frac{B^2}N\{4N+2q^{-1}+2\bar q^{-1}-8\}>0,\\
        \alpha_2(q)-2\alpha_1(q)&=\frac{B^2}{N}\left\{4N+4q^{-2}+4\bar q^{-2}-4q^{-1}-4\bar q^{-1}\right\}>0, \\
        \alpha_3(q)-2\alpha_2(q)+\alpha_1(q)&=\frac{B^2}{N}\left\{2q^{-1}\left(2q^{-1}-1\right)^2+2\bar q^{-1}\left(2\bar q^{-1}-1\right)^2+8N\right\}\geq 0,
    \end{align*}
    and when $N=1$,
    \begin{align*}
        \alpha_1(q)&=\frac{B^2}N\{2q^{-1}+2\bar q^{-1}\}>0,\\
        \alpha_2(q)-2\alpha_1(q)&=\frac{B^2}{N}\left\{4q^{-2}+4\bar q^{-2}-4q^{-1}-4\bar q^{-1}\right\}>0, \\
        \alpha_3(q)-2\alpha_2(q)+\alpha_1(q)&=\frac{B^2}{N}\left\{2q^{-1}\left(2q^{-1}-1\right)^2+2\bar q^{-1}\left(2\bar q^{-1}-1\right)^2\right\}\geq 0,
    \end{align*}
    \begin{align*}
        \beta_1(z)&=\frac{B^2}N\{4N+2q_{1,z}^{-1}+2q_{2,z}^{-1}-4\}>0,\\
        \beta_2(z)-2\beta_1(z)&=\frac{4B^2}N\left\{q_{1,z}^{-1}(q_{1,z}^{-1}-1)+q_{2,z}^{-1}(q_{2,z}^{-1}-1) \right\}>0,\\
        \beta_3(z)-2\beta_2(z)+\beta_1(z)&=\frac{B^2}{N}\left\{2q_{1,z}^{-1}\left(2q_{1,z}^{-1}-1\right)^2+2q_{2,z}^{-1}\left(2q_{2,z}^{-1}-1\right)^2+4(N-1) \right\}\geq 0.
    \end{align*}
\end{proof}

Now, we can prove \Cref{theorem.minimaxdesign}.
\begin{proof}[Proof of Theorem~\ref{theorem.minimaxdesign}]
    When $\mathbb T\in \mathcal T$ and $\R[q_1]=\R[q_2]=0.5$,  $\max_{\mathbb Y\in\mathcal Y}\mathcal L(1,0)$ is given by
    \begin{align*}
        \left(\sum_{l=0}^{L}\left\{(t_{l+1}-t_l)^2+(L-1)p^2+2p(t_L-t_1)\right\}\gamma_1^d+Lp^2\gamma_2^d+\sum_{l=2}^{L}\left[\{(p-t_l+t_{l-1})^+\}^2\right]\gamma_3^d\right)B^2,
    \end{align*}
    where 
    \begin{align}
    \begin{aligned}
    \label{gamma.d}
        \gamma_1^d&=\alpha_1(q_1)+\alpha_1(q_2)= \frac{1}{N(T-p)^2} \{8N+2q_1^{-1}+2\bar q_1^{-1}+2q_2^{-1}+2\bar q_2^{-1}-16\},\\
        \gamma_2^d&=\{\alpha_2(q_1)-2\alpha_1(q_1)\}+\{\alpha_2(q_2)-2\alpha_1(q_2)\}\\
        &=\frac{1}{N(T-p)^2}\left\{8N+4q_1^{-2}+4\bar q_1^{-2}-4q_1^{-1}-4\bar q_1^{-1}+4q_2^{-2}+4\bar q_2^{-2}-4q_2^{-1}-4\bar q_2^{-1}\right\},\\
        \gamma_3^d&=\{\alpha_3(q_1)-2\alpha_2(q_1)+\alpha_1(q_1)\}+\{\alpha_3(q_2)-2\alpha_2(q_2)+\alpha_1(q_2)\}=\frac{1}{N(T-p)^2}\\
        &\quad \times\left\{2q_1^{-1}(2q_1^{-1}-1)^2+2\bar q_1^{-1}(2\bar q_1^{-1}-1)^2+2q_2^{-1}(2q_2^{-1}-1)^2+2\bar q_2^{-1}(2\bar q_2^{-1}-1)^2+16N\right\},
    \end{aligned}
    \end{align}
    when $N\geq 2$, and
    \begin{align}
    \begin{aligned}
    \label{gamma.d0}
        \gamma_1^d&=\alpha_1(q_1)+\alpha_1(q_2)= \frac{1}{N(T-p)^2} \{2q_1^{-1}+2\bar q_1^{-1}+2q_2^{-1}+2\bar q_2^{-1}\},\\
        \gamma_2^d&=\{\alpha_2(q_1)-2\alpha_1(q_1)\}+\{\alpha_2(q_2)-2\alpha_1(q_2)\}\\
        &=\frac{1}{N(T-p)^2}\left\{4q_1^{-2}+4\bar q_1^{-2}-4q_1^{-1}-4\bar q_1^{-1}+4q_2^{-2}+4\bar q_2^{-2}-4q_2^{-1}-4\bar q_2^{-1}\right\},\\
        \gamma_3^d&=\{\alpha_3(q_1)-2\alpha_2(q_1)+\alpha_1(q_1)\}+\{\alpha_3(q_2)-2\alpha_2(q_2)+\alpha_1(q_2)\}=\frac{1}{N(T-p)^2}\\
        &\quad \times\left\{2q_1^{-1}(2q_1^{-1}-1)^2+2\bar q_1^{-1}(2\bar q_1^{-1}-1)^2+2q_2^{-1}(2q_2^{-1}-1)^2+2\bar q_2^{-1}(2\bar q_2^{-1}-1)^2\right\},
    \end{aligned}
    \end{align}
    when $N=1$.

    Similarly, when $\mathbb T\in \mathcal T$ and $\R[q_1]=\R[q_2]=0.5$, $ \max_{\mathbb Y\in\mathcal Y}\mathcal L(0,1)$ is given by
    \begin{align*}
        &\left(\sum_{l=0}^{L}\left\{(t_{l+1}-t_l)^2+(L-1)p^2+2p(t_L-t_1)\right\}\gamma_1^s+Lp^2\gamma_2^s+\sum_{l=2}^{L}\left[\{(p-t_l+t_{l-1})^+\}^2\right]\gamma_3^s\right)B^2,
    \end{align*}
    where
    \begin{align}
    \begin{aligned}
    \label{gamma.s}
        \gamma_1^s&=\beta_1(1)+\beta_1(0)=\frac{1}{N(T-p)^2}\{8N+2q_1^{-1}+2\bar q_1^{-1}+2q_2^{-1}+2\bar q_2^{-1}-8\},\\
        \gamma_2^s&=\{\beta_2(1)-2\beta_1(1)\}+\{\beta_2(0)-2\beta_1(0)\}\\
        &=\frac{4}{N(T-p)^2}\left\{q_1^{-1}(q_1^{-1}-1)+\bar q_1^{-1}(\bar q_1^{-1}-1)+q_2^{-1}(q_2^{-1}-1)+\bar q_2^{-1}(\bar q_2^{-1}-1)\right\},\\
        \gamma_3^s&=\{\beta_3(1)-2\beta_2(1)+\beta_1(1)\}+\{\beta_3(0)-2\beta_2(0)+\beta_1(0)\}=\frac{1}{N(T-p)^2}\\
        &\quad \times\Big\{2q_1^{-1}\left(2q_1^{-1}-1\right)^2+2\bar q_1^{-1}\left(2\bar q_1^{-1}-1\right)^2+2q_2^{-1}\left(2q_2^{-1}-1\right)^2+2\bar q_2^{-1}\left(2\bar q_2^{-1}-1\right)^2\\
        &\qquad +4(N-1)\Big\}.
    \end{aligned}
    \end{align}
    The weighted sum $\max_{\mathbb Y\in\mathcal Y}\mathcal L(\psi_d,\psi_s)$, is given by
    \begin{align*}
        \left(\sum_{l=0}^{L}\left\{(t_{l+1}-t_l)^2+(L-1)p^2+2p(t_L-t_1)\right\}\gamma_1^*+Lp^2\gamma_2^*+\sum_{l=2}^{L}\left[\{(p-t_l+t_{l-1})^+\}^2\right]\gamma_3^*\right)B^2,
    \end{align*}
    with $\gamma_J^*=\psi_d\gamma_J^d+\psi_s\gamma_J^s$ for $J=1,2,3$. Based on $B^2>0$, obtaining the minimax optimal design is then equivalent to minimizing
    \begin{align*}
        &\max_{\mathbb Y\in\mathcal Y}\mathcal L(\psi_d,\psi_s)/B^2\\
        =&\sum_{l=0}^{L}\left\{(t_{l+1}-t_l)^2+(L-1)p^2+2p(t_L-t_1)\right\}\gamma_1^*+Lp^2\gamma_2^*+\sum_{l=2}^{L}\left[\{(p-t_l+t_{l-1})^+\}^2\right]\gamma_3^*.
    \end{align*}
\end{proof}

To simplify the notation, we define
\begin{align}
\label{eq:theta}
    \theta^*=\frac{\gamma_2^*}{\gamma_1^*}=\frac{4N\psi_d I(N\geq 2)+2q_1^{-2}+2\bar q_1^{-2}-2q_1^{-1}-2\bar q_1^{-1}+2q_2^{-2}+2\bar q_2^{-2}-2q_2^{-1}-2\bar q_2^{-1}}{(4N-4-4\psi_d)I(N\geq 2)+(q_1^{-1}+\bar q_1^{-1}+q_2^{-1}+\bar q_2^{-1})},
\end{align}
where $I(\cdot)$ is the indicator function. Then $\max_{\mathbb Y\in\mathcal Y}\mathcal L(\psi_d,\psi_s)/B^2$ can be rewritten as
\begin{align}
\label{eq:target}
    \left\{\sum_{l=0}^{L}(t_{l+1}-t_l)^2+(L-1+\theta^*L)p^2+2p(t_L-t_1)\right\}\gamma_1^*B^2+\left(\sum_{l=2}^{L}\left[\{(p-t_l+t_{l-1})^+\}^2\right]\right)\gamma_3^*B^2.
\end{align}

\section{A polynomial--time algorithm to solve the integer optimization problem}
\label{sec:algorithm}

In this section, we propose a polynomial--time algorithm to solve the integer optimization problem stated in \Cref{theorem.minimaxdesign}.

\begin{lemma}
\label{lemma.discussion}
    Under Assumptions \ref{assumption.nonanticipativity}--\ref{assumption.bounded} and with $\R[q_1]=\R[q_2]=0.5$, any design minimizing
    \begin{align}
    \label{eq:discussion}
        \sum_{l=0}^{L}(t_{l+1}-t_l)^2+(L-1+\theta^*L)p^2+2p(t_L-t_1)
    \end{align}
    must satisfy
    \begin{align*}
        |(t_1-t_0)-(t_{L+1}-t_L)|\leq 1,\quad |(t_{l+1}-t_l)-(t_{l'+1}-t_{l'})|\leq 1,\quad \forall 1\leq l, l'\leq L-1.
    \end{align*}
\end{lemma}

\begin{proof}[Proof of \Cref{lemma.discussion}]
    We will prove the lemma by contradiction.

    (1) Assume there exists a design $\mathbb{T}$ such that $ |(t_1-t_0)-(t_{L+1}-t_L)| \geq 2 $. We will construct a new design $\tilde{\mathbb{T}}$ with a lower value of \Cref{eq:discussion}. Without loss of generality, let us assume $ (t_1-t_0)-(t_{L+1}-t_L) \geq 2 $. The case where $ (t_{L+1}-t_L)-(t_1-t_0) \geq 2 $ can be addressed similarly. Consider the design $\tilde{\mathbb{T}}=\{\tilde{t}_0=1,\tilde{t}_1=t_1-1,\tilde{t}_2=t_2-1,\ldots,\tilde{t}_{L-1}=t_{L-1}-1,\tilde{t}_L=t_L-1\}$. Notably, we have $ \tilde{t}_{l+1}-\tilde{t}_l=t_{l+1}-t_l $ for $ l=1,\ldots,L-1 $. The change in the first term of \Cref{eq:discussion} is given by:
    \begin{align*}
        &\sum_{l=0}^{L}(\tilde{t}_{l+1}-\tilde{t}_l)^2 - \sum_{l=0}^{L}(t_{l+1}-t_l)^2\\
        =&\{(\tilde{t}_{L+1}-\tilde{t}_L)^2+(\tilde{t}_1-\tilde{t}_0)^2\} - \{(t_{L+1}-t_L)^2+(t_1-t_0)^2\}\\
        =&\{(t_{L+1}-t_L+1)^2+(t_1-t_0-1)^2\} - \{(t_{L+1}-t_L)^2+(t_1-t_0)^2\}\\
        =&-2\{(t_1-t_0)-(t_{L+1}-t_L)-1\}< 0.
    \end{align*}
     The second and third terms remain unchanged because $L$ remains unchanged and $\tilde{t}_L-\tilde{t}_1=t_L-t_1$. Thus, there exists a design $\tilde{\mathbb{T}}$ with a lower value of \Cref{eq:discussion} than $\mathbb{T}$ when $ (t_1-t_0)-(t_{L+1}-t_L) \geq 2 $. The case where $ (t_{L+1}-t_L)-(t_1-t_0) \geq 2 $ can be handled similarly by defining $\tilde{\mathbb{T}}=\{\tilde{t}_0=1,\tilde{t}_1=t_1+1,\tilde{t}_2=t_2+1,\ldots,\tilde{t}_{L-1}=t_{L-1}+1,\tilde{t}_L=t_L+1\}$.

    (2) Assume there exists a design $\mathbb T$ such that $|(t_{l+1}-t_l)-(t_{l'+1}-t_{l'})|\geq 2$ for some $1\leq l<l'\leq L+1$. We will construct a new design $\tilde{\mathbb T}$ with a lower value of \Cref{eq:discussion}. Without loss of generality, let us assume that $(t_{l+1}-t_l)-(t_{l'+1}-t_{l'})\geq 2$. The case where $(t_{l'+1}-t_{l'})-(t_{l+1}-t_l)\geq 2$ can be addressed similarly. Consider the design $\tilde{\mathbb T}=\{\tilde t_0=1,\tilde t_1=t_1,\ldots,\tilde t_l=t_l,\tilde t_{l+1}=t_{l+1}-1,\ldots,\tilde t_{l'}=t_{l'}-1,\tilde t_{l'+1}=t_{l'+1},\ldots,\tilde t_L=t_L\}$. Notably, we have $\tilde t_{\ell+1}-\tilde t_{\ell}=t_{\ell+1}-t_{\ell}$ when $\ell \neq l, l'$. The change in the first term in \Cref{eq:discussion} is given by:
    \begin{align*}
        &\sum_{l=0}^{L}(\tilde t_{l+1}-\tilde t_l)^2-\sum_{l=0}^{L}(t_{l+1}-t_l)^2\\
        =&\{(\tilde t_{l'+1}-\tilde t_{l'})^2+(\tilde t_{l+1}-\tilde t_l)^2\}-\{(t_{l'+1}-t_{l'})^2+(t_{l+1}-t_l)^2\}\\
        =&\{(t_{l'+1}-t_{l'}+1)^2+(t_{l+1}-t_l-1)^2\}-\{(t_{L+1}-t_L)^2+(t_1-t_0)^2\}\\
        =&-2\{(t_{l+1}-t_l)-(t_{l'+1}-t_{l'})-1\}< 0.
    \end{align*}
    The second and third terms remain unchanged because $L$ remains unchanged and $\tilde t_L-\tilde t_1=t_L-t_1$. Thus, there exists a design $\tilde{\mathbb T}$ with a lower value of \Cref{eq:discussion} than $\mathbb T$ when $(t_{l+1}-t_l)-(t_{l'+1}-t_{l'})\geq 2$. The case where $(t_{l'+1}-t_{l'})-(t_{l+1}-t_l)\geq 2$ can be handled similarly by defining  $\tilde{\mathbb T}=\{\tilde t_0=1,\tilde t_1=t_1,\ldots,\tilde t_l=t_l,\tilde t_{l+1}=t_{l+1}+1,\ldots,\tilde t_{l'}=t_{l'}+1,\tilde t_{l'+1}=t_{l'+1},\ldots,\tilde t_L=t_L\}$.
\end{proof}

Based on \Cref{lemma.discussion}, a design that minimizes \Cref{eq:discussion} can be characterized by two integers: $a = \min\{t_1 - t_0, t_{L+1} - t_L\}$ and $b = \min_{l=1,\ldots,L-1}\{t_{l+1} - t_l\}$. Another important value in the objective function is $L$.
Although the integer $L$ is unknown, it can be computed using $T$, $a$, $b$ and $\theta^*$ under two cases, as discussed in \Cref{lemma.L}.

\begin{lemma}
    \label{lemma.L}
    Under Assumptions \ref{assumption.nonanticipativity}--\ref{assumption.bounded} and $\R[q_1]=\R[q_2]=0.5$, given a candidate integer pair $(a,b)$ to characterize a design $\mathbb T$, we consider two cases: $t_1 - t_0 = t_{L+1} - t_L$ (Case 1) or $t_1 - t_0 \neq t_{L+1} - t_L$ (Case 2). 
    
    (1.1) For the case where $t_1-t_0=t_{L+1}-t_L$ and $\theta^*\leq b(b+1)/p^2-1$, we have $L = \lfloor(T - 2a) / b\rfloor + 1$; 
    
    (1.2) for the case where $t_1-t_0=t_{L+1}-t_L$ and $\theta^*> b(b+1)/p^2-1$, we have $L=\lceil (T - 2a) / (b+1)\rceil+1$; 
    
    (2.1) for the case where $t_1-t_0\neq t_{L+1}-t_L$ and $\theta^*\leq b(b+1)/p^2-1$, we have $L = \lfloor(T - 2a-1) / b\rfloor + 1$; 
    
    (2.2) for the case where $t_1-t_0\neq t_{L+1}-t_L$ and $\theta^*> b(b+1)/p^2-1$, we have $L=\lceil (T - 2a-1) / (b+1)\rceil+1$.
\end{lemma}

\begin{proof}[Proof of \Cref{lemma.L}]
    To determine the value of $ L $, it is crucial to explicitly assess the number of time periods of lengths $ b $ and $ b + 1 $. Our analysis primarily focuses on comparing the objective function related to $ b $ time periods of length $ b+1 $ with the objective function associated with $ b+1  $ time periods of length $ b $. The difference between the objective functions, as presented in \Cref{eq:discussion}, is given by
    $b(b+1) - (\theta^* + 1)p^2$.
    
    When this difference is greater than 0, it suggests that we should prioritize maximizing the use of time periods of length $ b $. Conversely, when the difference is less than 0, it indicates that we should maximize the utilization of time periods of length $ b + 1 $.
    
    When $b(b+1) - (\theta^* + 1)p^2\geq 0$, i.e., $\theta^*\leq b(b+1)/p^2-1$, it indicates that we should maximize the use of time periods of length $b$. Under this circumstance, we differentiate between two cases: 
    
    (1.1) In Case 1, $L = \lfloor (T - 2a) / b \rfloor + 1$, and the the time period between $t_1$ and $t_L$ will be divided into $\lfloor (T - 2a) / b \rfloor (b + 1) - (T - 2a)$ time periods of length $b$ and $(T - 2a) - \lfloor (T - 2a) / b \rfloor b$ time periods of length $b + 1$. The objective function in \Cref{eq:discussion} is given by
    \begin{align}
    \begin{aligned}
    \label{eq:case1.1}
        &2a^2 + \left\{\left\lfloor \frac{T - 2a}{b} \right\rfloor (b + 1) - (T - 2a)\right\} b^2 \\
        &+ \left\{(T - 2a) - \left \lfloor \frac{T - 2a}{b} \right\rfloor b\right\} (b + 1)^2 \\
        &+ \left\{\left\lfloor \frac{T - 2a}{b} \right\rfloor + \theta^*\left(\left\lfloor \frac{T - 2a}{b} \right\rfloor + 1\right)\right\} p^2 + 2p(T - 2a)\\
        =&2a^2+\theta^* p^2+(T-2a)(2b+2p+1)+\left\lfloor \frac{T-2a}b\right\rfloor \{(\theta^*+1)p^2-b(b+1)\}. 
    \end{aligned}
    \end{align}
    This equation implies that $\lfloor (T-2a)/b\rfloor\in[(T-2a)/(b+1),(T-2a)/b]$, for the coefficients of both $b^2$ and $(b+1)^2$ must be non-negative. 
    
    (2.1) In Case 2, $L = \lfloor (T - 2a - 1) / b \rfloor + 1$, and the time period between $t_1$ and $t_L$ will be divided into $\lfloor (T - 2a - 1) / b \rfloor (b + 1) - (T - 2a - 1)$ time periods of length $b$ and $(T - 2a - 1) - \lfloor (T - 2a - 1) / b \rfloor b$ time periods of length $b + 1$. The objective function \Cref{eq:discussion} is given by
    \begin{align}
    \begin{aligned}
    \label{eq:case2.1}
        &a^2 + (a + 1)^2 + \left\{\left\lfloor \frac{T - 2a - 1}{b} \right\rfloor (b + 1) - (T - 2a - 1)\right\} b^2 \\
        &+ \left\{(T - 2a - 1) - \left \lfloor \frac{T - 2a - 1}{b} \right\rfloor b\right\} (b + 1)^2 \\
        &+ \left\{\left\lfloor \frac{T - 2a - 1}{b} \right\rfloor + \theta^*\left(\left\lfloor \frac{T - 2a - 1}{b} \right\rfloor + 1\right)\right\} p^2 + 2p(T - 2a - 1)\\
        =&a^2+(a+1)^2+\theta^*p^2+(T-2a-1)(2b+2p+1)\\
        &+\left\lfloor \frac{T-2a-1}b\right\rfloor \{(\theta^*+1)p^2-b(b+1)\}.
    \end{aligned}
    \end{align}
    The equation implies that $\lfloor (T-2a-1)/b\rfloor\in[(T-2a-1)/(b+1),(T-2a-1)/b]$, for the coefficients of both $b^2$ and $(b+1)^2$ must be non-negative. 

    When $b(b+1) - (\theta^* + 1)p^2< 0$, i.e., $\theta^*\leq b(b+1)/p^2-1$, it indicates that we should maximize the use of time periods of length $b+1$. Under this circumstance, we differentiate between two cases: 
    
    (1.2) In Case 1, $L = \lceil (T - 2a) / (b + 1) \rceil + 1$, and the the time period between $t_1$ and $t_L$ will be divided into $\lceil (T - 2a) / (b + 1) \rceil (b + 1) - (T - 2a)$ time periods of length $b$ and $(T - 2a) - \lceil (T - 2a) / (b + 1) \rceil b$ time periods of length $b + 1$. The objective function in \Cref{eq:discussion} is given by
    \begin{align}
    \begin{aligned}
    \label{eq:case1.2}
        &2a^2 + \left\{\left\lceil \frac{T - 2a}{b+1} \right\rceil (b + 1) - (T - 2a)\right\} b^2 \\
        &+ \left\{(T - 2a) - \left\lceil \frac{T - 2a}{b+1} \right\rceil  b\right\} (b + 1)^2 \\
        &+ \left\{\left\lceil \frac{T - 2a}{b+1} \right\rceil + \theta^*\left(\left\lceil \frac{T - 2a}{b+1} \right\rceil  + 1\right)\right\} p^2 + 2p(T - 2a)\\
        =&2a^2+(T-2a)(2b+2p+1)+\left\lceil\frac{T-2a}{b+1}\right\rceil \{(\theta^*+1)p^2-b(b+1)\}.
    \end{aligned}
    \end{align}
    This equation implies that $\lceil (T-2a)/(b+1)\rceil\in[(T-2a)/(b+1),(T-2a)/b]$, for the coefficients of both $b^2$ and $(b+1)^2$ must be non-negative. 
    
    (2.2) In Case 2, $L = \lceil (T - 2a) / (b + 1) \rceil + 1$, and the time period between $t_1$ and $t_L$ will be divided into $\lceil (T - 2a - 1) / (b + 1) \rceil (b + 1) - (T - 2a - 1)$ time periods of length $b$ and $(T - 2a - 1) - \lceil (T - 2a - 1) / (b + 1) \rceil b$ time periods of length $b + 1$. The objective function in \Cref{eq:discussion} is given by
    \begin{align}
    \begin{aligned}
    \label{eq:case2.2}
        &a^2 + (a+1)^2 + \left\{\left\lceil \frac{T - 2a - 1}{b+1} \right\rceil (b + 1) - (T - 2a - 1)\right\} b^2 \\
        &+ \left\{(T - 2a - 1) - \left\lceil \frac{T - 2a - 1}{b+1} \right\rceil  b\right\} (b + 1)^2 \\
        &+ \left\{\left\lceil \frac{T - 2a - 1}{b+1} \right\rceil + \theta^*\left(\left\lceil \frac{T - 2a - 1}{b+1} \right\rceil  + 1\right)\right\} p^2 + 2p(T - 2a - 1) \\
        =&a^2+(a+1)^2+\theta^*p^2+(T-2a-1)(2b+2p+1)\\
        &+\left\lceil \frac{T-2a-1}{b+1}\right\rceil \{(\theta^*+1)p^2-b(b+1)\}.
    \end{aligned}
    \end{align}
    This equation implies that $\lceil (T-2a-1)/(b+1)\rceil\in[(T-2a-1)/(b+1),(T-2a-1)/b]$, for the coefficients of both $b^2$ and $(b+1)^2$ must be non-negative. 
\end{proof}

Without additional information, the number of candidate pairs $(a, b)$ is $O(T^2)$, since both $a$ and $b$ are integers less than $T$. Furthermore, we have the constraints $a \geq p + 1$ as specified in \Cref{lemma.t1tl}, and $b \geq p/2$ according to \Cref{lemma.tp}. For each candidate pair $(a, b)$, we need to evaluate the expression in \Cref{eq:discussion} under two scenarios: when $t_1 - t_0 = t_{L+1} - t_L$ and when $t_1 - t_0 \neq t_{L+1} - t_L$. By comparing the results from all candidate pairs across these two cases, we can determine the design that minimizes the value of \Cref{eq:discussion}.

It is important to note that whether $t_{l+1} - t_l = b$ or $t_{l+1} - t_l = b + 1$ does not affect the value of \Cref{eq:discussion}. Instead, the value depends solely on the counts of indices $l$ for which $t_{l+1} - t_l = b$ and $l'$ for which $t_{l'+1} - t_{l'} = b + 1$. Consequently, if $(t_L - t_1)$ is not a multiple of $b$, we can identify a class of designs with different arrangements of $b$ and $b+1$ that yield the same value of \Cref{eq:discussion}. Additionally, in the case where $t_1 - t_0 \neq t_{L+1} - t_L$, we can choose $t_1 - t_0 = a$ and $t_{L+1} - t_L = a + 1$ or vice versa, resulting in two designs that produce the same value for \Cref{eq:discussion}.

From the discussions outlined, we conclude that the time complexity of the proposed polynomial algorithm is $O(T^2)$, for we only need to iterate through all possible integer pairs $(a,b)$ to identify the design $\mathbb T^*$ with $(a^*,b^*)$ that minimizes \Cref{eq:discussion}. 

However, apart from the first term, which is associated with \Cref{eq:discussion}, there is another term in the objective function in \Cref{eq:target}:
\begin{align*}
    \left(\sum_{l=2}^{L}\left[\{(p-t_l+t_{l-1})^+\}^2\right]\right)\gamma_3^*B^2.
\end{align*}
Its value is not less than 0 and reaches its minimum 0 when $b\geq p$. We provide an additional lemma to show the design that minimizes \Cref{eq:discussion} is precisely the minimax optimal design.

\begin{lemma}
\label{lemma.discussion2}
    Under Assumptions \ref{assumption.nonanticipativity}--\ref{assumption.bounded} and $\R[q_1]=\R[q_2]=0.5$, if there exists a design $\mathbb T^*$ that minimizes \Cref{eq:discussion}, it is precisely the minimax optimal design as described in \Cref{theorem.minimaxdesign}.
\end{lemma}

\begin{proof}[Proof of \Cref{lemma.discussion2}]
    The minimax optimal design is just the design that minimizes \Cref{eq:target}. This equation consists of two terms: the first term is linked to \Cref{eq:discussion} and is minimized when \Cref{eq:discussion} is minimized, while the second term is given by:
    \begin{align*}
    \left(\sum_{l=2}^{L}\left[\{(p-t_l+t_{l-1})^+\}^2\right]\right)\gamma_3^*B^2.
    \end{align*}
    Since $\gamma_3^* > 0$, to demonstrate that the design $\mathbb{T}^*$ minimizing \Cref{eq:discussion} is indeed the minimax optimal design, we must show that it also minimizes 
    \begin{align}
    \label{eq:discussion2}
    \left(\sum_{l=2}^{L}\left[\{(p-t_l+t_{l-1})^+\}^2\right]\right).
    \end{align}
    Notably, the value of \Cref{eq:discussion2} reaches its minimum 0 when $b \geq p$. We will prove that the lower bound of \Cref{eq:discussion} when $b \in [p/2, p-1]$ is greater than the upper bound of that when $b = p$. Therefore, we can conclude that the value of $b^*$ in $\mathbb{T}^*$ must be greater than $p$, which implies that it will also minimize \Cref{eq:discussion2}.
    In Case 1, when $b\in[p/2,p-1]$, we have $-b(b + 1) + (\theta^* + 1)p^2 > 0$, which is the Case (1.2) stated in \Cref{lemma.L}, and \Cref{eq:case1.2} is expressed as follows:
    \begin{align*}
        2a^2 + \theta^*p^2 + (T - 2a)\{2p + 2b + 1\} + \left\lceil \frac{T - 2a}{b+1} \right\rceil \{-b(b + 1) + (\theta^* + 1)p^2\}.
    \end{align*}
    Based on the property $\lceil (T-2a)/(b+1) \rceil \in [(T-2a)/(b+1), (T-2a)/b]$ as stated in the proof of \Cref{lemma.L}, the lower bound of \Cref{eq:case1.2} can be established as:
    \begin{align}
    \label{eq:casebound}
        2a^2 + \theta^*p^2 + (T - 2a)\{2p + (b+1) + (\theta^* + 1)p^2 / (b+1)\}.
    \end{align}
    When $b = p$ and $-p(p+1)+(\theta^*+1)p^2> 0$, which is the Case (1.2) stated in \Cref{lemma.L}, the upper bound of \Cref{eq:case1.2} is
    \begin{align*}
        2a^2+\theta^*p^2+(T-2a)\{2p+p+(\theta^*+1)p\},
    \end{align*}
    and when $b=p$ and $-p(p+1)+(\theta^*+1)p^2\leq 0$, which is the Case (1.1) stated in \Cref{lemma.L}, the upper bound of \Cref{eq:case1.2} is
    \begin{align*}
        2a^2+\theta^*p^2+(T-2a)\{2p+p+(\theta^*+1)p\}.
    \end{align*}
    Both values are lower than or equal to \Cref{eq:casebound}, implying that $b^*\geq p$.
    
    The proof for Case 2 follows a similar structure to that of Case 1.
\end{proof}

Based on Lemmas \ref{lemma.discussion}--\ref{lemma.discussion2}, we can give a polynomial--time algorithm with a time complexity of $O(T^2)$ to detect the minimax optimal design $\mathbb T^*$. The algorithm is outlined as follows: 

\begin{algorithm}
\caption{Minimax optimal design detection}
\label{poly.algorithm}
\begin{algorithmic}[1]
    \State \textbf{Input:} $T$, $N$, $q_1$, $q_2$, $p$
    \State Calculate $\theta^*$ as shown in \Cref{eq:theta}.
    
    \For{each candidate pair $(a, b)$ such that $a \geq p + 1$ and $b \geq p/2$}
        \State Calculate the objective function in Case 1:
        \If{$\theta^* \leq b(b+1)/p^2-1$}
            \State Use \Cref{eq:case1.1}
        \Else
            \State Use \Cref{eq:case1.2}
        \EndIf
        
        \State Calculate the objective function in Case 2:
        \If{$\theta^* \leq b(b+1)/p^2-1$}
            \State Use \Cref{eq:case2.1}
        \Else
            \State Use \Cref{eq:case2.2}
        \EndIf
    \EndFor

    \State Identify the pair and case that yields the lowest value of the objective function.
    \State \textbf{Output:} A class of minimax optimal designs characterized by $(a^*, b^*)$ and Case 1 or 2.
\end{algorithmic}
\end{algorithm}


\section{Proof of Corollary~\ref{corollary.minimaxdesign}}
\label{sec:proof-lemma.optimal}

\begin{lemma}[Minimax optimal design with specific $T$ and $\theta^*$]
    \label{lemma.optimal}
    Under Assumptions \ref{assumption.nonanticipativity}--\ref{assumption.bounded} and $\R[q_1]=\R[q_2]=0.5$, we have
    (1) when $p=0$, the minimax optimal design is given by $\{1,2,3,\ldots,T\}$;
    and (2) when $p>0$, we have:

    (2.1) for an integer pair $(a^*,b^*)$ with $b^*\in[\{-1+\sqrt{1+4(\theta^*+1)p^2}\}/2,\{1+\sqrt{1+4(\theta^*+1)p^2}\}/2]$ and $a^*\in[\{2p+b^*+(\theta^*+1)p^2/b^*-1\}/2,\{2p+b^*+(\theta^*+1)p^2/b^*+1\}/2]$, if $T-2a^*$ is a multiple of $b^*$ with $(T-2a^*)/b^*=K-4\geq 0$, then the minimax optimal design is given by $\{1,a^*+1,a^*+b^*+1,a^*+2b^*+1,\ldots,a^*+(K-4)b^*+1\}$;

    (2.2) for an integer pair $(a^*,b^*)$ with $b^*\in[\{-1+\sqrt{1+4(\theta^*+1)p^2}\}/2,\{1+\sqrt{1+4(\theta^*+1)p^2}\}/2]$ and $a^*\in[\{2p+b^*+(\theta^*+1)p^2/b^*\}/2,\{2p+b^*+(\theta^*+1)p^2/b^*+2\}/2]$, if $T-2a^*-1$ is a multiple of $b^*$ with $(T-2a^*-1)/b^*=K-4\geq 0$, then the minimax optimal design is given by
    $\{1,a^*+1,a^*+b^*+1,a^*+2b^*+1,\ldots,a^*+(K-4)b^*+1\}$ or $\{1,a^*+2,a^*+b^*+2,a^*+2b^*+2,\ldots,a^*+(K-4)b^*+2\}$.
\end{lemma}

\begin{proof}[Proof of \Cref{lemma.optimal}]
    By \Cref{lemma.t1tl} and \Cref{lemma.tp}, we need to minimize the objective function under the constraints $t_1\geq p+2$, $t_L\leq T-p$, and $t_{l+1}-t_{l-1}\geq p$ for $l=1,\ldots,L-1$. 
    
    (1) When $p=0$, we notice that the coefficient of $\gamma_1^*$ is $\sum_{l=0}^{L}(t_{l+1}-t_l)^2$, which is minimized when $t_{l+1}-t_l=1$ for $l=0,\ldots,L$. Furthermore, the coefficients of $\gamma_3^*$ is 0. As a result, the minimax optimal design is $\{1,2,3,\ldots,T\}$.

    (2.1.1) When $p\neq 0$ and $\theta^*\leq b(b+1)/p^2-1$, \Cref{eq:case1.1}, the coefficient of $\gamma_1^*$ in Case 1 ($t_1-t_0=t_{L+1}-t_L$), is expressed as follows:
    \begin{align*}
        2a^2 + \theta^*p^2 + (T - 2a)\{2p + 2b + 1\} + \left\lfloor \frac{T - 2a}{b} \right\rfloor \{-b(b + 1) + (\theta^* + 1)p^2\}.
    \end{align*}
    If we want to prove that the objective function reaches the minimum under an integer pair $(a^*,b^*)$, an approach is to compare the objective function under this integer pair with the lower bounds of the objective function under other candidate integer pairs. Based on $\lfloor(T-2a)/b\rfloor\leq (T-2a)/b$, the lower bound of the objective function is
    \begin{align*}
        &2a^2 + \theta^*p^2 + (T - 2a)\{2p + 2b + 1\} + \frac{T - 2a}{b}\{-b(b + 1) + (\theta^* + 1)p^2\}\\
        =&2a^2+\theta^*p^2+(T-2a)\{2p+b+(\theta^*+1)p^2/b\}.
    \end{align*}
    If $T-2a$ is a multiple of $b$, the objective function is also $2a^2+\theta^*p^2+(T-2a)\{2p+b+(\theta^*+1)p^2/b\}$.
    For any integer $a$, an integer $b^*$ that minimizes $2p+b+(\theta^*+1)p^2/b$ must satisfy $2p+b^*+(\theta^*+1)p^2/b^*\leq 2p+(b^*-1)+(\theta^*+1)p^2/(b^*-1)$ and $2p+b^*+(\theta^*+1)p^2/b^*\leq 2p+(b^*+1)+(\theta^*+1)p^2/(b^*+1)$. It implies that the optimal $b^*$ should satisfy $b^*\in[\{-1+\sqrt{1+4(\theta^*+1)p^2}\}/2,\{1+\sqrt{1+4(\theta^*+1)p^2}\}/2]$.
    
    Additionally, consider the remaining term associated with $a$: $2a^2+\theta^* p^2+(T-2a)\{2p+b^*+(\theta^*+1)p^2/b^*\}$. Similarly, an integer $a^*$ that minimizes $2a^2+\theta^* p^2+(T-2a)\{2p+b^*+(\theta^*+1)p^2/b^*\}$ must satisfy $2(a^*)^2+(T-2a^*)\{2p+b^*+(\theta^*+1)p^2/b^*\}\leq 2(a^*+1)^2+(T-2a^*-1)\{2p+b^*+(\theta^*+1)p^2/b^*\}$ and $2(a^*)^2+(T-2a^*)\{2p+b^*+(\theta^*+1)p^2/b^*\}\leq 2(a^*-1)^2+(T-2a^*+2)\{2p+b^*+(\theta^*+1)p^2/b^*\}$. It implies that $a^*\in[\{2p+b^*+(\theta^*+1)p^2/b^*-1\}/2,\{2p+b^*+(\theta^*+1)p^2/b^*+1\}/2]$. 

    As a result, for an integer pair $(a^*,b^*)$ with $b^*\in[\{-1+\sqrt{1+4(\theta^*+1)p^2}\}/2,\{1+\sqrt{1+4(\theta^*+1)p^2}\}/2]$ and $a^*\in[\{2p+b^*+(\theta^*+1)p^2/b^*-1\}/2,\{2p+b^*+(\theta^*+1)p^2/b^*+1\}/2]$, if $T-2a^*$ is a multiple of $b^*$ with $(T-2a^*)/b^*=K-4\geq 0$, the minimax optimal design  is given by $\{1,a^*+1,(a^*+1)+b^*,\ldots,a^*+(K-4)b^*+1\}$.
    

    (2.1.2) When $p\neq 0$ and $\theta^*> b(b+1)/p^2-1$, \Cref{eq:case1.2}, the objective function in Case 1 ($t_1-t_0=t_{L+1}-t_L$), is expressed as follows:
    \begin{align*}
        2a^2+\theta^* p^2+(T-2a)\{2p+2b+1\}+\left\lceil \frac{T - 2a}{b+1} \right\rceil\{-b(b+1)+(\theta^*+1)p^2\}.
    \end{align*}
    We similarly consider its lower bound. 
    Based on $\lceil (T-2a)/(b+1)\rceil \geq (T-2a)/(b+1)$, the lower bound is
    \begin{align*}
        &2a^2+\theta^* p^2+(T-2a)\{2p+2b+1\}+\frac{T - 2a}{b+1}\{-b(b+1)+(\theta^*+1)p^2\}\\
        =&2a^2+\theta^*p^2+(T-2a)\{2p+b+1+(\theta^*+1)p^2/(b+1)\}.
    \end{align*}
    If $T-2a$ is a multiple of $b+1$, the objective function is also $2a^2+\theta^*p^2+(T-2a)\{2p+b+1+(\theta^*+1)p^2/(b+1)\}$. With a similar manner to that in (2.1.1), we can prove that for any $a$, the optimal $(b^*+1)\in\mathbb N$ satisfies $\{-1+\sqrt{1+4(\theta^*+1)p^2}\}/2\leq b^*+1\leq \{1+\sqrt{1+4(\theta^*+1)p^2}\}/2$. Replacing $b^*+1$ with $b^*$, we find that the result is identical to (2.1.1). Consequently, we can derive (2.1) in \Cref{lemma.optimal} by combining (2.1.1) and (2.1.2).
    
    (2.2) 
    The proof for Case 2 ($t_1-t_0\neq t_{L+1}-t_L$) follows a similar structure to that of (2.1.1) and (2.1.2).

    While sequential optimization is typically unfeasible in integer programming, this problem is unique in that the optimal $b$ remains constant for different values of $a$. Suppose there exists a pair $(a^*, b^*)$ that minimizes the objective function. If $b^*$ were to take on any other value $b'$, it would allow for a pair $(a^*, b^*)$ with a larger value of the objective function. The result for $a$ is similar: other values of $a$ would result in a larger objective function value.

\end{proof}

Based on \Cref{lemma.optimal}, we can prove \Cref{corollary.minimaxdesign}.

\begin{proof}[Proof of \Cref{corollary.minimaxdesign}]
    We only need to consider the two typical cases: $\theta^*\leq 1/p$ and $\theta^*\in (1/p,(3p+2)/p^2]$. 

    When $\theta^*\leq 1/p$, we have $a^*=2p$ and $b^*=p$ based on \Cref{lemma.optimal}. Therefore, when $T-4p$ is a multiple of $p$ with $(T-4p)/p=K-4\geq 0$, the minimax optimal design is $\{1,2p+1,3p+1,\ldots,(K-2)p+1\}$.
    
    When $\theta^*\leq (3p+2)/p^2$, we have $a^*=2p+1$ and $b^*=p+1$. Therefore, when $T-4p-2$ is a multiple of $p+1$ with $(T-4p-2)/(p+1)=K-4\geq 0$, the minimax optimal design is $\mathbb T^*_2=\{1,2p+2,3p+3,\ldots,(K-2)(p+1)\}$.
\end{proof}

\section{Proof of Theorem \ref{theorem.variance.sm}}
\label{sec:sm.theorem.variance.sm}

\begin{proof}[Proof of Theorem~\ref{theorem.variance.sm}]
    Under the minimax optimal design $\mathbb T^*$ with $(T-2\aop)/\bop=K-4\geq 0$, define
    \begin{align*}
    \tilde {\mathbf D}_{i,1}(q\mathbf 1,z\mathbf 1)&=\sum_{t=p+1}^{t_2-1}\mathbf D_{i,t}(q\mathbf 1,z\mathbf 1),\  
    \tilde {\mathbf D}_{i,k}(q\mathbf 1,z\mathbf 1)=\sum_{t=t_k}^{t_{k+1}-1}\mathbf D_{i,t}(q\mathbf 1,z\mathbf 1),\quad \text{for}\quad k=2,\ldots,K-2,
    \end{align*}
   Then,
    \begin{align*}
        \taudhat-\taud=\frac 1{T-p}\sum_{k=1}^{K-2}\tilde{\mathbf D}_k(q),\quad  \mathrm{var}\{\taudhat\}=\frac 1{(T-p)^2}\mathrm{var}\left\{\sum_{k=1}^{K-2}\tilde{\mathbf D}_k(q)\right\},\\
        \taushat-\taus=\frac 1{T-p}\sum_{k=1}^{K-2}\tilde{\mathbf S}_k(z),\quad  \mathrm{var}\{\taushat\}=\frac 1{(T-p)^2}\mathrm{var}\left\{\sum_{k=1}^{K-2}\tilde{\mathbf S}_k(z)\right\}.
    \end{align*}
    
    (1) By \Cref{lemma.properties}, we have $E\{\mathbf D_t(\mathbb T^*,q)\}=E\{\mathbf S_t(\mathbb T^*,z)\}=0$. Thus,
    $E\{\tilde {\mathbf D}_k(q)\}=0$ and $E\{\tilde {\mathbf S}_k(z)\}=0$.

    (2) We consider the expectation of $\{\tilde{\mathbf D}_k(q)\}^2$:
    when $k=1$,
    \begin{align*}
        E\{\tilde{\mathbf D}_k(q)\}^2&=E\left\{\sum_{t=p+1}^{t_2-1}\mathbf D_t(\mathbb T^*,q)\right\}^2\\
        &=\sum_{t=p+1}^{t_2-1}E\{\mathbf D_t(\mathbb T^*,q)\}^2+2\sum_{p+1\leq t<t'\leq t_2-1}E\{\mathbf D_t(\mathbb T^*,q)\mathbf D_{t'}(\mathbb T^*,q)\},
    \end{align*}
    and when $k\geq 2$,
    \begin{align*}
        E\{\tilde{\mathbf D}_k(q)\}^2&=E\left\{\sum_{t=t_k}^{t_{k+1}-1}\mathbf D_t(\mathbb T^*,q)\right\}^2\\
        &=\sum_{t=t_k}^{t_{k+1}-1}E\{\mathbf D_t(\mathbb T^*,q)\}^2+2\sum_{t_k\leq t<t'\leq t_{k+1}-1}E\{\mathbf D_t(\mathbb T^*,q)\mathbf D_{t'}(\mathbb T^*,q)\}.
    \end{align*}
    By \Cref{lemma.properties}, we have
    \begin{align*}
        &N^2E\left\{\mathbf D_t(\mathbb T^*,q)\right\}^2\\
        =&(2^{J_t}-1)\left\{\sum_{i=1}^N\yit{q}{}{1}-\sum_{i=1}^N\yit{q}{}{0}\right\}^2\\
        &+2^{J_t+1}\sum_{i=1}^N\left\{\yit{q}{}{1}\yit{q}{}{0}\right\}\\
        &+2^{J_t}(q^{-J_t}-1)\sum_{i=1}^N\left\{\yit{q}{}{1}\right\}^2+2^{J_t}(\bar q^{-J_t}-1)\sum_{i=1}^N\left\{\yit{q}{}{0}\right\}^2,
    \end{align*}
    and when $p+1\leq t<t'\leq t_2-1$ or $t_k\leq t<t'\leq t_{k+1}-1$, we have
    \begin{align*}
        &N^2E\left\{\mathbf D_t(\mathbb T^*,q)\mathbf D_{t'}(\mathbb T^*,q)\right\}\\
        =&(2^{J_{t,t'}^{\circ}}-1)\left\{\sum_{i=1}^N\yit{q}{}{1}-\sum_{i=1}^N\yit{q}{}{0}\right\}\left\{\sum_{i=1}^N\yitp{q}{}{1}-\sum_{i=1}^N\yitp{q}{}{0}\right\}\\
        &+2^{J_{t,t'}^{\circ}}\sum_{i=1}^N\left\{\yit{q}{}{1}\yitp{q}{}{0}\right\}+2^{J_{t,t'}^{\circ}}\sum_{i=1}^N\left\{\yit{q}{}{0}\yitp{q}{}{1}\right\}\\
        &+2^{J_{t,t'}^{\circ}}(q^{-J_{t,t'}^{\circ}}-1)\sum_{i=1}^N\left\{\yit{q}{}{1}\yitp{q}{}{1}\right\}\\
        &+2^{J_{t,t'}^{\circ}}(\bar q^{-J_{t,t'}^{\circ}}-1)\sum_{i=1}^N\left\{\yit{q}{}{0}\yitp{q}{}{0}\right\}.
    \end{align*}
    Considering $J_t=J_{t'}=2$ when $t_k\leq t<t'\leq t_k+p-1$ and $J_t=J_{t'}=J_{t,t'}^{\circ}=1$ otherwise, and  $\eta_1=1$ and $\eta_k=2$ otherwise under the minimax optimal design $\mathbb T^*$, we have
    \begin{align*}
        &N^2E\{\tilde {\mathbf D}_k^2(q)\}\\
        =&\{\yk{q}{}{1}\}^{\top}\{\mathbf J_N+2(q^{-1}-1)\mathbf I_N\}\{\yk{q}{}{1}\}\\
        &+\{\yk{q}{}{0}\}^{\top}\{\mathbf J_N+2(\bar q^{-1}-1)\mathbf I_N\}\{\yk{q}{}{0}\}\\
        &+\{\yk{q}{}{1}\}^{\top}\{-2\mathbf J_N+4\mathbf I_N\}\{\yk{q}{}{0}\}\\
        &+\{\ykc{q}{}{1}\}^{\top}\{(2^{\eta_k}-2)\mathbf J_N+(2^{\eta_k}q^{-\eta_k}-2q^{-1}-2^{\eta_k}+2)\mathbf I_N\}\{\ykc{q}{}{1}\}\\
        &+\{\ykc{q}{}{0}\}^{\top}\{(2^{\eta_k}-2)\mathbf J_N+(2^{\eta_k}\bar q^{-\eta_k}-2\bar q^{-1}-2^{\eta_k}+2)\mathbf I_N\}\{\ykc{q}{}{0}\}\\
        &+\{\ykc{q}{}{1}\}^{\top}\{-2(2^{\eta_k}-2)\mathbf J_N+2(2^{\eta_k}-2)\mathbf I_N\}\{\ykc{q}{}{0}\}.
    \end{align*}
    Then we know
    \begin{align*}
        N^2\sum_{k=1}^{K-1}E\{\tilde {\mathbf D}_k^2(q)\}&=A^d+B^d+C^d.
    \end{align*}
    
    (3) Considering $\{\tilde{\mathbf S}_k(z)\}^2$, we have: when $k=1$,
    \begin{align*}
        E\{\tilde{\mathbf S}_k(z)\}^2&=E\left\{\sum_{t=p+1}^{t_2-1}\mathbf S_t(\mathbb T^*,z)\right\}^2\\
        &=\sum_{t=p+1}^{t_2-1}E\{\mathbf S_t(\mathbb T^*,z)\}^2+2\sum_{p+1\leq t<t'\leq t_2-1}E\{\mathbf S_t(\mathbb T^*,z)\mathbf S_{t'}(\mathbb T^*,z)\},
    \end{align*}
    and when $k\geq 2$,
    \begin{align*}
        E\{\tilde{\mathbf S}_k(z)\}^2&=E\left\{\sum_{t=t_k}^{t_{k+1}-1}\mathbf S_t(\mathbb T^*,z)\right\}^2\\
        &=\sum_{t=t_k}^{t_{k+1}-1}E\{\mathbf S_t(\mathbb T^*,z)\}^2+2\sum_{t_k\leq t<t'\leq t_{k+1}-1}E\{\mathbf S_t(\mathbb T^*,z)\mathbf S_{t'}(\mathbb T^*,z)\}.
    \end{align*}
    By \Cref{lemma.properties}, we have
    \begin{align*}
        &N^2E\left\{\mathbf S_t(\mathbb T^*,z)\right\}^2\\
        =&(2^{J_t}-1)\left\{\sum_{i=1}^N\yit{q_1}{z}{1}\right\}^2+(2^{J_t}-1)\left\{\sum_{i=1}^N\yit{q_2}{z}{1}\right\}^2\\
        &+2\left\{\sum_{i=1}^N\yit{q_1}{z}{1}\right\}\left\{\sum_{i=1}^N\yit{q_2}{z}{1}\right\}\\
        &+2^{J_t}(q_{1,z}^{-J_t}-1)\sum_{i=1}^N\left\{\yit{q_1}{z}{1}\right\}^2+2^{J_t}(q_{2,z}^{-J_t}-1)\sum_{i=1}^N\left\{\yit{q_2}{z}{1}\right\}^2,
    \end{align*}
    and when $p+1\leq t<t'\leq t_2-1$ or $t_k\leq t<t'\leq t_{k+1}-1$, we have
    \begin{align*}
        &N^2E\left\{\mathbf S_t(\mathbb T^*,z)\mathbf S_{t'}(\mathbb T^*,z)\right\}\\
        =&(2^{J_{t,t'}^{\circ}}-1)\left\{\sum_{i=1}^N\yit{q_1}{z}{1}\right\}\left\{\sum_{i=1}^N\yitp{q_1}{z}{1}\right\}\\
        &+(2^{J_{t,t'}^{\circ}}-1)\left\{\sum_{i=1}^N\yit{q_2}{z}{1}\right\}\left\{\sum_{i=1}^N\yitp{q_2}{z}{1}\right\}\\
        &+\left\{\sum_{i=1}^N\yit{q_1}{z}{1}\right\}\left\{\sum_{i=1}^N\yitp{q_2}{z}{1}\right\}+\left\{\sum_{i=1}^N\yitp{q_1}{z}{1}\right\}\left\{\sum_{i=1}^N\yit{q_2}{z}{1}\right\}\\
        &+2^{J_{t,t'}^{\circ}}(q_{1,z}^{-J_{t,t'}^{\circ}}-1)\sum_{i=1}^N\left\{\yit{q_1}{z}{1}\yitp{q_1}{z}{1}\right\}\\
        &+2^{J_{t,t'}^{\circ}}(q_{2,z}^{-J_{t,t'}^{\circ}}-1)\sum_{i=1}^N\left\{\yit{q_2}{z}{1}\yitp{q_2}{z}{1}\right\}.
    \end{align*}
    Considering $J_t=J_{t'}=2$ when $t_k\leq t<t'\leq t_k+p-1$ and $J_t=J_{t'}=J_{t,t'}^{\circ}=1$ otherwise, and  $\eta_1=1$ and $\eta_k=2$ otherwise under the minimax optimal design $\mathbb T^*$, we have
    \begin{align*}
        &N^2E\{\tilde {\mathbf S}_k(z)\}^2\\
        =&\{\yk{q_1}{z}{1}\}^{\top}\{\mathbf J_N+2(q_{1,z}^{-1}-1)\mathbf I_N\}\{\yk{q_1}{z}{1}\}\\
        &+\{\yk{q_2}{z}{1}\}^{\top}\{\mathbf J_N+2(q_{2,z}^{-1}-1)\mathbf I_N\}\{\yk{q_2}{z}{1}\}\\
        &+\{\yk{q_1}{z}{1}\}^{\top}\{2\mathbf J_N\}\{\yk{q_2}{z}{1}\}\\
        &+\{\ykc{q_1}{z}{1}\}^{\top}\{(2^{\eta_k}-2)\mathbf J_N+(2^{\eta_k}q_{1,z}^{-\eta_k}-2q_{1,z}^{-1}-2^{\eta_k}+2)\mathbf I_N\}\{\ykc{q_1}{z}{1}\}\\
        &+\{\ykc{q_2}{z}{1}\}^{\top}\{(2^{\eta_k}-2)\mathbf J_N+(2^{\eta_k}q_{2,z}^{-\eta_k}-2q_{2,z}^{-1}-2^{\eta_k}+2)\mathbf I_N\}\{\ykc{q_2}{z}{1}\}.\\
    \end{align*}
    Considering $\eta_1=1$ and $\eta_k=2$ when $k\geq 2$, we know
    \begin{align*}
        N^2\sum_{k=1}^{K-1}E\{\tilde {\mathbf S}_k^2(z)\}&=A^s+B^s+C^s.
    \end{align*}

    (4) We consider the expectations of $\tilde{\mathbf D}_k(q)\tilde{\mathbf D}_{k+1}(q)$ and $\tilde{\mathbf S}_k(z)\tilde{\mathbf S}_{k+1}(z)$, where $k=1,\ldots,K-3$, as other cross-terms are all 0. 
    Considering that $J_{t,t'}^{\circ}=1$ under the minimax optimal design $\mathbb T^*$ when $p+1\leq t\leq t_{2}-1$, $t_{2}\leq t'\leq t_{3}-1$ or $t_k\leq t\leq t_{k+1}-1$, $t_{k+1}\leq t'\leq t_{k+2}-1$, we have
    \begin{align*}
        N^2E\{\tilde{\mathbf D}_k(q)\tilde {\mathbf D}_{k+1}(q)\}
        &=\{\yk{q}{}{1}\}^{\top}\{\mathbf J_N+2(q^{-1}-1)\mathbf I_N\}\{\ykcplus{q}{}{1}\}\\
        &\quad +\{\yk{q}{}{0}\}^{\top}\{\mathbf J_N+2(\bar q^{-1}-1)\mathbf I_N\}\{\ykcplus{q}{}{0}\}\\
        &\quad +\{\yk{q}{}{1}\}^{\top}\{-\mathbf J_N+2\mathbf I_N\}\{\ykcplus{q}{}{0}\}\\
        &\quad +\{\yk{q}{}{0}\}^{\top}\{-\mathbf J_N+2\mathbf I_N\}\{\ykcplus{q}{}{1}\},
    \end{align*}
    and
    \begin{align*}
        N^2E\{\tilde{\mathbf S}_k(z)\tilde {\mathbf S}_{k+1}(z)\}
        &=\{\yk{q_1}{z}{1}\}^{\top}\{\mathbf J_N+2(q_{1,z}^{-1}-1)\mathbf I_N\}\{\ykcplus{q_1}{z}{1}\}\\
        &\quad +\{\yk{q_2}{z}{1}\}^{\top}\{\mathbf J_N+2(q_{2,z}^{-1}-1)\mathbf I_N\}\{\ykcplus{q_2}{z}{1}\}\\
        &\quad +\{\yk{q_1}{z}{1}\}^{\top}\{\mathbf J_N\}\{\ykcplus{q_2}{z}{1}\}\\
        &\quad +\{\yk{q_2}{z}{1}\}^{\top}\{\mathbf J_N\}\{\ykcplus{q_1}{z}{1}\}.
    \end{align*}
    Then we know
    \begin{align*}
        2N^2\sum_{k=1}^{K-3}E\{\tilde{\mathbf D}_k(q)\tilde {\mathbf D}_{k+1}(q)\}&=D^d+E^d+F^d+G^d,\\
        2N^2\sum_{k=1}^{K-3}E\{\tilde{\mathbf S}_k(z)\tilde {\mathbf S}_{k+1}(z)\}&=D^s+E^s+F^s+G^s.
    \end{align*}

    The variance can be calculated by
    \begin{align*}
        N^2(T-p)^2\mathrm{var}\{\taudhat\}&=N^2\mathrm{var}\left\{\sum_{k=1}^{K-1}\tilde{\mathbf D}_k(q)\right\}\\
        &=N^2\sum_{k=1}^{K-2}E\{\tilde{\mathbf D}_k(q)\}^2+2N^2\sum_{k=1}^{K-3}E\{\tilde {\mathbf D}_k(q)\tilde{\mathbf D}_{k+1}(q)\},\\
        N^2(T-p)^2\mathrm{var}\{\taushat\}&=N^2\mathrm{var}\left\{\sum_{k=1}^{K-2}\tilde{\mathbf S}_k(z)\right\}\\
        &=N^2\sum_{k=1}^{K-2}E\{\tilde{\mathbf S}_k(z)\}^2+2N^2\sum_{k=1}^{K-3}E\{\tilde {\mathbf S}_k(z)\tilde{\mathbf S}_{k+1}(z)\}.
    \end{align*}
\end{proof}


\section{Proof of \Cref{corollary.upperbound}}
\label{sec:sm.corollary.upperbound}

\begin{proof}[Proof of \Cref{corollary.upperbound}]
    To estimate $E\{\tilde {\mathbf D}_k(q)\}^2$, we notice that all terms can be unbiasedly estimated except the terms with $\yik{q}{}{1}\yik{q}{}{0}$ (in $C^d$), $\yikc{q}{}{1}\yikc{q}{}{0}$ (in $C^d$), $\yik{q}{}{1}\yikcplus{q}{}{0}$ (in $F^d$) and $\yik{q}{}{0}\yikcplus{q}{}{1}$ (in $G^d$). To solve this issue, we conservatively estimate them with $2^{-1}[\{\yik{q}{}{1}\}^2+\{\yik{q}{}{0}\}^2]$, $2^{-1}[\{\yikc{q}{}{1}\}^2+\{\yikc{q}{}{0}\}^2]$, $2^{-1}[\{\yik{q}{}{1}\}^2+\{\yikcplus{q}{}{0}\}^2]$, and $2^{-1}[\{\yik{q}{}{0}\}^2+\{\yikcplus{q}{}{1}\}^2]$, respectively, according to Cauchy--Schwarz inequality. Then an upper bound for $C^d$ is
    \begin{align*}
        \bar C^d&= \sum_{k=1}^{K-2}\{\yk{q}{}{1}\}^{\top}\{-2(\mathbf J_N-\mathbf I_N)\}\{\yk{q}{}{0}\}\\
        &\quad +\sum_{k=2}^{K-2}\{\ykc{q}{}{1}\}^{\top}\{-4(\mathbf J_N-\mathbf I_N)\}\{\ykc{q}{}{0}\}\\
        &\quad +\sum_{k=1}^{K-2}\{\yk{q}{}{1}\}^{\top}\{\mathbf I_N\}\{\yk{q}{}{1}\}+\sum_{k=1}^{K-2}\{\yk{q}{}{0}\}^{\top}\{\mathbf I_N\}\{\yk{q}{}{0}\},
    \end{align*}
    where $\mathbf J_N$ denotes the $N\times N$ matrix of ones and $\mathbf I_N$ denotes the $N\times N$ identity matrix.
    Similarly, upper bounds for $F^d$ and $G^d$ are
    \begin{align*}
     \bar F^d&= 2\sum_{k=1}^{K-3}\{\yk{q}{}{1}\}^{\top}\{-(\mathbf J_N-\mathbf I_N)\}\{\ykcplus{q}{}{0}\}\\
    &\quad +\sum_{k=1}^{K-3}\{\yk{q}{}{1}\}^{\top}\{\mathbf I_N\}\{\yk{q}{}{1}\}+\sum_{k=2}^{K-2}\{\ykc{q}{}{0}\}^{\top}\{\mathbf I_N\}\{\ykc{q}{}{0}\},\\
    \bar G^d&= 2\sum_{k=1}^{K-3}\{\yk{q}{}{0}\}^{\top}\{-(\mathbf J_N-\mathbf I_N)\}\{\ykcplus{q}{}{1}\}\\
    &\quad +\sum_{k=2}^{K-2}\{\ykc{q}{}{1}\}^{\top}\{\mathbf I_N\}\{\ykc{q}{}{1}\}+\sum_{k=1}^{K-3}\{\yk{q}{}{0}\}^{\top}\{\mathbf I_N\}\{\yk{q}{}{0}\}.
    \end{align*}

    Similar to direct effects, we can conservatively estimate the cross terms by Cauchy--Schwarz inequality. An upper bound for $C^s$ is
    \begin{align*}
        \bar C^s&= \sum_{k=1}^{K-2}\{\yk{q_1}{z}{1}\}^{\top}\{\mathbf J_N\}\{\yk{q_1}{z}{1}\}+\sum_{k=1}^{K-2}\{\yk{q_2}{z}{1}\}^{\top}\{\mathbf J_N\}\{\yk{q_2}{z}{1}\}.
    \end{align*}
    Upper bounds for $F^s$ and $G^s$ are
    \begin{align*}
        \bar F^s&=\sum_{k=1}^{K-3}\{\yk{q_1}{z}{1}\}^{\top}\{\mathbf J_N\}\{\yk{q_1}{z}{1}\}+\sum_{k=2}^{K-2}\{\ykc{q_2}{z}{1}\}^{\top}\{\mathbf J_N\}\{\ykc{q_2}{z}{1}\},\\
        \bar G^s&=\sum_{k=2}^{K-2}\{\ykc{q_1}{z}{1}\}^{\top}\{\mathbf J_N\}\{\ykc{q_1}{z}{1}\}+\sum_{k=1}^{K-3}\{\yk{q_2}{z}{1}\}^{\top}\{\mathbf J_N\}\{\yk{q_2}{z}{1}\}.
    \end{align*}

\end{proof}

\section{Proof of \Cref{corollary.estimator}}
\label{sec:sm.corollary.estimator}

To prove that the estimators are conservative, we first introduce  a lemma established by \citet{delevoye2020consistency}. Suppose that the quantity of interest is $\mu=N^{-1}\sum_{i=1}^Ny_i$ and its Horvitz--Thompson estimator is $\hat \mu=N^{-1}\sum_{i=1}^N\nu_i\omega_iy_i$, where $\nu_i$ is the indicator function taking the value 1 when the unit $i$ is in the sample and 0 otherwise, and $\omega_i=1/\pi_i$ with $\pi_i=E(\nu_i)$ being the inclusion probability that the unit $i$ is in the sample. The Horvitz--Thompson estimator is unbiased because $E(\nu_i \omega_i)=1$. Let $\bar \pi=N^{-1}\sum_{i=1}^N\pi_i$ denote the average inclusion probability, $\tilde \pi_i=\pi_i/\bar \pi$ denote the normalized version of $\pi_i$, and $\tilde \omega_i=\bar \pi \omega_i=1/\tilde \pi_i$ denote the normalized version of $\omega_i$.

\begin{assumption}[Condition 1 of \citet{delevoye2020consistency}]
\label{assumption.sm}
    There exist $\rho_y>2$ and $\rho_{\omega}>1$ with $\rho_y\rho_{\omega}\geq \rho_y+2\rho_{\omega}$ such that
    \begin{align*}
        \left\{\frac 1N\sum_{i=1}^N|y_i|^{\rho_y}\right\}^{1/\rho_y}\leq \kappa_y,\quad \left\{\frac 1N\sum_{i=1}^N\tilde \omega_i^{\rho_{\omega}}\right\}^{1/\rho_{\omega}}\leq \kappa_{\omega}.
    \end{align*}
\end{assumption}

\begin{lemma}[Lemma 1 of \citet{delevoye2020consistency}]
    \label{lemma.HT}
    Let $\pi_{i\mid j}$ be the inclusion probability of unit $i$ conditional on that unit $j$ is sampled and $\tilde \pi_{i\mid j}=\pi_{i\mid j}/\bar \pi$ be the normalized version of $\pi_{i\mid j}$. Define 
    \begin{align*}
        H(h)=\left\{\frac 1{N^2}\sum_{i=1}^N\sum_{j\neq i}|\tilde \pi_{i\mid j}-\tilde \pi_i|^{1/h}\right\}^h.
    \end{align*}
    Under \Cref{assumption.sm}, $\mathrm{var}(\hat \mu)\leq \kappa_y^2\kappa_{\omega}/\bar \pi N+\kappa_y^2\kappa_{\omega} H(1-1/\rho_y-1/\rho_{\omega})$.
\end{lemma}

\begin{proof}[Proof of \Cref{corollary.estimator}]
    The conservative variance estimators $\hat A^*,\hat B^*,\hat{\bar C}^*,\hat D^*,\hat E^*,\hat {\bar F}^*,\hat{\bar G}^*$ ($*=d,s$) presented in \Cref{corollary.estimator} are, in fact, the Horvitz--Thompson estimators of the upper bounds provided in \Cref{corollary.upperbound}. Since all of the Horvitz--Thompson estimators are unbiased, it follows that the variance estimators in \Cref{corollary.estimator} are unbiased for the upper bounds of the variances. To prove the estimators are conservative, we are going to additionally prove that the estimators are consistent. Recall that
    \begin{align*}
        \hat A^d&= \sum_{k=1}^{K-2}\{\ykc{q}{}{1}\circ \Ikc{q}{}{1}\}^{\top}\left\{\frac{3(\mathbf J_N-\mathbf I_N)}{2^{-\eta_k}q^{2\eta_k}}+\frac{(4q^{-2}-1)\mathbf I_N}{2^{-\eta_k}q^{\eta_k}}\right\}\{\ykc{q}{}{1}\circ \Ikc{q}{}{1}\}\\
        & + \sum_{k=1}^{K-2}\{\ykd{q}{}{1}\circ \Ik{q}{}{1}\}^{\top}\left\{\frac{2(\mathbf J_N-\mathbf I_N)}{2^{-\eta_k}q^{\eta_k+1}}+\frac{(4q^{-1}-2)\mathbf I_N}{2^{-\eta_k}q^{\eta_k}}\right\}\{\ykc{q}{}{1}\circ \Ikc{q}{}{1}\} \\
        & + \sum_{k=1}^{K-2}\{\ykd{q}{}{1}\circ \Ik{q}{}{1}\}^{\top}\left\{\frac{\mathbf J_N-\mathbf I_N}{2^{-1}q^{2}}+\frac{(2q^{-1}-1)\mathbf I_N}{2^{-1}q}\right\}\{\ykd{q}{}{1}\circ \Ik{q}{}{1}\},
    \end{align*}
    where $\circ$ denotes the Hadamard product. Both $\hat A^d$ and its expectation $A^d$ scale as $O\{N^2(T-p)\}$; the terms with $\mathbf J_N-\mathbf I_N$ scale as $O\{N^2(T-p)\}$ and the terms with $\mathbf I_N$ scale as $O\{N(T-p)\}$. To prove that the variance estimators are conservative, we will show that $\{N^2(T-p)\}^{-1}\hat A^d$ converges in probability to $\{N^2(T-p)\}^{-1}A^d$. We consider 
    \begin{align*}
        \sum_{k=1}^{K-2}\{\ykc{q}{}{1}\circ \Ikc{q}{}{1}\}^{\top}\left\{\frac{3(\mathbf J_N-\mathbf I_N)}{2^{-\eta_k}q^{2\eta_k}}\right\}\{\ykc{q}{}{1}\circ \Ikc{q}{}{1}\}
    \end{align*}
    in $\hat A^d$. It can be further rewritten as 
    \begin{align*}
        3\sum_{k=1}^{K-2}\sum_{i\neq j}\frac{\Iikc{q}{}{1}\Ijkc{q}{}{1}}{2^{-\eta_k}q^{2\eta_k}}\yikc{q}{}{1}\yjkc{q}{}{1},
    \end{align*}
    where $\Iikc{q}{}{1}$ is defined in \Cref{corollary.estimator} and $\yik{q}{}{1}$ is defined in the first paragraph in \Cref{sec:sm.variance}.
    We further notice that the following term is a Horvitz--Thompson estimator:
    \begin{align*}
        \frac 3{(K-2)N(N-1)}\sum_{k=1}^{K-2}\sum_{i\neq j}\frac{\Iikc{q}{}{1}\Ijkc{q}{}{1}}{2^{-\eta_k}q^{2\eta_k}}\yikc{q}{}{1}\yjkc{q}{}{1}.
    \end{align*}
    In our setting, a sample can be expressed with $(ijk)$ with $i\neq j$. For example, the outcome is $y_{ijk}=3\yikc{q}{}{1}\yjkc{q}{}{1}$ and the indicator function is $\nu_{ijk}=\Iikc{q}{}{1}\Ijkc{q}{}{1}$. The total population size is then $(K-2)N(N-1)$. The inclusion probability is $\pi_{ijk}=2^{-\eta_k}q^{2\eta_k}=O(1)$, the average inclusion probability is $\bar \pi=\{(K-2)N(N-1)\}^{-1}\sum_{k=1}^{K-2}\sum_{i\neq j}\pi_{ijk}=O(1)$ and the normalized version of $\pi_{ijk}$ is $\tilde \pi_{ijk}=\pi_{ijk}/\bar \pi=O(1)$. Letting both $\rho_y$ and $\rho_{\omega}$ in \Cref{lemma.HT} tend to infinity, we have $\kappa_y=\max\{|y_{ijk}|\}\leq 3p^2B^2$ and  $\kappa_{\omega}=\max\{1/\tilde \pi_{ijk}\}=O(1)$. As a result, we can let both $\rho_y$ and $\rho_{\omega}$ tend to infinity. When $|k-k'|\geq 3$, the conditional inclusion probability $\pi_{ijk\mid i'j'k'}=\pi_{ijk}$ because $\Iikc{q}{}{1}$ and $\check I_{j,k'}(q\mathbf 1,\mathbf 1)$ are independent when $|k-k'|\geq 3$. When $|k-k'|\leq 2$, we have $\pi_{ijk\mid i'j'k'}\asymp \pi_{ijk}$, where $a\asymp b$ means that $a/b=O(1)$ and $b/a=O(1)$, resulting in $|\pi_{i'j'k'\mid ijk}-\pi_{ijk}|\asymp \pi_{ijk}$. Then, we have
    \begin{align*}
        H(1)&=\frac 1{(K-2)^2N^2(N-1)^2}\sum_{(ijk)}\sum_{(i'j'k')\neq (ijk)}|\tilde \pi_{ijk\mid i'j'k'}-\tilde \pi_{ijk}|\\
        &\asymp \frac 1{(K-2)^2N^2(N-1)^2}\sum_{k=2}^{K-1}\sum_{i\neq j}N^2|\tilde \pi_{ijk}|=O\left(\frac 1{K-2}\right),
    \end{align*}
    where $H(\cdot)$ is defined in \Cref{lemma.HT}.
    By \Cref{lemma.HT}, the variance of 
    \begin{align*}
        \frac 3{(K-2)N(N-1)}\sum_{k=1}^{K-1}\sum_{i\neq j}\frac{\Iikc{q}{}{1}\Ijkc{q}{}{1}}{2^{-\eta_k}q^{2\eta_k}}\yikc{q}{}{1}\yjkc{q}{}{1}
    \end{align*}
    is no greater than 
    \begin{align*}
        \frac{\kappa_y^2\kappa_{\omega}}{\bar \pi (K-2)N(N-1)}+\kappa_y^2\kappa_{\omega}H(1)=O\left(\frac 1{K-2}\right)=O(T^{-1}).
    \end{align*}
    It tends to 0 as $T$ approaches infinity. The $L^2$ convergence implies that this term converges in probability.
    Consistency for $\hat A^d$ can be established by considering other terms in $\hat A^d$, which are all Horvitz--Thompson estimators by similar arguments.
    Therefore, 
    \begin{align*}
        \frac{1}{(T-p)N^2}(\hat A^d-A^d)=o_P(1).
    \end{align*}
    The proofs for other quantities are similar. Consequently, the variance estimators are consistent for the upper bounds of the variances, leading us to conclude that they are conservative.
    
\end{proof}




\section{Proof of Theorem~\ref{theorem.CLT}}
\label{sec:sm.theorem.CLT}

\begin{definition}[$\phi$-dependent random variables, \citet{hoeffding1948}]
    For any sequence $\{X_1,X_2,\ldots\}$, if there exists $\phi$ such that for any $t$ and $s-r>\phi$, the two sets
    \begin{align*}
        (X_t,X_{t+1},\ldots,X_{t+r}),\quad (X_{t+s},X_{t+s+1},\ldots,X_{t+s+r})
    \end{align*}
    are independent, then the sequence is said to be $\phi$-dependent.
\end{definition}

\begin{lemma}[Theorem 2.1 of \citet{romano2000}; Lemma 5 of \citet{han2024}]
    \label{lemma.CLT}
    Let $\{X_{n,i}\}$ be a triangular array of zero-mean random variables and $\phi\in\mathbb N$. For each $n=1,2,\ldots$, let $r=r_n$, and suppose that $X_{n,1},X_{n,2},\ldots,X_{n,r}$ is a $\phi$-dependent sequence of random variables. Define
    \begin{align*}
        C_{n,e_1,e_0}=\mathrm{var}\left(\sum_{i=e_0}^{e_0+e_1-1}X_{n,i}\right),\quad C_n=C_{n,r,1}=\mathrm{var}\left(\sum_{i=1}^rX_{n,i}\right).
    \end{align*}
    For some $\delta>0$ and $-1\leq \gamma\leq 1$ and $\bar \phi>2\phi$, if the following conditions hold,

    1. $E|X_{n,i}|^{2+\delta}\leq \Delta_n$, for all $i$,

    2. $C_{n,e_1,e_0}/e_1^{1+\gamma}\leq K_n$, for all $e_0$ and $e_1\geq \phi$,

    3. $C_n/(r\phi^{\gamma})\geq L_n$,

    4. $\phi/\bar \phi\rightarrow 0$, $(K_n/L_n)\cdot(\phi/\bar \phi)\rightarrow 0$, and $(K_n/L_n)\cdot(\phi/\bar \phi)^{(1-\gamma)/2}\rightarrow 0$,

    5. $\Delta_n L_n^{-(2+\delta) / 2} \bar \phi^{\delta / 2+(1-\gamma)(2+\delta) / 2} r^{-\delta / 2}(\phi/\bar \phi)^{(1-\gamma)(2+\delta) / 2} \rightarrow 0$,
    
\noindent then 
    \begin{align*}
        \frac{\sum_{i=1}^rX_{n,i}}{C_n}\stackrel{d}{\rightarrow}\mathcal N(0,1).
    \end{align*}
\end{lemma}

\begin{proof}[Proof of Theorem~\ref{theorem.CLT}]
    In the proof, let $n=NT$, $r_n=N(T-p)$, $X_{i,t}^d(q)=\{N(T-p)\}^{-1/2}\mathbf D_{i,t}(\mathbb T^*,q)$ and $X_{i,t}^s(z)=\{N(T-p)\}^{-1/2}\mathbf S_{i,t}(\mathbb T^*,z)$, then the following two sequences are both $2\bop N$-dependent, for $E\{\mathbf D_{i,t}(\mathbb T^*,q)\mathbf D_{i,t+2\bop}(\mathbb T^*,q)\}=0$ and $E\{\mathbf S_{i,t}(\mathbb T^*,z)\mathbf S_{i,t+2\bop}(\mathbb T^*,z)\}=0$ (we omit $(q)$ and $(z)$ for notation simplicity):
    \begin{align*}
        \{X_{1,p+1}^d,X_{2,p+1}^d,\ldots,X_{N,p+1}^d,X_{1,p+2}^d,X_{2,p+2}^d,\ldots,X_{N,p+2}^d,\ldots,X_{N,T}^d\},\\
        \{X_{1,p+1}^s,X_{2,p+1}^s,\ldots,X_{N,p+1}^s,X_{1,p+2}^s,X_{2,p+2}^s,\ldots,X_{N,p+2}^s,\ldots,X_{N,T}^s\}.
    \end{align*}
    We then set $\phi=2\bop N$ in this framework. If we consider multi-center experiments, we can consider the sequence of each center first, and then connect them together. Then the whole sequence is $2\bop N_{\max}$-dependent, where $N_{\max}$ is the maximum number of units in all centers. Without loss of generality, we can let $\phi=2\bop N_{\max}$, as in single-center experiments, $N_{\max}=N$.
    Let $C_{e_1,e_0}^d(q)=\mathrm{var}\{\sum_{(i,t)=e_0}^{(i,t)=e_0+e_1-1}X_{i,t}^d(q)\}$, where $(i,t)=e$ means that $X_{i,t}^d$ is the $e$-th element in the sequence. Similarly, we can define $C_{e_1,e_0}^s(z)=\mathrm{var}\{\sum_{(i,t)=e_0}^{(i,t)=e_0+e_1-1}X_{i,t}^s(z)\}$. Then we have
    \begin{align*}
        C_{e_1,e_0}^d(q)&=\frac 1{N(T-p)}\mathrm{var}\left\{\sum_{(i,t)=e_0}^{(i,t)=e_0+e_1-1}\mathbf D_{i,t}(\mathbb T,q)\right\}\leq \frac 1{N(T-p)}(e_1M_1^d+2e_1\phi M_2^d)=O\left(\frac{e_1\phi}{NT}\right),\\
        C_{e_1,e_0}^s(z)&=\frac 1{N(T-p)}\mathrm{var}\left\{\sum_{(i,t)=e_0}^{(i,t)=e_0+e_1-1}\mathbf S_{i,t}(\mathbb T,z)\right\}\leq \frac 1{N(T-p)}(e_1M_1^s+2e_1\phi M_2^s)=O\left(\frac{e_1\phi}{NT}\right),
    \end{align*}
    where $M_1^d$, $M_2^d$, $M_1^s$, and $M_2^s$ are constants associated with the upper bound of variances and covariances.
    We further have $C^d(q)=C_{N(T-p),1}^d(q)=N(T-p)\mathrm{var}\{\taudhat\}=O(\phi)$ and $C^s(z)=C_{N(T-p),1}^s(z)=N(T-p)\mathrm{var}\{\taushat\}=O(\phi)$. Next, we will check the five conditions in Lemma~\ref{lemma.CLT} with $\gamma=0$. 

    (1) There exists $\Delta$ such that $E|X_{i,t}^d(q)|^{2+\delta}\leq \Delta$ and $E|X_{i,t}^s(z)|^{2+\delta}\leq \Delta$ for all $i$ and $t$, because all the potential outcomes are bounded. In the inequality, $\Delta=O\{(NT)^{-1-\delta/2}\}$.

    (2) There exist $K^d(q)$ and $K^s(z)$ such that $C_{e_1,e_0}^d(q)/e_1\leq K^d(q)$ and $C_{e_1,e_0}^s(z)/e_1\leq K^s(z)$, for all $e_0$ and $e_1\geq \phi$. In the inequality, $K^d(q)=O\{\phi/(NT)\}$ and $K^s(z)=O\{\phi/(NT)\}$.

    (3) There exists $L^d(q)$ and $L^s(z)$ such that $C^d(q)/\{N(T-p)\}\geq L^d(q)$ and $C^s(z)/\{N(T-p)\}\geq L^s(z)$. In the inequality, $L^d(q)=O\{\phi/(NT)\}$ and $L^s(z)=O\{\phi/(NT)\}$.

    (4) We set $\bar \phi=8\bop N^{\alpha}T^{\beta}N_{\max}>\phi$. Then the conditions are satisfied when $\phi/\bar \phi\rightarrow 0$, which is implied by $N^{-\alpha}T^{-\beta}\rightarrow 0$. 

    (5) This condition is satisfied when $N^{\alpha-1}T^{\beta-1}N_{\max}\rightarrow 0$.

    The conditions required in (4) and (5) are satisfied under Assumption~\ref{assumption.NT}. Therefore, we conclude that the asymptotic normality holds.
    
\end{proof}


\section{Proof of Theorem~\ref{theorem.misspecified}}
\label{sec:sm.theorem.misspecified}

In this section, we set $m$ as the order of the carryover effects, and $p$ is the experimenter's knowledge of $m$. When $p\geq m$, the estimation and inference methods will still hold, for $\yitmis{q}{z}{1}{p}=\yitmis{q}{z}{1}{m}$. However, when $p<m$, the estimator will no longer be unbiased and the effects are not well-defined. Let $\mathbf Q^{\mathrm{obs}}_{t_1:t_2}$ denote the observed treated probability and $\mathbf Z^{\mathrm{obs}}_{i,t_1:t_2}$ denote the observed treatment path for unit $i$ during the time period $[t_1,t_2]$. The $m$-misspecified lag-$p$ direct effect and spillover effect can be defined as 
\begin{align*}
    &\taudorder[m]=\frac 1N\sum_{i=1}^N\frac 1{T-p}\\
    &E\bigg[\sum_{t=p+1}^m\left\{Y_{i,t}(\mathbf Q_{1:t-p-1}^{\mathrm{obs}},q\mathbf 1_{p+1},\mathbf Z_{i,1:t-p-1}^{\mathrm{obs}},\mathbf 1_{p+1})-Y_{i,t}(\mathbf Q_{1:t-p-1}^{\mathrm{obs}},q\mathbf 1_{p+1},\mathbf Z_{i,1:t-p-1}^{\mathrm{obs}},\mathbf 0_{p+1}) \right\}\\
    &+\sum_{t=m+1}^T\Big\{Y_{i,t}(\mathbf Q_{t-m:t-p-1}^{\mathrm{obs}},q\mathbf 1_{p+1},\mathbf Z_{i,t-m:t-p-1}^{\mathrm{obs}},\mathbf 1_{p+1})\\
    &\qquad \qquad \quad -Y_{i,t}(\mathbf Q_{t-m:t-p-1}^{\mathrm{obs}},q\mathbf 1_{p+1},\mathbf Z_{i,t-m:t-p-1}^{\mathrm{obs}},\mathbf 0_{p+1}) \Big\}\bigg],\\
    &\tausorder[m]=\frac 1N\sum_{i=1}^N\frac 1{T-p}\\
    &E\bigg[\sum_{t=p+1}^m\left\{Y_{i,t}(\mathbf Q_{1:t-p-1}^{\mathrm{obs}},q_1\mathbf 1_{p+1},\mathbf Z_{i,1:t-p-1}^{\mathrm{obs}},z\mathbf 1_{p+1})-Y_{i,t}(\mathbf Q_{1:t-p-1}^{\mathrm{obs}},q_2\mathbf 1_{p+1},\mathbf Z_{i,1:t-p-1}^{\mathrm{obs}},z\mathbf 1_{p+1}) \right\}\\
    &+\sum_{t=m+1}^T\Big\{Y_{i,t}(\mathbf Q_{t-m:t-p-1}^{\mathrm{obs}},q_1\mathbf 1_{p+1},\mathbf Z_{i,t-m:t-p-1}^{\mathrm{obs}},z\mathbf 1_{p+1})\\
    &\qquad \qquad \quad -Y_{i,t}(\mathbf Q_{t-m:t-p-1}^{\mathrm{obs}},q_2\mathbf 1_{p+1},\mathbf Z_{i,t-m:t-p-1}^{\mathrm{obs}},z\mathbf 1_{p+1}) \Big\}\bigg].
\end{align*}

\begin{lemma}
    \label{lemma.misspecified}
    Under Assumptions~\ref{assumption.nonanticipativity}--\ref{assumption.bounded}, $\R[q_1]=\R[q_2]=0.5$, and the minimax optimal design $\mathbb T^*$, for $q=q_1,q_2$ and $z=0,1$, we have (1) for $p<m$, $t\geq m+1$ and $\mathcal F_{\mathbb T^*}(t-p)\leq t-m$, 
    \begin{align*}
        E\left\{Y_{i,t}\frac{\Iitmis{q}{}{1}{p}}{\Pitmis{q}{}{1}{p}}-Y_{i,t}\frac{\Iitmis{q}{}{0}{p}}{\Iitmis{q}{}{0}{p}}\right\}&=\yitmis{q}{}{1}{m}-\yitmis{q}{}{0}{m},\\
        E\left\{Y_{i,t}\frac{\Iitmis{q_1}{z}{1}{p}}{\Pitmis{q_1}{z}{1}{p}}-Y_{i,t}\frac{\Iitmis{q_2}{z}{1}{p}}{\Pitmis{q_2}{z}{1}{p}}\right\}&=\yitmis{q_1}{z}{1}{m}-\yitmis{q_2}{z}{1}{m};
    \end{align*}
    
    (2) for $p<m$, $t\geq m+1$, and $\mathcal F_{\mathbb T^*}(t-p)> t-m$,
    \begin{align*}
        E&\bigg[\left\{Y_{i,t}\frac{\Iitmis{q}{}{1}{p}}{\Pitmis{q}{}{1}{p}}-Y_{i,t}\frac{\Iitmis{q}{}{0}{p}}{\Iitmis{q}{}{0}{p}}\right\} \mid \mathbf Q^{\mathrm{obs}}_{t-m:\mathcal F_{\mathbb T^*}(t-p)-1},
        \mathbf Z^{\mathrm{obs}}_{i,t-m:\mathcal F_{\mathbb T^*}(t-p)-1}\bigg]\\
        =&Y_{i,t}(\mathbf Q^{\mathrm{obs}}_{t-m:\mathcal F_{\mathbb T^*}(t-p)-1},q\mathbf 1_{t-\mathcal F_{\mathbb T(t-p)+1}},\mathbf Z^{\mathrm{obs}}_{i,t-m:\mathcal F_{\mathbb T^*}(t-p)-1},\mathbf 1_{t-\mathcal F_{\mathbb T(t-p)+1}})\\
        & -Y_{i,t}(\mathbf Q^{\mathrm{obs}}_{t-m:\mathcal F_{\mathbb T^*}(t-p)-1},q\mathbf 1_{t-\mathcal F_{\mathbb T(t-p)+1}},\mathbf Z^{\mathrm{obs}}_{i,t-m:\mathcal F_{\mathbb T^*}(t-p)-1},\mathbf 0_{t-\mathcal F_{\mathbb T(t-p)+1}}),\\
        E&\bigg[\left\{Y_{i,t}\frac{\Iitmis{q_1}{z}{1}{p}}{\Pitmis{q_1}{z}{1}{p}}-Y_{i,t}\frac{\Iitmis{q_2}{z}{1}{p}}{\Pitmis{q_2}{z}{1}{p}}\right\} \mid \mathbf Q^{\mathrm{obs}}_{t-m:\mathcal F_{\mathbb T^*}(t-p)-1},\mathbf Z^{\mathrm{obs}}_{i,t-m:\mathcal F_{\mathbb T^*}(t-p)-1}\bigg]\\
        =&Y_{i,t}(\mathbf Q^{\mathrm{obs}}_{t-m:\mathcal F_{\mathbb T^*}(t-p)-1},q_1\mathbf 1_{t-\mathcal F_{\mathbb T(t-p)+1}},\mathbf Z^{\mathrm{obs}}_{i,t-m:\mathcal F_{\mathbb T^*}(t-p)-1},z\mathbf 1_{t-\mathcal F_{\mathbb T(t-p)+1}})\\
        & -Y_{i,t}(\mathbf Q^{\mathrm{obs}}_{t-m:\mathcal F_{\mathbb T^*}(t-p)-1},q_2\mathbf 1^{t-\mathcal F_{\mathbb T(t-p)+1}},\mathbf Z^{\mathrm{obs}}_{i,t-m:\mathcal F_{\mathbb T^*}(t-p)-1},z\mathbf 1_{t-\mathcal F_{\mathbb T(t-p)+1}});
    \end{align*}
    
    (3) for $p<m$, $p+1\leq t\leq m$, and $\mathcal F_{\mathbb T^*}(t-p)=1$,
    \begin{align*}
        E\left\{Y_{i,t}\frac{\Iitmis{q}{}{1}{p}}{\Pitmis{q}{}{1}{p}}-Y_{i,t}\frac{\Iitmis{q}{}{0}{p}}{\Pitmis{q}{}{0}{p}}\right\}&=\yitmis{q}{}{1}{t}-\yitmis{q}{}{0}{t},\\
        E\left\{Y_{i,t}\frac{\Iitmis{q_1}{z}{1}{p}}{\Pitmis{q_1}{z}{1}{p}}-Y_{i,t}\frac{\Iitmis{q_2}{z}{1}{p}}{\Pitmis{q_2}{z}{1}{p}}\right\}&=\yitmis{q_1}{z}{1}{t}-\yitmis{q_2}{z}{1}{t};
    \end{align*}
    
    (4) for $p<m$, $p+1\leq t\leq m$, and $\mathcal F_{\mathbb T^*}(t-p)>1$,
    \begin{align*}
        E&\bigg[\left\{Y_{i,t}\frac{\Iitmis{q}{}{1}{p}}{\Pitmis{q}{}{1}{p}}-Y_{i,t}\frac{\Iitmis{q}{}{0}{p}}{\Iitmis{q}{}{0}{p}}\right\} \mid \mathbf Q^{\mathrm{obs}}_{1:\mathcal F_{\mathbb T^*}(t-p)-1},
        \mathbf Z^{\mathrm{obs}}_{i,1:\mathcal F_{\mathbb T^*}(t-p)-1}\bigg]\\
        =&Y_{i,t}(\mathbf Q^{\mathrm{obs}}_{1:\mathcal F_{\mathbb T^*}(t-p)-1},q\mathbf 1_{t-\mathcal F_{\mathbb T(t-p)+1}},\mathbf Z^{\mathrm{obs}}_{i,1:\mathcal F_{\mathbb T^*}(t-p)-1},\mathbf 1_{t-\mathcal F_{\mathbb T(t-p)+1}})\\
        &\quad -Y_{i,t}(\mathbf Q^{\mathrm{obs}}_{1:\mathcal F_{\mathbb T^*}(t-p)-1},q\mathbf 1_{t-\mathcal F_{\mathbb T(t-p)+1}},\mathbf Z^{\mathrm{obs}}_{i,1:\mathcal F_{\mathbb T^*}(t-p)-1},\mathbf 0_{t-\mathcal F_{\mathbb T(t-p)+1}}),\\
        E&\bigg[\left\{Y_{i,t}\frac{\Iitmis{q_1}{z}{1}{p}}{\Pitmis{q_1}{z}{1}{p}}-Y_{i,t}\frac{\Iitmis{q_2}{z}{1}{p}}{\Pitmis{q_2}{z}{1}{p}}\right\} \mid \mathbf Q^{\mathrm{obs}}_{1:\mathcal F_{\mathbb T^*}(t-p)-1},\mathbf Z^{\mathrm{obs}}_{i,1:\mathcal F_{\mathbb T^*}(t-p)-1}\bigg]\\
        =&Y_{i,t}(\mathbf Q^{\mathrm{obs}}_{1:\mathcal F_{\mathbb T^*}(t-p)-1},q_1\mathbf 1_{t-\mathcal F_{\mathbb T(t-p)+1}},\mathbf Z^{\mathrm{obs}}_{i,1:\mathcal F_{\mathbb T^*}(t-p)-1},z\mathbf 1_{t-\mathcal F_{\mathbb T(t-p)+1}})\\
        &\quad -Y_{i,t}(\mathbf Q^{\mathrm{obs}}_{1:\mathcal F_{\mathbb T^*}(t-p)-1},q_2\mathbf 1^{t-\mathcal F_{\mathbb T(t-p)+1}},\mathbf Z^{\mathrm{obs}}_{i,1:\mathcal F_{\mathbb T^*}(t-p)-1},z\mathbf 1_{t-\mathcal F_{\mathbb T(t-p)+1}}).
    \end{align*}
\end{lemma}

\begin{proof}[Proof of Lemma~\ref{lemma.misspecified}]
    When $t\geq m+1$, the inequality $\mathcal F_{\mathbb T^*}(t-p)\leq t-m$ indicates that the set of decision points of time $t$ remains unchanged between $p$ and $m$, i.e, $\mathcal F_{\mathbb T^*}^p(t)=\mathcal F_{\mathbb T^*}^m(t)$.
    With probability $\Pitmis{q}{z}{1}{p}\neq 0$, $\Iitmis{q}{z}{1}{p}=1$, and then $Y_{i,t}=\yitmis{q}{z}{1}{m}$. Then we obtain
    \begin{align*}
        E\left\{Y_{i,t}\frac{\Iitmis{q}{z}{1}{p}}{\Pitmis{q}{z}{1}{p}}\right\}=\yitmis{q}{z}{1}{m}.
    \end{align*}

    When $\mathcal F_{\mathbb T^*}(t-p)>t-m$, $\mathcal F_{\mathbb T^*}^m(t)$ differs from $\mathcal F_{\mathbb T^*}^p(t)$. With conditional probability
    $\mathrm{pr}\{\mathbf Q_{t-p:t}=q\mathbf 1_{p+1},\mathbf Z_{i,t-p:t}=z\mathbf 1_{p+1}\mid \mathbf Q_{t-m:\mathcal F_{\mathbb T^*}(t-p)-1}^{\mathrm{obs}},\mathbf Z_{i,t-m:\mathcal F_{\mathbb T^*}(t-p)-1}^{\mathrm{obs}}\}\neq 0$ and conditional on $\mathbf Q_{t-m:\mathcal F_{\mathbb T^*}(t-p)-1}^{\mathrm{obs}}$ and $\mathbf Z_{i,t-m:\mathcal F_{\mathbb T^*}(t-p)-1}^{\mathrm{obs}}$, 
    $$
    Y_{i,t}=Y_{i,t}(\mathbf Q^{\mathrm{obs}}_{t-m:\mathcal F_{\mathbb T^*}(t-p)-1},q\mathbf 1_{t-\mathcal F_{\mathbb T(t-p)+1}},\mathbf Z^{\mathrm{obs}}_{i,t-m:\mathcal F_{\mathbb T^*}(t-p)-1},z\mathbf 1_{t-\mathcal F_{\mathbb T(t-p)+1}}).
    $$
    Then we obtain
    \begin{align*}
        E&\bigg[Y_{i,t}\frac{\Iitmis{q}{z}{1}{p}}{\Pitmis{q}{z}{1}{p}}-Y_{i,t}(\mathbf Q^{\mathrm{obs}}_{t-m:\mathcal F_{\mathbb T^*}(t-p)-1},q\mathbf 1_{t-\mathcal F_{\mathbb T(t-p)+1}},\mathbf Z^{\mathrm{obs}}_{i,t-m:\mathcal F_{\mathbb T^*}(t-p)-1},z\mathbf 1^{t-\mathcal F_{\mathbb T(t-p)+1}}) \mid \\
        &\quad \mathbf Q^{\mathrm{obs}}_{t-m:\mathcal F_{\mathbb T^*}(t-p)-1},\mathbf Z^{\mathrm{obs}}_{i,t-m:\mathcal F_{\mathbb T^*}(t-p)-1}\bigg]=0.
    \end{align*}
    The results for $p+1\leq t\leq m$ can be proved similarly.
\end{proof}

By \Cref{lemma.misspecified}, the expectations of the Horvitz--Thompson estimators are $\taudorder[m]$ and $\tausorder[m]$. Without additional information, these estimators exhibit a non-negligible bias when compared to $\taud$ and $\taus$. The expressions for the variances are more complex, making it challenging to derive conservative variance estimators from the observed data. However, despite the bias present in the estimators, they remain asymptotically normal, as demonstrated in \Cref{theorem.misspecified}. Next, we will prove \Cref{theorem.misspecified}.

\begin{proof}[Proof of Theorem~\ref{theorem.misspecified}]
    We only consider $p<m$, as the proof for the case of $p>m$ is similar. In this case, with the same trick as in the proof of \Cref{theorem.CLT}, we have
    \begin{align*}
        \taudhat-E\{\taudhat\}&=\frac 1{N(T-p)}\sum_{t=p+1}^T\sum_{i=1}^N\mathbf D_{i,t}^{(m)}(q),\\
        \taushat-E\{\taushat\}&=\frac 1{N(T-p)}\sum_{t=p+1}^T\sum_{i=1}^N\mathbf S_{i,t}^{(m)}(z),
    \end{align*}
    where 
    \begin{align*}
        \mathbf D_{i,t}^{(m)}(q)&=\left\{Y_{i,t}\frac{\Iit{q}{}{1}}{\Pit{q}{}{1}}-Y_{i,t}\frac{\Iit{q}{}{0}}{\Pit{q}{}{0}}\right\}-E\left\{Y_{i,t}\frac{\Iit{q}{}{1}}{\Pit{q}{}{1}}-Y_{i,t}\frac{\Iit{q}{}{0}}{\Pit{q}{}{0}}\right\},\\
        \mathbf S_{i,t}^{(m)}(q)&=\left\{Y_{i,t}\frac{\Iit{q_1}{z}{1}}{\Pit{q_1}{z}{1}}-Y_{i,t}\frac{\Iit{q_1}{z}{1}}{\Pit{q_1}{z}{1}}\right\}-E\left\{Y_{i,t}\frac{\Iit{q_1}{z}{1}}{\Pit{q_1}{z}{1}}-Y_{i,t}\frac{\Iit{q_2}{z}{1}}{\Pit{q_2}{z}{1}}\right\}.
    \end{align*}
    Similar to the proof of Theorem~\ref{theorem.CLT}, let $r_n=N(T-p)$, $\tilde X^{d}_{i,t}(q)=\{N(T-p)\}^{-1/2}\mathbf D_{i,t}^{(m)}(q)$ and $\tilde X^{s}_{i,t}(z)=\{N(T-p)\}^{-1/2}\mathbf S_{i,t}^{(m)}(z)$, then, $\{\tilde X_{1,1}^{d},\tilde X_{1,2}^{d},\ldots\}$ and $\{\tilde X_{1,1}^{s},\tilde X_{1,2}^{s},\ldots\}$ are sequences of $\tilde \phi$-independent random variables, respectively, where $\tilde \phi=(\lceil m/p\rceil+1)\bop N_{\max}$. Let $\tilde C_{e_1,e_0}^{d}(q)=\mathrm{var} \{ \sum_{(i,t)=e_0}^{(i,t)=e_0+e_1-1}\tilde X_{i,t}^{d}(q) \}$ and $\tilde C_{e_1,e_0}^{s}(z)= \mathrm{var} \{ \sum_{(i,t)=e_0}^{(i,t)=e_0+e_1-1}\tilde X_{n,i}^{s}(z) \}$, we have
    \begin{align*}
        \tilde C_{e_1,e_0}^{d}(q)&=\frac 1{N(T-p)}\mathrm{var}\left\{\sum_{(i,t)=e_0}^{(i,t)=e_0+e_1-1}\tilde X_{i,t}^{d}(q)\right\}
        \leq \frac 1{N(T-p)}(e_1\tilde M_1^d+2e_1\tilde \phi \tilde M_2^d)=O\left(\frac{e_1\tilde \phi}{NT}\right),\\
        \tilde C_{e_1,e_0}^{s}(z)&=\frac 1{N(T-p)}\mathrm{var}\left\{\sum_{(i,t)=e_0}^{(i,t)=e_0+e_1-1}\tilde X_{i,t}^{s}(z)\right\}\leq \frac 1{N(T-p)}(e_1\tilde M_1^s+2e_1\tilde \phi \tilde M_2^s)=O\left(\frac{e_1\tilde \phi}{NT}\right),
    \end{align*}
    where $\tilde M_1^d$, $\tilde M_2^d$, $\tilde M_1^s$ and $\tilde M_2^s$ are constants associated with the upper bound of variances and covariances. We further have $\tilde C^{d}(q)=\tilde C_{N(T-p),1}^{d}(q)=O(\tilde \phi)$ and $\tilde C^{s}(z)=\tilde C_{N(T-p),1}^{s}(z)=O(\tilde \phi)$.
    We will check the five conditions in Lemma~\ref{lemma.CLT} with $\gamma=0$.

    (1) There exists $\tilde \Delta$ such that $E|\tilde X_{i,t}^{d}(q)|^{2+\delta}\leq \tilde \Delta$ and $E|\tilde X_{i,t}^{s}(z)|^{2+\delta}\leq \tilde \Delta$ for all $i$ and $t$, because all the potential outcomes are bounded. In the inequality, $\tilde \Delta=O\{(NT)^{-1-\delta/2}\}$.

    (2) There exist $\tilde K^d(q)$ and $\tilde K^s(z)$ such that $\tilde C_{e_1,e_0}^{d}(q)/e_1\leq \tilde K^{d}(q)$ and $\tilde C_{e_1,e_0}^{s}(z)/e_1\leq \tilde K^{s}(z)$, for all $e_0$ and $e_1\geq \tilde \phi$. In the inequality, $\tilde K^{d}(q)=O\{\tilde \phi/(NT)\}$ and $\tilde K^{s}(z)=O\{\tilde \phi/(NT)\}$.

    (3) There exist $\tilde L^d(q)$ and $\tilde L^s(z)$ such that $\tilde C^{d}(q)/\{N(T-p)\}\geq \tilde L^{d}(q)$ and $\tilde C^{s}(z)/\{N(T-p)\}\geq \tilde L^{s}(z)$. In the inequality, $\tilde L^{d}(q)=O\{\tilde \phi/(NT)\}$ and $\tilde L^{s}(z)=O\{\tilde \phi/(NT)\}$.

    (4) We set $\tilde {\bar \phi}=4(\lceil m/p\rceil+1)\bop N^{\alpha}T^{\beta}N_{\max}>\tilde \phi$. Then, this condition is satisfied when $\tilde \phi/\tilde {\bar \phi}\rightarrow 0$, which is implied by $N^{-\alpha}T^{-\beta}\rightarrow 0$. 

    (5) This condition is satisfied when $N^{\alpha-1}T^{\beta-1}N_{\max}\rightarrow 0$.

    Therefore, the conclusion holds.
\end{proof}

\section{Additional simulation results}
\label{sec:sm.sim}

\Cref{table.single.sm}, \Cref{fig.mis.sm}, and \Cref{table.multi.sm} present the simulation results used to evaluate the asymptotic properties of the estimators under another typical optimal design, $\mathbb T^*_2$, as discussed in \Cref{corollary.minimaxdesign}. The conclusions drawn are consistent with those under $\mathbb T^*_1$. Additionally, \Cref{fig.compare.sm} and \Cref{fig.order.sm} show results similar to \Cref{fig.compare} and \Cref{fig.order} in the main text, respectively. 

\begin{table}[h]
    \centering
    \caption{Simulation results in single-center randomized experiments under $\mathbb T^*_2$}\label{table.single.sm}
    \begin{threeparttable}
        \begin{tabular}{ccccccccccc}
            \hline
            $N$ & $T$ & $p$ & Estimand & Value & Bias & $\mathrm{var}(\hat \tau)$ & $\widehat{\mathrm{var}}^U(\hat \tau)$ & CP\\
            \hline
            \multirow{4}{*}{10} & \multirow{4}{*}{610} & \multirow{4}{*}{2} & $\taud[q_1]$ & 6 & 0.01 & 2.33 & 2.43 & 0.943\\
            & & & $\taud[q_2]$ & 3 & 0.02 & 1.02 & 1.03 & 0.941\\
            & & & $\taus[1]$ & 6 & 0.00 & 7.96 & 9.01 & 0.966\\
            & & & $\taus[0]$ & 3 & 0.01 & 3.28 & 3.50 & 0.959\\
            \hline
            \multirow{4}{*}{10} & \multirow{4}{*}{610} & \multirow{4}{*}{3} & $\taud[q_1]$ & 6 & -0.08 & 3.57 & 3.58 & 0.938\\
            & & & $\taud[q_2]$ & 3 & 0.04 & 1.49 & 1.53 & 0.959\\
            & & & $\taus[1]$ & 6 & -0.08 & 12.01 & 13.16 & 0.954\\
            & & & $\taus[0]$ & 3 & 0.03 & 4.92 & 5.16 & 0.965\\
            \hline
            \multirow{4}{*}{20} & \multirow{4}{*}{610} & \multirow{4}{*}{2} & $\taud[q_1]$ & 6 & -0.01 & 1.50 & 1.51 & 0.948\\
            & & & $\taud[q_2]$ & 3 & 0.02 & 0.56 & 0.58 & 0.939\\
            & & & $\taus[1]$ & 6 & -0.01 & 7.36 & 8.40 & 0.959\\
            & & & $\taus[0]$ & 3 & 0.02 & 2.85 & 3.17 & 0.963\\
            \hline
            \multirow{4}{*}{20} & \multirow{4}{*}{610} & \multirow{4}{*}{3} & $\taud[q_1]$ & 6 & 0.00 & 2.15 & 2.21 & 0.945\\
            & & & $\taud[q_2]$ & 3 & 0.03 & 0.90 & 0.86 & 0.937\\
            & & & $\taus[1]$ & 6 & -0.07 & 11.15 & 12.26 & 0.961\\
            & & & $\taus[0]$ & 3 & -0.04 & 4.28 & 4.61 & 0.954\\
            \hline
            \multirow{4}{*}{20} & \multirow{4}{*}{1090} & \multirow{4}{*}{2} & $\taud[q_1]$ & 6 & 0.03 & 0.93 & 0.91 & 0.946\\
            & & & $\taud[q_2]$ & 3 & -0.01 & 0.37 & 0.35 & 0.934\\
            & & & $\taus[1]$ & 6 & 0.10 & 5.01 & 5.17 & 0.949\\
            & & & $\taus[0]$ & 3 & 0.06 & 2.00 & 2.06 & 0.949\\
            \hline
            \multirow{4}{*}{20} & \multirow{4}{*}{1090} & \multirow{4}{*}{3} & $\taud[q_1]$ & 6 & 0.02 & 1.41 & 1.32 & 0.938\\
            & & & $\taud[q_2]$ & 3 & 0.01 & 0.53 & 0.54 & 0.946\\
            & & & $\taus[1]$ & 6 & 0.02 & 6.82 & 7.56 & 0.962\\
            & & & $\taus[0]$ & 3 & 0.01 & 2.75 & 3.02 & 0.966\\
            \hline
        \end{tabular}
        \begin{tablenotes}
        \footnotesize
        \item[] Note: Value, true value; CP, coverage probability.
        \end{tablenotes}
    \end{threeparttable}
\end{table}

\begin{figure}[h]
    \centering
    
    \begin{subfigure}[b]{0.4\linewidth}
        \caption{Q-Q plot of $\taudhat[q_1]$}
    \includegraphics[width=\linewidth]{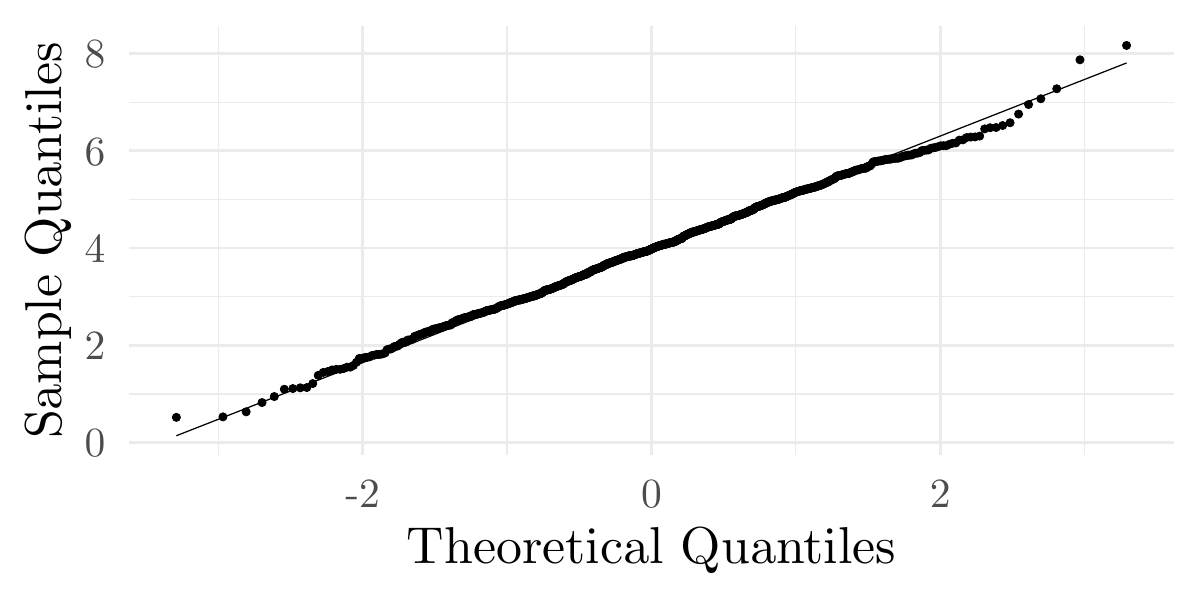}
    \end{subfigure}
    \quad
    \begin{subfigure}[b]{0.4\linewidth}
        \caption{Q-Q plot of $\taudhat[q_2]$}
    \includegraphics[width=\linewidth]{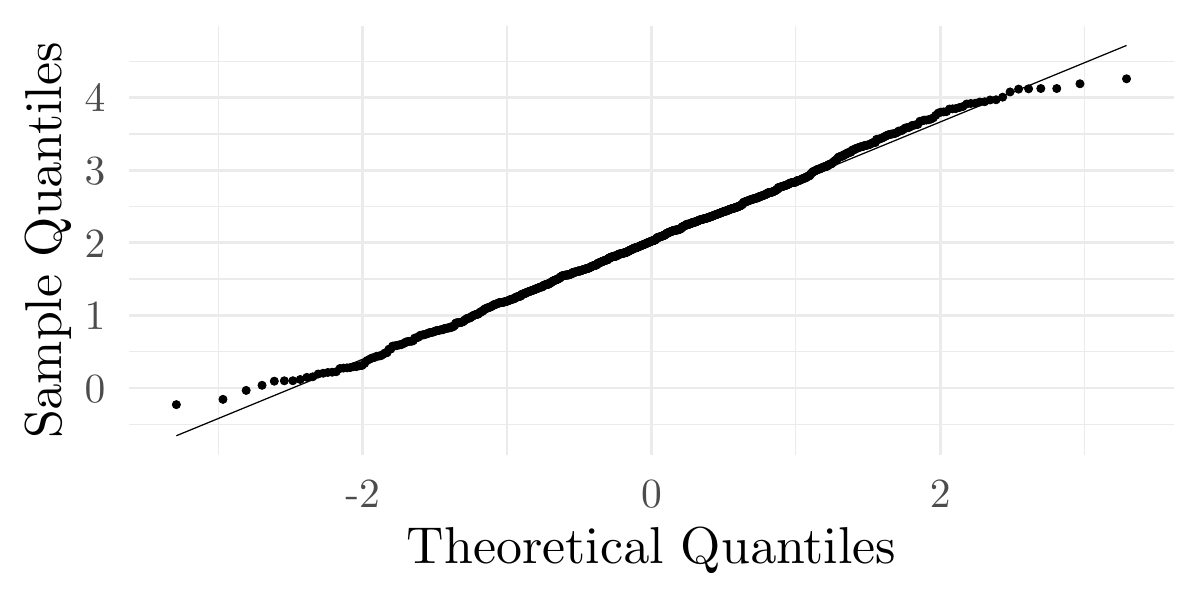}
    \end{subfigure}
    
    \medskip 
    
    \begin{subfigure}[b]{0.4\linewidth}
        \caption{Q-Q plot of $\taushat[1]$}
    \includegraphics[width=\linewidth]{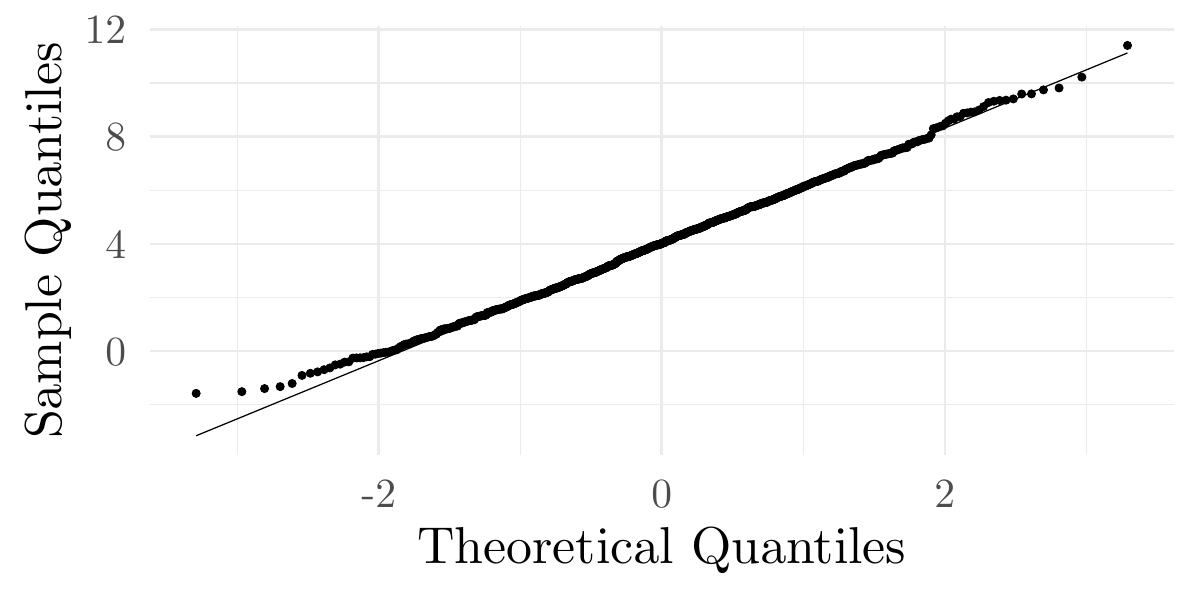}
    \end{subfigure}
    \quad
    \begin{subfigure}[b]{0.4\linewidth}
        \caption{Q-Q plot of $\taushat[0]$}
    \includegraphics[width=\linewidth]{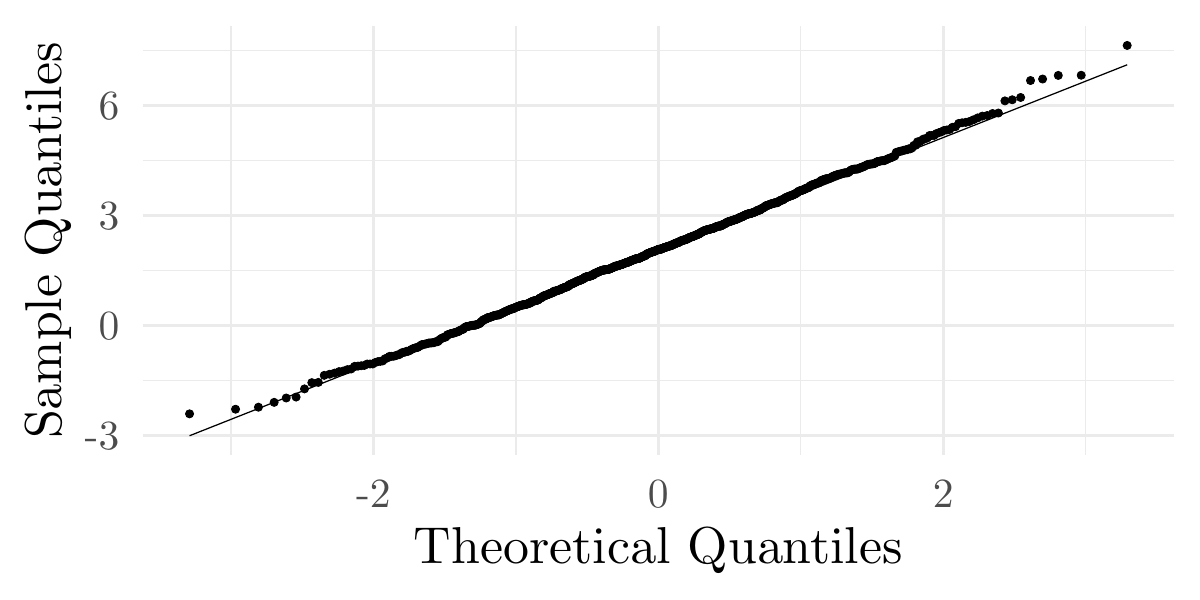}
    \end{subfigure}
    
    \caption{Q-Q plots of the estimators under $\mathbb T^*_2$ when $p=1$, $N=10$ and $T=480$. }\label{fig.mis.sm}
\end{figure}

\begin{table}[h]
    \centering
    \caption{Simulation results in multi-center randomized experiments under $\mathbb T^*_2$}\label{table.multi.sm}
    \begin{threeparttable}
        \begin{tabular}{ccccccccccc}
            \hline
            $N$ & $T$ & $p$ & Estimand & Value & Bias & $\mathrm{var}(\hat \tau)$ & $\widehat{\mathrm{var}}^U(\hat \tau)$ & CP\\
            \hline
            \multirow{4}{*}{$48\times 5$} & \multirow{4}{*}{130} & \multirow{4}{*}{2} & $\taud[q_1]$ & 6 & -0.01 & 0.33 & 0.32 & 0.954\\
            & & & $\taud[q_2]$ & 3 & -0.01 & 0.13 & 0.13 & 0.946\\
            & & & $\taus[1]$ & 6 & -0.04 & 0.76 & 0.75 & 0.951\\
            & & & $\taus[0]$ & 3 & -0.04 & 0.25 & 0.26 & 0.937\\
            \hline
            \multirow{4}{*}{$48\times 5$} & \multirow{4}{*}{130} & \multirow{4}{*}{3} & $\taud[q_1]$ & 6 & 0.01 & 0.50 & 0.48 & 0.944\\
            & & & $\taud[q_2]$ & 3 & -0.01 & 0.20 & 0.19 & 0.948\\
            & & & $\taus[1]$ & 6 & 0.01 & 1.16 & 1.12 & 0.946\\
            & & & $\taus[0]$ & 3 & -0.02 & 0.38 & 0.39 & 0.953\\
            \hline
            \multirow{4}{*}{$48\times 5$} & \multirow{4}{*}{490} & \multirow{4}{*}{2} & $\taud[q_1]$ & 6 & 0.00 & 0.10 & 0.11 & 0.958\\
            & & & $\taud[q_2]$ & 3 & 0.00 & 0.05 & 0.05 & 0.947\\
            & & & $\taus[1]$ & 6 & -0.01 & 0.25 & 0.25 & 0.955\\
            & & & $\taus[0]$ & 3 & -0.01 & 0.10 & 0.10 & 0.947\\
            \hline
            \multirow{4}{*}{$48\times 5$} & \multirow{4}{*}{490} & \multirow{4}{*}{3} & $\taud[q_1]$ & 6 & 0.01 & 0.16 & 0.16 & 0.954\\
            & & & $\taud[q_2]$ & 3 & 0.00 & 0.07 & 0.07 & 0.956\\
            & & & $\taus[1]$ & 6 & -0.01 & 0.35 & 0.37 & 0.965\\
            & & & $\taus[0]$ & 3 & -0.02 & 0.15 & 0.15 & 0.951\\
            \hline
        \end{tabular}
        \begin{tablenotes}
        \footnotesize
        \item[] Note: Value, true value; CP, coverage probability.
        \end{tablenotes}
    \end{threeparttable}
\end{table}

\begin{figure}[h]
    \centering
    
    \begin{subfigure}[b]{0.45\linewidth}
        \caption{$N=10$ and $T$ varies (single-center)}
    \includegraphics[width=\linewidth]{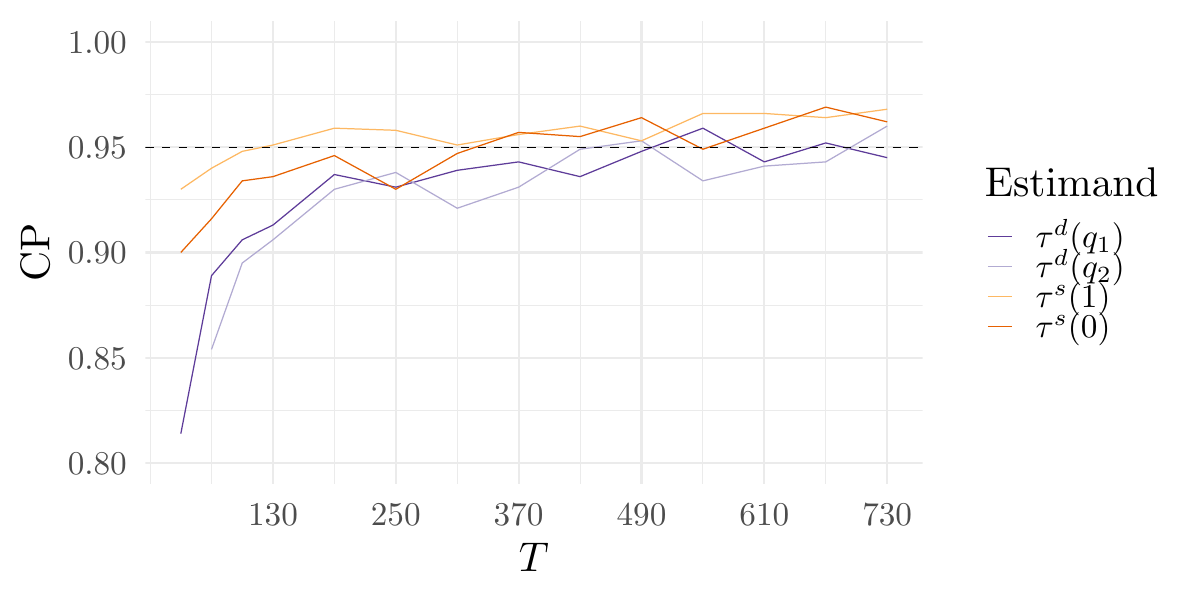}
    \end{subfigure}
    \quad 
    \begin{subfigure}[b]{0.45\linewidth}
        \caption{$T=490$ and $N$ varies (single-center)}
    \includegraphics[width=\linewidth]{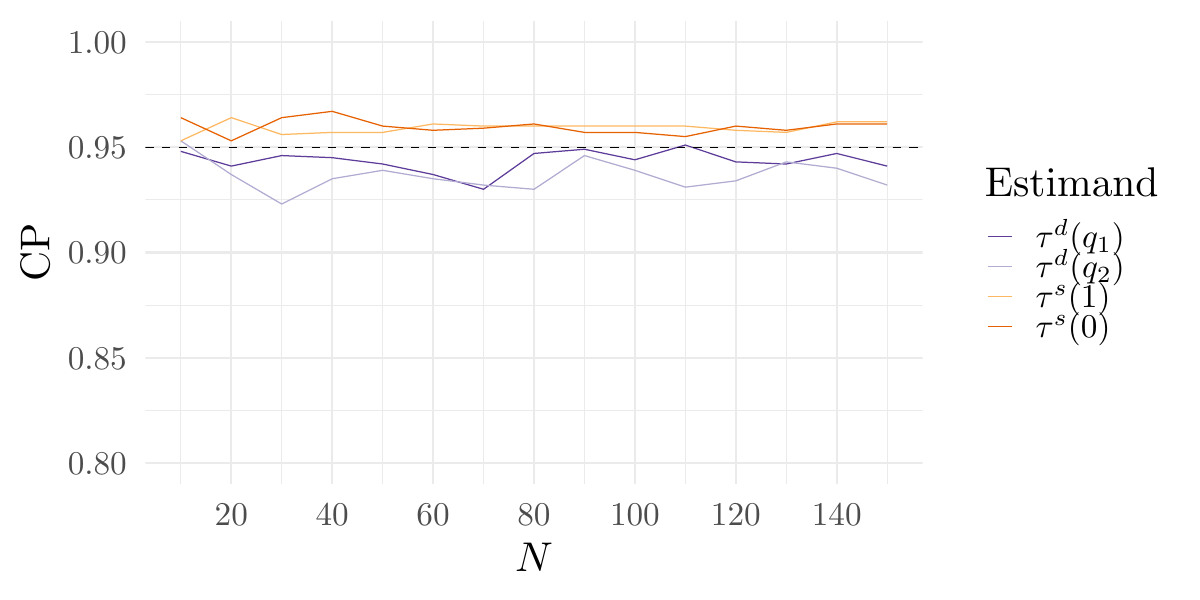}
    \end{subfigure}

    \medskip

    \begin{subfigure}[b]{0.45\linewidth}
        \caption{$N_{[g]}=5$ and $T$ varies (multi-center)}
    \includegraphics[width=\linewidth]{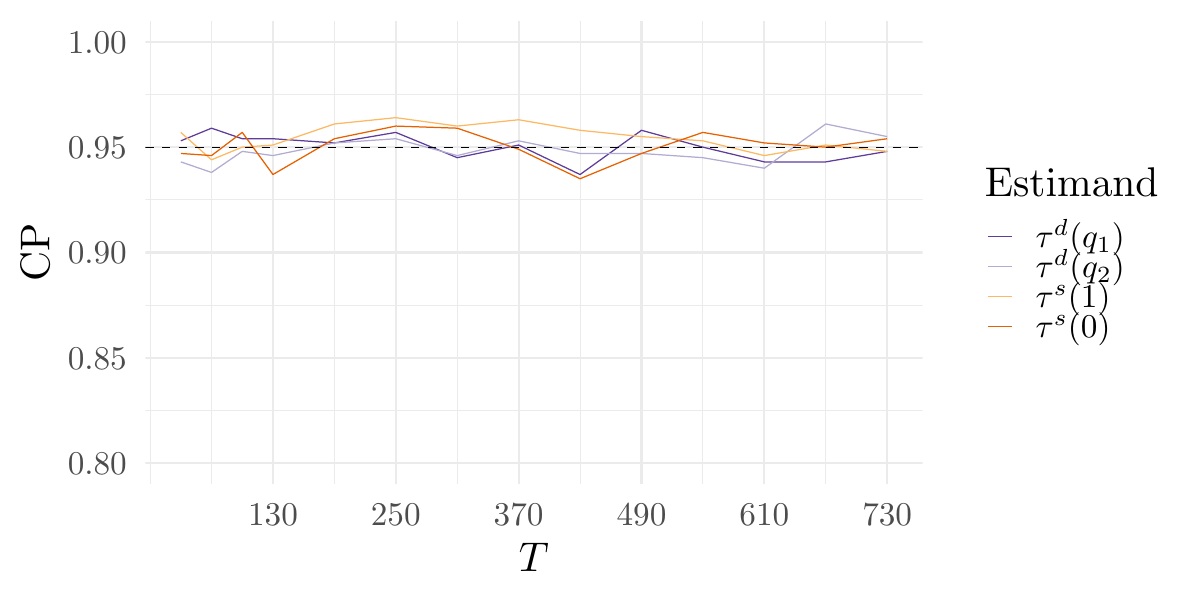}
    \end{subfigure}
    \quad
    \begin{subfigure}[b]{0.45\linewidth}
        \caption{$T=370$ and $N$ varies (multi-center)}
    \includegraphics[width=\linewidth]{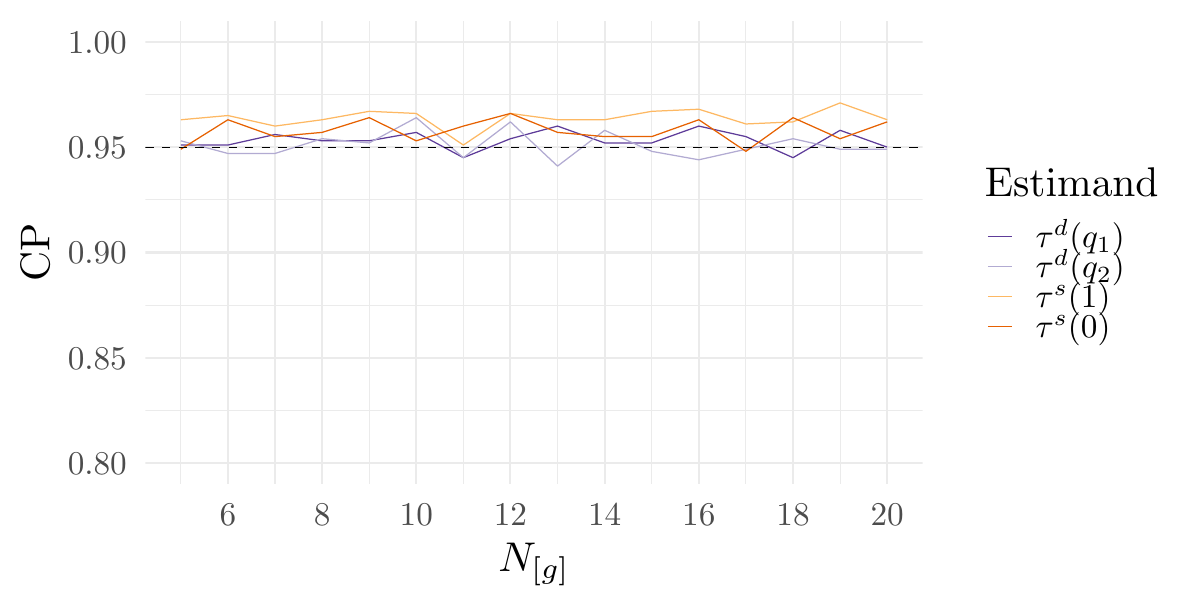}
    \end{subfigure}
    
    \caption{CP under $\mathbb T^*_2$ with different population size when $p=2$. }\label{fig.compare.sm}
\end{figure}

\begin{figure}[h]
    \centering
    
    \begin{subfigure}[b]{0.45\linewidth}
        \caption{$N=480$ and $T$ varies (single-center)}
    \includegraphics[width=\linewidth]{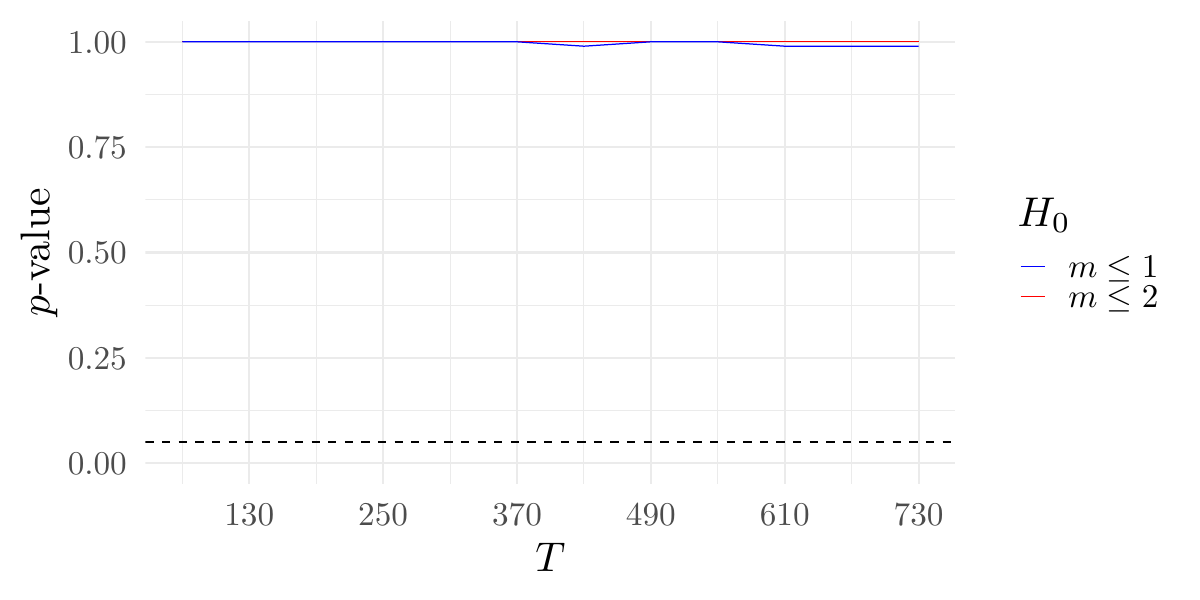}
    \end{subfigure}
    \quad 
    \begin{subfigure}[b]{0.45\linewidth}
        \caption{$T=490$ and $N$ varies (single-center)}
    \includegraphics[width=\linewidth]{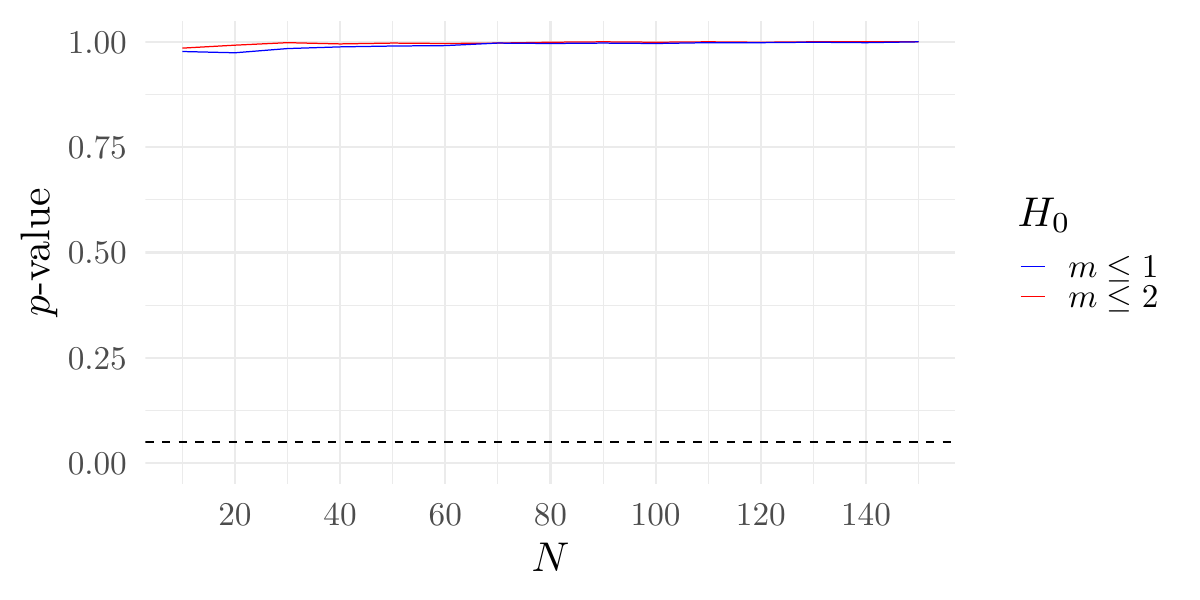}
    \end{subfigure}
    
    \medskip 

    \begin{subfigure}[b]{0.45\linewidth}
        \caption{$N_{[g]}=10$ and $T$ varies (multi-center)}
    \includegraphics[width=\linewidth]{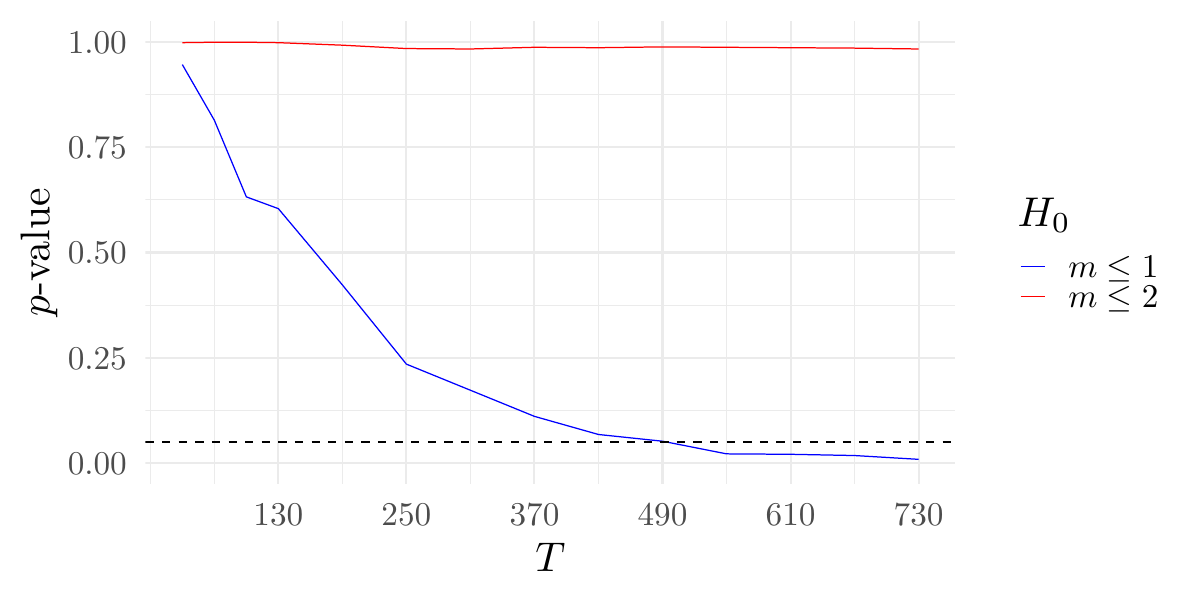}
    \end{subfigure}
    \quad 
    \begin{subfigure}[b]{0.45\linewidth}
        \caption{$T=490$ and $N$ varies (multi-center)}
    \includegraphics[width=\linewidth]{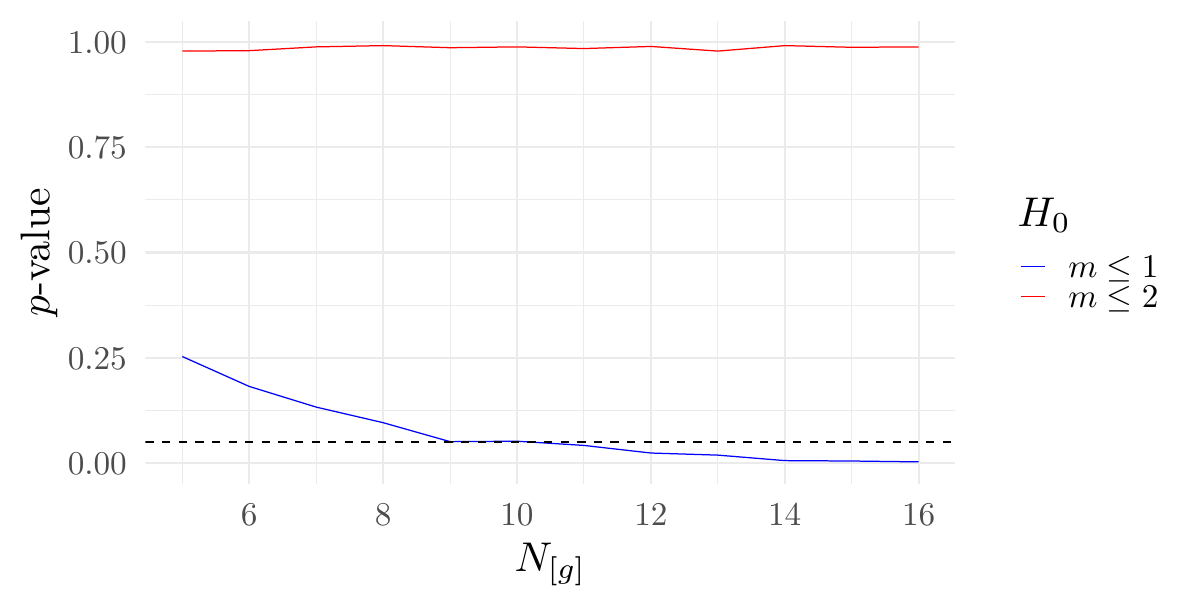}
    \end{subfigure}
    
    \caption{$p$-value from the Wald test for order identification under $\mathbb T^*_2$. }\label{fig.order.sm}
\end{figure}

\clearpage



\end{document}